\documentclass[12pt]{article}
\usepackage{graphicx} 
\usepackage[utf8]{inputenc}
\usepackage{amsmath}
\usepackage{amsthm}
\usepackage{natbib}
\usepackage{color}
\usepackage{ulem}
\usepackage{algorithm}
\usepackage{algpseudocode}
\usepackage{url}

\newtheorem{subprob}{SP}
\newtheorem{thm}{Theorem}
\newtheorem{prop}{Proposition}

\begin{document}
\title{Statistical equilibria of two-dimensional turbulent flows for generic initial vorticity fields on a sphere, calculated on the basis of the original Miller-Robert-Sommeria theory}
\author{K Ryono\footnote{Current affiliation: Research Institute for Applied Mechanics, Kyushu University.}, K Ishioka\\
Graduate School of Science, Kyoto University,\\ Kitashirakawa-Oiwake-cho, Sakyo-ku, Kyoto 606-8502, Japan}
\date{revised: 22 October 2024, accepted: 3 December 2024}

\begin{abstract}
Based on the original Miller-Robert-Sommeria theory, we explicitly compute a statistical equilibrium of two-dimensional turbulent flow on a sphere for a generic initial vorticity field introduced in a previous study. The macroscopic vorticity field corresponding to the obtained statistical equilibrium has a quadrupole structure. The resulting quadrupole structure is topologically consistent with the final state of the long-term time integration of the vorticity equation. However, the statistical equilibrium does not predict the formation of concentrated vortices as seen in the time integration. We also calculate statistical equilibria for the initial vorticity field with a planetary vorticity term, and find a {change of statistical equilibria} from quadrupole states to zonally symmetric states as the angular velocity of the sphere increases. The quadrupole statistical equilibria show nearly linear relations between the macroscopic vorticity and the macroscopic stream function, implying that higher-order Casimir invariants are virtually ineffective even when all Casimir invariants are considered. The discrepancy between the equilibria and the time integration results emphasizes the importance of mixing barriers, which prevent the relaxation of the evolving vorticity field to the statistical equilibria and allow the point-vortex-like dynamics of coherent vortices to persist.
\end{abstract}

\noindent{\it keywords\/}: two-dimensional turbulence, Miller-Robert-Sommeria theory, statistical equilibrium, Casimir invariant, rotating sphere

\maketitle

\textit{This is the Accepted Manuscript version of an article accepted for publication in Fluid Dynamics Research. IOP Publishing Ltd is not responsible for any errors or omissions in this version of the manuscript or any version derived from it. The Version of Record is available online at \url{https://doi.org/10.1088/1873-7005/ad9a75}.} 

\section{Introduction}\label{section_introduction}
What patterns of vorticity or velocity fields arise in the time evolution of two-dimensional turbulence, and why such patterns arise, has long been an open problem in fluid mechanics. One category of theories for pattern formation from two-dimensional turbulence is based on statistical mechanics. The Miller-Robert-Sommeria theory (MRS theory, \citealt{robert1991statistical}, \citealp{miller1990statistical}) is one of the theories in this category. The MRS theory defines the statistical equilibrium state as the state that maximizes the mixing entropy of the vorticity field under the hypothesis of ergodicity about vorticity mixing. Soon after the proposal of the theory, it was shown that statistical equilibria were consistent with the patterns in the vorticity field resulting from long-term time evolutions in cases where the initial vorticity fields consisted of one or two vortex strips in a channel geometry (\citealp{sommeria1991final}, \citealp{thess1994inertial}). For a spherical domain, \cite{herbert2012statistical} and \cite{herbert2013additional} considered MRS-like theory simplified by using generalized entropy \citep{bouchet2008simpler} and showed that the energy condenses into the total wavenumber two modes of spherical harmonics expansion of the macroscopic vorticity field. The MRS theory for a spherical domain was also studied by \cite{qi2014hyperviscosity}. They compared the MRS theory that considered only second-order Casimir invariants (MRS-2) with the MRS theory that partially considered Casimir invariants up to the fourth order (perturbative MRS-4), and found that the perturbative MRS-4 predicted more intense coherent vortices as the statistical equilibria, which was more consistent with the result of long-term time evolution than the prediction by the MRS-2.

Recently, new time integration methods have been employed to perform longer time integrations for two-dimensional Euler flows, providing a clearer picture to solve this problem. \cite{dritschel2015late} used the Combined Lagrangian Advecting Method (CLAM) and the Geodesic Grid Method (GGM) and showed that a randomly generated initial vorticity field evolved to a long-term stable state with four coherent vortices, which exhibited quasi-periodic motion. By considering a finite-dimensional approximation of the geometric structure of the Euler equation, \cite{modin2020casimir} devised a numerical integration method that conserves many Casimir invariants, the number of which is proportional to the discretization resolution. They used this method to reexamine the time integration carried out by \cite{dritschel2015late}, and confirmed the robustness of the four-vortex regime in the final state. In addition, \cite{modin2020casimir} performed time integrations from initial vorticity fields with non-zero angular momentum. It was found that with increasing angular momentum, the number of coherent vortices in the final state tends to gradually decrease from four, to three, to two. In the above and similar studies, a two-dimensional sphere is often used as the flow domain. This choice is based on the following reasons. (1) A sphere is the simplest closed surface, and there is no need to consider the effects of boundary conditions. (2) Two-dimensional turbulence on a rotating sphere is sometimes regarded as a simple model of the atmospheric flows of the Earth and other planets, and therefore it is a subject of interest in geophysical fluid dynamics. (3) Because of the geophysical interest, numerical methods for flows on a spherical surface are well developed. (4) There are conserved quantities (angular momentum) associated with rotational symmetries on a sphere, and the relationship between these conserved quantities and the pattern structure of the flow is of interest. However, the final state obtained by time integration of the vorticity field in \cite{dritschel2015late} and \cite{modin2020casimir} did not satisfy the relation between the vorticity and the stream function that statistical equilibrium states should satisfy in the MRS theory. On the basis of this result, they concluded that there was an inconsistency between the statistical equilibrium state and the final state of the time evolution. Before \cite{dritschel2015late} and \cite{modin2020casimir}, the existence of such an inconsistency had already been discussed in the case of some special initial vorticity fields on a flat torus. \cite{segre1998late} showed that statistical equilibria were not necessarily reached in the time evolution of the vorticity field from an initial vorticity field consisting of two values of vorticity patches, casting doubt on correspondence between statistical equilibria and patterns in time integration results. \cite{morita2011collective} showed that oscillating states emerged from two-dimensional turbulence that did not reach stationary or translational states. While such previous studies had existed, \cite{dritschel2015late} was new in that it showed the existence of such an inconsistency even when the initial vorticity field was randomly generated, which \cite{modin2020casimir} termed ``generic''.  

{How does the prediction of the MRS theory differ from the coherent states that appear in the time integrations? One possibility is that the MRS theory successfully predicts the number and shapes of the coherent vortices, but fails to predict the asymmetric distribution of vorticity among the four vortices. In fact, \cite{qi2014hyperviscosity} also performed a time integration similar to that of \cite{dritschel2015late}, and noted that the resulting state of four coherent vortices shows two-branched relation between the vorticity and the stream function, where each branch corresponds to a pair of positive and negative vortices. Each branch appears as a ``sinh-type'' curve. The sinh-type functional relation of the vorticity and the stream function is typically seen in the statistical equilibrium in the MRS theory, representing a pair of anti-signed two vortices with concentrated cores. Such two-branched relations are repeatedly observed in \cite{dritschel2015late} and \cite{modin2020casimir}. Hence, one can imagine a statistical equilibrium consisting of two stationary pairs of anti-signed two concentrated vortices, with vorticity equally distributed to each pair, where the sinh-type curves corresponding to each pair match. If this assumption is true, then the quasi-periodic states observed after long time evolution can be considered as a perturbed state of the statistical equilibrium. There is another possibility that the MRS theory totally fails, and the statistical equilibrium is quite different from the observed coherent states.} 

The previous studies cited above, however, did not calculate the statistical equilibria explicitly, nor did they compare the statistical equilibria with the final states of the time evolution. One reason for this seems to be that the calculation of the statistical equilibrium based on the original MRS theory is very difficult, especially when the initial vorticity field consists of many levels of vorticity patches.  
Recently, a new numerical method to compute the statistical equilibria strictly based on the MRS theory was proposed by \cite{ryono2022new}, where the spherical harmonic expansion coefficients of the macroscopic vorticity field were treated as variables in the entropy maximization problem. They then proposed a method to search for the maximum mixing entropy using the gradient method on a quadratic surface of constant energy. By using the proposed method, they computed the statistical equilibrium state for an initial vorticity field consisting of 32 levels of vorticity patches, which imitated an analytically defined zonal initial vorticity field, and showed that the final state of the time evolution and the statistical equilibrium had similar vorticity field structures.

In this study, we compute the statistical equilibrium state for the initial vorticity field introduced by \cite{dritschel2015late} and investigate the correspondence between the statistical equilibrium and the final state of the long-term time evolution of the vorticity field. In the calculation of the statistical equilibrium, we use a newly proposed method to approximate continuous initial vorticity fields with a finite number of vorticity patches. In addition, we examine the effect of introducing a non-zero planetary vorticity term to the initial vorticity field, with focusing on how the planetary vorticity affects the number of coherent vortices in the statistical equilibrium to compare with the results of \cite{modin2020casimir}. To the best of the author's knowledge, the statistical equilibria for such generic initial vorticity fields are computed for the first time. The paper is organized as follows. In section \ref{section_theory}, a brief description of the MRS theory is given. Section \ref{section_method} describes the calculation method of \cite{ryono2022new} and a new method to approximate the initial vorticity field by vorticity patches. Section \ref{section_initialsetting} describes the definition of the initial vorticity fields, the time integration results from them, and the settings of the statistical equilibrium computation. In section \ref{section_results}, the statistical equilibria are computed and compared with the final state of time evolutions of the vorticity field. Finally, a discussion and conclusions are given in section \ref{section_discussion}.

\section{MRS theory on a sphere}\label{section_theory}
In this section, we describe a basic concept of the MRS theory on a sphere. For details of the theory, see \cite{miller1990statistical} and \cite{robert1991statistical}, as well as a review by \cite{bouchet2012statistical}, which provides an extensive history of statistical theories of two-dimensional turbulence, including the two-decade evolution of the MRS theory.

We consider a sphere \(S\) of radius \(1\), and {a} two-dimensional incompressible inviscid flow that is subject to the Euler equation:
\begin{align*}
    \frac{\partial q}{\partial t} + J(\psi, q) = 0.
\end{align*}
Here, \(q\) represents the vorticity and \(\psi\) represents the stream function, where they are related by the equation \(\Delta \psi = q\). The operator \(J(\cdot,\cdot)\) denotes Jacobian operator. In this paper, we use the longitude \(\lambda\) and the sine-latitude \(\mu\) as the coordinates of the sphere \(S\). By using these coordinates, the Jacobian operator can be written as,
\begin{align*}
    J(f, g) = \frac{\partial f}{\partial \lambda}\frac{\partial g}{\partial \mu}-\frac{\partial g}{\partial \lambda}\frac{\partial f}{\partial \mu},
\end{align*}
and the Laplacian operator \(\Delta\) is written as,
\begin{align*}
    \Delta = \frac{1}{1-\mu^2} \frac{\partial^2}{\partial \lambda^2} + \frac{\partial}{\partial \mu} \left\{(1-\mu^2)\frac{\partial}{\partial \mu}\right\}.
\end{align*}
On the sphere, the area element is \({\rm d} S = {\rm d}\lambda {\rm d}\mu\). Since the fluid is incompressible, it is obvious that the Casimir invariants: 
\begin{align*}
    C_f = \frac{1}{4\pi} \int_S f(q) {\rm d} S,
\end{align*}
are conserved for any function \(f\). Note that, in this system the three components of angular momentum are conserved due to the rotational symmetry of the domain, which is the main difference from the system on a flat torus.

Assuming that the initial vorticity field consists of a finite number, \(K\), of vorticity patches. Let \(Q_k\) and \(S_k\,(k=1,\cdots, K)\) represent the vorticity value and the area of each vorticity patch, respectively. The integral of vorticity over \(S\) must be zero, and the sum of the areas of all vorticity patches must equal to the area of \(S\), i.e., 
\begin{align*}
    \sum_{k=1}^K Q_k S_k = 0,\qquad \sum_{k=1}^K S_k = 4\pi.
\end{align*}
As the vorticity field evolves from this initial state and the flow becomes highly turbulent, the vorticity patches are deformed to have complex boundaries and filamentous shapes. Thus, the vorticity value at each point of \(S\) can be viewed as a probabilistic variable. Let \(r_k(x)\) denote the probability of observing the \(k\)--th vorticity patch at a point \(x\in S\). From the definition of probability, 
\begin{align}
    \sum_{k=1}^K r_k(x) = 1\label{sum_be_unity}
\end{align}
stands. In addition, the area of \(k\)--th patch is conserved due to the incompressibility of the flow:
\begin{align}
    \int_S r_k(x) {\rm d} S = S_k.\label{area_preservation}
\end{align}
In the Euler system, the vorticity field evolves with conservation of the energy and the angular momentum \(M=(M_1,M_2,M_3)\), therefore, 
\begin{align}
    E &= -\frac{1}{4\pi} \int_S \frac{1}{2}\overline{\psi}\overline{q} {\rm d} S=E^{\rm ini} ,\\
    M_1 &= \frac{1}{4\pi} \int_S \overline{q} \mu {\rm d} S=M_1^{\rm ini},\label{AM_constraint1}\\
    M_2 &= \frac{1}{4\pi} \int_S \overline{q} \sqrt{1-\mu^2} \cos \lambda {\rm d} S=M_2^{\rm ini},\label{AM_constraint2}\\
    M_3 &= \frac{1}{4\pi} \int_S \overline{q} \sqrt{1-\mu^2} \sin \lambda {\rm d} S=M_3^{\rm ini}\label{AM_constraint3} 
\end{align}
must be fulfilled. Here, 
\begin{align}
    \overline{q}(x) = \sum_{k=1}^K Q_k r_k(x)\label{qbar_def}
\end{align}
is the local expected value of the vorticity \(q\), which is called the macroscopic vorticity, and \(\overline{\psi}\) is the macroscopic stream function defined by \(\overline{q}=\Delta\overline{\psi}\). Here, \(E^{\rm ini}\) and \(M^{\rm ini}=(M_1^{\rm ini}, M_2^{\rm ini}, M_3^{\rm ini})\) are the values of the energy and the three components of angular momentum, respectively, that the initial vorticity field has. The statistical equilibrium is defined as the state which maximizes the mixing entropy:
\begin{align*}
    S_{\rm mix} := -\frac{1}{4\pi} \int_S r_k(x) \log r_k(x) {\rm d} S,
\end{align*}
under the constraints \eqref{sum_be_unity}--\eqref{AM_constraint3}.

\section{Computational methods}\label{section_method}
In this section, after briefly describing the numerical method for computing statistical equilibria introduced by \cite{ryono2022new}, we propose a new technique to determine a finite set of vorticity patches that approximate continuous initial vorticity fields. 

\subsection{Discretization of field variables}\label{subsec_discretize}
We take \(I\times J\) grid points \((\lambda_i, \mu_j)\, (i=1,\cdots, I; j=1,\cdots, J)\). Here, \(\lambda_i = 2\pi i/I\) and \(\mu_j\) is the \(j\)-th Gaussian node, which is the \(j\)-th zero of the Legendre polynomial \(P_J(\mu)\). Each \(r_k (\lambda,\mu)\) is represented by its values at the grid points: \(r_{ijk} = r_k(\lambda_i, \mu_j)\). Furthermore, we consider the expansion of the macroscopic vorticity field \(\overline{q}\) by spherical harmonics \(Y_{m,n}(\lambda,\mu)\) with the truncation wavenumber \(N\):
\begin{align}
    \overline{q}(\lambda_i, \mu_j) = \sum_{k=1}^K Q_k r_{ijk} = \sum_{n=1}^N \sum_{m=-n}^n (\hat{\xi}_{m,n}+\sqrt{-1}\hat{\eta}_{m,n}) Y_{m,n}(\lambda_i, \mu_j).\label{expansion}
\end{align}
Here, note that the spherical harmonics \(Y_{m,n}\) are normalized so that
\begin{align*}
    \frac{1}{4\pi}\int_{-1}^1\int_{0}^{2\pi} |Y_{m,n}(\lambda,\mu)|^2 d\lambda d\mu = 1.
\end{align*} For the definitions of the spherical harmonics, the Gaussian nodes and the Gaussian weights, see Appendix A and B of \cite{ryono2022new}. {We choose the value of \(I,J\) and \(N\) so that \(I\geq 3N+1\) and \(I=2J\), according to the conventional setting of the spectral method.} The expansion coefficients \(\hat{\xi}_{m,n}\in \mathbf{R}\) and \(\hat{\eta}_{m,n} \in \mathbf{R}\) must fulfill \(\hat{\xi}_{-m,n}=\hat{\xi}_{m,n}\) and \(\hat{\eta}_{-m,n}=-\hat{\eta}_{m,n}\). Therefore, the number of independent real coefficients is \((N+1)^2-1\). Furthermore, imposing the angular momentum constraints \eqref{AM_constraint1}, \eqref{AM_constraint2}, and \eqref{AM_constraint3} is equivalent to fixing \(\hat{\xi}_{0,1}, \hat{\xi}_{1,1}\) and {\(\hat{\eta}_{1,1}\)}. We denote a vector of the independent \((N+1)^2-4\) real coefficients by \(Z=(\hat{\xi}_{0,2}, \cdots, \hat{\xi}_{0,N}, \hat{\xi}_{1,2}, \hat{\eta}_{1,2}, \cdots, \hat{\xi}_{N,N}, \hat{\eta}_{N,N})\). Then, the energy constraint is written by a quadratic form of \(Z\):
\begin{align}
    \sum_{n=2}^N \frac{\hat{\xi}_{0,n}^2}{2n(n+1)} + \sum_{n=2}^N \sum_{m=1}^n \frac{\hat{\xi}_{m,n}^2+\hat{\eta}_{m,n}^2}{n(n+1)}=E_0.\label{energy_surface}
\end{align}
Note that \(E_0\) is the initial energy value minus the contribution of \(n=1\) components of the initial vorticity field. That is, 
\begin{align*}
    E_{0} = E^{\rm ini} - \frac{3}{4}\{{(M^{\rm ini}_1)}^2+{(M^{\rm ini}_2)}^2 +{(M^{\rm ini}_3)}^2\}.
\end{align*}

\subsection{Definition of a subproblem and change of the variable}\label{subsec_subprob}
We introduce a subproblem in the following way so that we can use \(Z\), which has \((N+1)^2-4\) components, as the variable of the entropy maximization problem instead of \((r_{ijk})_{1\leq i\leq I,1 \leq j\leq J, 1\leq k\leq K}\), which has \(IJK\) components. The subproblem is a problem to find \((r_{ijk})\in \mathbf{R}^{IJK}\) which maximizes \(S_{\rm mix}\) for a given \(Z\). The formulation of the subproblem can be written as:
\begin{subprob}\label{subproblem}
    On the set of \(r_{ijk}\geq 0\) which fulfills linear constraints 
        \begin{align}
        &\sum_{k=1}^K r_{ijk} = 1\qquad (i=1,\cdots, I; j=1,\cdots, J),\label{sum_be_unity_discrete}\\
        &\sum_{i,j} w_{j} r_{ijk} = \frac{S_{k}}{2\pi}\label{area_preservation_discrete}\qquad (k=1,\cdots, K)
    \end{align}
     and \eqref{expansion}, find an \((r_{ijk})\in \mathbf{R}^{IJK}\) which maximizes
    \begin{align*}
        S_{\rm mix} = - \frac{1}{2}\sum_{i,j,k}w_j r_{ijk}\log r_{ijk}.
    \end{align*}
\end{subprob}
\noindent Note that, equations \eqref{sum_be_unity_discrete} and \eqref{area_preservation_discrete} are the discretized forms of \eqref{sum_be_unity} and \eqref{area_preservation}. The Gaussian weights \(w_j\) are normalized so that their sum is equal to \(2/I\).

As explained in detail in \cite{ryono2022new} and its supplementary material, by introducing the subproblem SP 1 the entropy maximization problem becomes two-folded. Solving SP 1 and obtaining the maximizer \((r_{ijk})\) for a given \(Z\), we can consider \((r_{ijk})\) as a function of \(Z\). We then use \(Z\) as the variable and {maximize the mixing entropy as a function of \(Z\)}:
\begin{align*}
    S_{\rm mix}(r_{ijk}(Z)) = -\frac{1}{2}\sum_{i,j,k} w_k r_{ijk}(Z) \log r_{ijk}(Z).
\end{align*}

\noindent We can search for the statistical equilibrium by using a gradient method on the energy surface (e.g., the projection gradient method of \citealp{tanabe1974algorithm}), given an initial search point that belongs to the intersection of the surface {specified by \eqref{energy_surface}} and the interior of the feasible region \(P\). Here, the feasible region \(P\) is defined as the set of such \(Z\) that {equations \eqref{expansion}, \eqref{sum_be_unity_discrete}, and \eqref{area_preservation_discrete}} have a solution with \(r_{ijk} \geq 0\) for all \(i,j\) and \(k\).

\subsection{Approximation of the initial vorticity field by vorticity patches and determination of the initial search point}\label{subsec_patchgenerate}
We need a set of vorticity patches that approximate the initial vorticity field for the numerical computation of the statistical equilibrium. We also need an initial search point which satisfies the energy constraint and belongs to the feasible region in order to apply the method of \cite{ryono2022new}. \cite{ryono2022new} used some geometric tricks to generate initial search points, but we propose here a new method to determine a set of vorticity patches that approximate a given initial vorticity field. For the determined vorticity patches, the point \(Z\) corresponding to the initial vorticity field is guaranteed to belong to the interior of the feasible region \(P\).

In the following, we consider an approximation of a initial vorticity field \(q_{\rm ini} (\lambda,\mu)\) by \(K\) vorticity patches. Suppose that \(q_{\rm ini}\) is given in the expansion form:
\begin{align*}
    q_{\rm ini}(\lambda, \mu)  = \sum_{n=1}^N \sum_{m=-n}^n (\hat{\xi}_{m,n}^{\rm ini}+\sqrt{-1}\hat{\eta}_{m,n}^{\rm ini}) Y_{m,n}(\lambda, \mu),
\end{align*}
and let 
\begin{align*}
    Z_{\rm ini} = (\hat{\xi}_{0,2}^{\rm ini}, \cdots, \hat{\xi}_{0,N}^{\rm ini}, \hat{\xi}_{1,2}^{\rm ini}, \hat{\eta}_{1,2}^{\rm ini}, \cdots, \hat{\xi}_{N,N}^{\rm ini}, \hat{\eta}_{N,N}^{\rm ini})\in \mathbf{R}^{(N+1)^2-4}
\end{align*}
be the point \(Z\) corresponding to \(q_{\rm ini}\). Let
\begin{align*}
    q_{ij} = q_{\rm ini}(\lambda_i, \mu_j)
\end{align*}
be the values of \(q_{\rm ini}\) at the computational grids defined in subsection \ref{subsec_discretize}, and let \(q_{\rm min}\) and \(q_{\rm max}\) be the minimum and maximum values of \(q_{ij}\), respectively. We partition the interval \([q_{\rm min}, q_{\rm max}]\) into \(K-1\) segments of equal length and let \(Q_k\) be the \(k\)-th node of the partition. For the case of \(k=1\) and \(k=K\), however, the value of \(Q_k\) is modified by a small parameter \(\alpha\), i.e., 
\begin{align*}
    Q_1 &= q_{\rm min}-\alpha, \\
    Q_k &= q_{\rm min} + \frac{q_{\rm max}-q_{\rm min}}{K-1}(k-1) \quad (k=2,\cdots, K-1), \\
    Q_K &= q_{\rm max}+\alpha.
\end{align*}
The modification of the values of \(Q_1\) and \(Q_K\) is introduced to ensure that the point \(Z\) corresponding to the initial vorticity field \(q_{ij}\) is represented as an interior point of the feasible region \(P\). Indeed, for a point \(Z\) to belong to the interior of \(P\), it is necessary and sufficient that there exists a set of positive \(r_{ijk}\) which fulfills \eqref{sum_be_unity_discrete}, \eqref{area_preservation_discrete}, and \eqref{expansion}. If we simply use \(Q_K=q_{\rm max}\), then there cannot be such a set of \(r_{ijk}\), because
\begin{align*}
    q_{ij} = \sum_{k=1}^K Q_k r_{ijk} \quad {\rm and} \quad r_{ijk}\geq 0
\end{align*}
means \(r_{ijK}= 1\) and \(r_{ijk}=0\,(k=1,\cdots,K-1)\), for such \((i,j)\) that \(q_{ij}=q_{\rm max}\). We use \(K\) vorticity patches, each with one of the \(Q_k\, (k=1,\cdots,K)\) values defined above. The areas of vorticity patches, \(S_k\), are determined as 
\begin{align*}
    S_{k} = 2\pi\sum_{i,j} w_j \tilde{r}_{ijk},
\end{align*}
where \(\tilde{r}_{ijk}\) represents the fraction that the \(k\)--th vorticity patch occupies the neighborhood of the grid \((\lambda_i,\mu_j)\). These fractions \(\tilde{r}_{ijk}\) are determined as follows. For fixed \(i\) and \(j\), we take an integer \(k_0\) such that \(Q_{k_0}< q_{ij}\leq Q_{k_0+1}\). Then, for \(k\neq k_0, k_0+1\), we let \(\tilde{r}_{ijk} = w_j^{-1} \varepsilon\), where \(\varepsilon >0\) is a small parameter. For \(k=k_0\) and \(k_0+1\), we determine \(\tilde{r}_{ijk}\) so that the following equations are satisfied:
\begin{align*}
   & \tilde{r}_{ij,k_0} + \tilde{r}_{ij,k_0+1} = 1 - \sum_{k\neq k_0, k_0+1} \tilde{r}_{ijk},\\
   & Q_{k_0}\tilde{r}_{ij,k_0} + Q_{k_0+1}\tilde{r}_{ij,k_0+1} = q_{ij} - \sum_{k\neq k_0,k_0+1} Q_k\tilde{r}_{ijk}.
\end{align*}
In other words, we give only small area \(2\pi\varepsilon\) to each vorticity patch whose vorticity value fulfills \(|Q_k-q_{ij}|>\Delta Q\) {, where \(\Delta Q = (q_{\rm max}-q_{\rm min})/(K-1)\)}, and then we determine \(\tilde{r}_{ijk}\) for \(k=k_0\) and \(k_0+1\) so that the equations \eqref{sum_be_unity_discrete} and \eqref{expansion} are satisfied. Note that the parameters \(\varepsilon >0\) and \(\alpha>0\) must be small enough so that all \(\tilde{r}_{ijk}\) are positive. Determining \(S_k\) as described above ensures that the point \(Z_{\rm ini}\in \mathbf{R}^{(N+1)^2-4}\) corresponding to the initial vorticity field is an interior point of the feasible region \(P\). Indeed, the set of \(\tilde{r}_{ijk}\) fulfills \(\tilde{r}_{ijk}>0\) because \(\varepsilon>0\), and the equations \eqref{sum_be_unity_discrete}, \eqref{area_preservation_discrete}, and \eqref{expansion} from definition. 

We note that errors of Casimir invariants arise by the approximation of the initial vorticity field \(q_{ij}\) by a finite set of vorticity patches. However, by choosing small values of \(\varepsilon\) and \(\alpha\), and setting \(K\) to a large value, the error can be limited to a small value. Indeed, for a Lipschitz function \(f\), the error {is} bounded as:
\begin{align}
    |C_{f, \rm grid}-C_{f, \rm patch}| \leq IJ(K-2)M\varepsilon + c (\Delta Q +\alpha), \label{casimir_error}
\end{align}
where \(C_{f, \rm grid}\) is the value of Casimir invariant for \(f\) computed by using the grid values of the initial vorticity field, \(q_{ij}\), whilst \(C_{f, \rm patch}\) is that computed on the basis of the patch approximation of the initial vorticity field. On the right-hand side of \eqref{casimir_error}, \(M\) is the maximum value of \(f(q)\) on the sphere, \(c>0\) is the Lipschitz constant for \(f\). The proof of the inequality \eqref{casimir_error} is given in {\ref{sec_appendix_casimir_error}}. {The whole procedure of the computation of statistical equilibria is shown in Algorithm \ref{alg_stat_equi}.}
\begin{algorithm}[htbp]

\caption{Pesudocode of the computation of statistical equilibria}\label{alg_stat_equi}
\begin{algorithmic}[1]
\State input :: $I, J, K, q_{ij}, \alpha, \varepsilon$
\State input :: $\eta$ (step size of the gradient method)
\State output :: $Q_k, S_k$
\State output :: $Z_{\rm eqb}$ ($Z$ corresponding to the statistical equilibrium)
\Statex
\State $Q_1 \gets \min q_{ij}-\alpha$ \Comment{Determination of \(Q_k\)}
\State $Q_K \gets \max q_{ij}+\alpha$
\State $\Delta Q \gets (\max q_{ij}-\min q_{ij})/(K-1)$
\ForAll{$k\gets 2,K-1 $}
    \State $Q_k \gets \min q_{ij} + (k-1)\Delta Q $
\EndFor
\Statex
\ForAll {$i\gets 1, I$}\Comment{Determination of \(S_k\)}
    \ForAll{$ j\gets 1,J $}
        \State ${\rm Determine}\,\,k_0 $ s.t. $ Q_{k_0}<q_{ij} \leq Q_{k_0+1}$
        \ForAll{$ k\gets 1,K $}
        \If {$k \neq k_0$ .and. $k \neq k_0+1 $}
            \State $\tilde{r}_{ijk} \gets \varepsilon/ w_j $
        \EndIf
        \EndFor
        \State Determine $\tilde{r}_{ij{k_0}}$ and $\tilde{r}_{ij{k_0+1}}$
        \ForAll{$ k\gets 1, K $}
            \State $S_k \gets S_k + w_j \tilde{r}_{ijk}$
        \EndFor
    \EndFor
\EndFor
\Statex
\State Determine \(Z\) from \(q_{ij}\) by FFT
\While{(.true.)}\Comment{maximization of the mixing entrooy by the gradient method}
    \State Solve SP1 to get $ r_{ijk}$ from $Z$
    \State Compute $\nabla S_{\rm mix}(Z)$   
    \State Compute $(P\nabla S_{\rm mix})(Z)$: the projection of $\nabla S_{\rm mix}(Z)$ to the energy surface
    \State $Z \gets Z+ \eta (P\nabla S_{\rm mix})(Z)$
    \If{$\|(P\nabla S_{\rm mix})(Z)\|< criterion $}
        \State exit
    \EndIf
\EndWhile
\Statex
\State $Z_{\rm eqb} \gets Z$
\end{algorithmic}
\end{algorithm}

\section{Initial vorticity fields, their time evolution and computational settings}\label{section_initialsetting}
In this section, we describe the setting of the initial vorticity fields and the resolution of the discretization for the computation of the statistical equilibrium. In the numerical calculations, ispack-3.0.1 (www.gfd-dennou.org/arch/ispack/), which is designed based on \cite{ishioka2018new}, is used for the spherical harmonics transform. 

\subsection{Definition of the initial vorticity fields}\label{subsec_setting}
As initial vorticity fields, we consider not only the one defined by \cite{dritschel2015late} but also those with an additional planetary vorticity term. That is, the initial vorticity fields can be represented as,
\begin{align*}
    q_{\rm ini}(\lambda,\mu) = \sqrt{4\pi}\, q_{\rm DQM} (\lambda,\mu) +2 \Omega \mu,
\end{align*}
where \(q_{\rm DQM}\) is the initial vorticity field defined by \cite{dritschel2015late}, which is given in Appendix A of \cite{dritschel2015late}, in the form of spectral coefficients of stream function. The coefficient \(\sqrt{4\pi}\) arises from the difference in the normalization of the spherical harmonics between this paper and \cite{dritschel2015late}. The difference in the normalization does not change the time evolution of the vorticity fields if we change the time unit appropriately. Note that the vorticity value \(\sqrt{4\pi}\) in this study corresponds to the unity in \cite{dritschel2015late} and the unit time interval in this study corresponds to the time interval of \(\sqrt{4\pi}\) in \cite{dritschel2015late}. {The initial vorticity field is composed of the spherical harmonics components up to a total wavenumber of 10, where the amplitude of each component is given by random numbers. Physically, the initial vorticity field contains a number of positive and negative vortex blobs that induce global turbulent mixing. \cite{modin2020casimir} also used this initial vorticity field as a test case of time evolution, and they called it a ``generic'' vorticity field.} As for the parameter \(\Omega\), we examine the following eight values: \(\Omega=0\), \(1.0\times 10^{-2}\), \(2.5\times 10^{-2}\), \(5.0\times 10^{-2}\), \(7.5\times 10^{-2}\), \(1.0\times 10^{-1}\), \(1.8\), and \(2.5\). The vorticity fields corresponding to the eight values of \(\Omega\) are shown in figure \ref{initial_a-d}. We compute the statistical equilibrium for each of these settings.

\begin{figure}[htbp]
    \centering
    \begin{tabular}{cc}
        \begin{minipage}[t]{0.45\hsize}
            \includegraphics[scale=0.25]{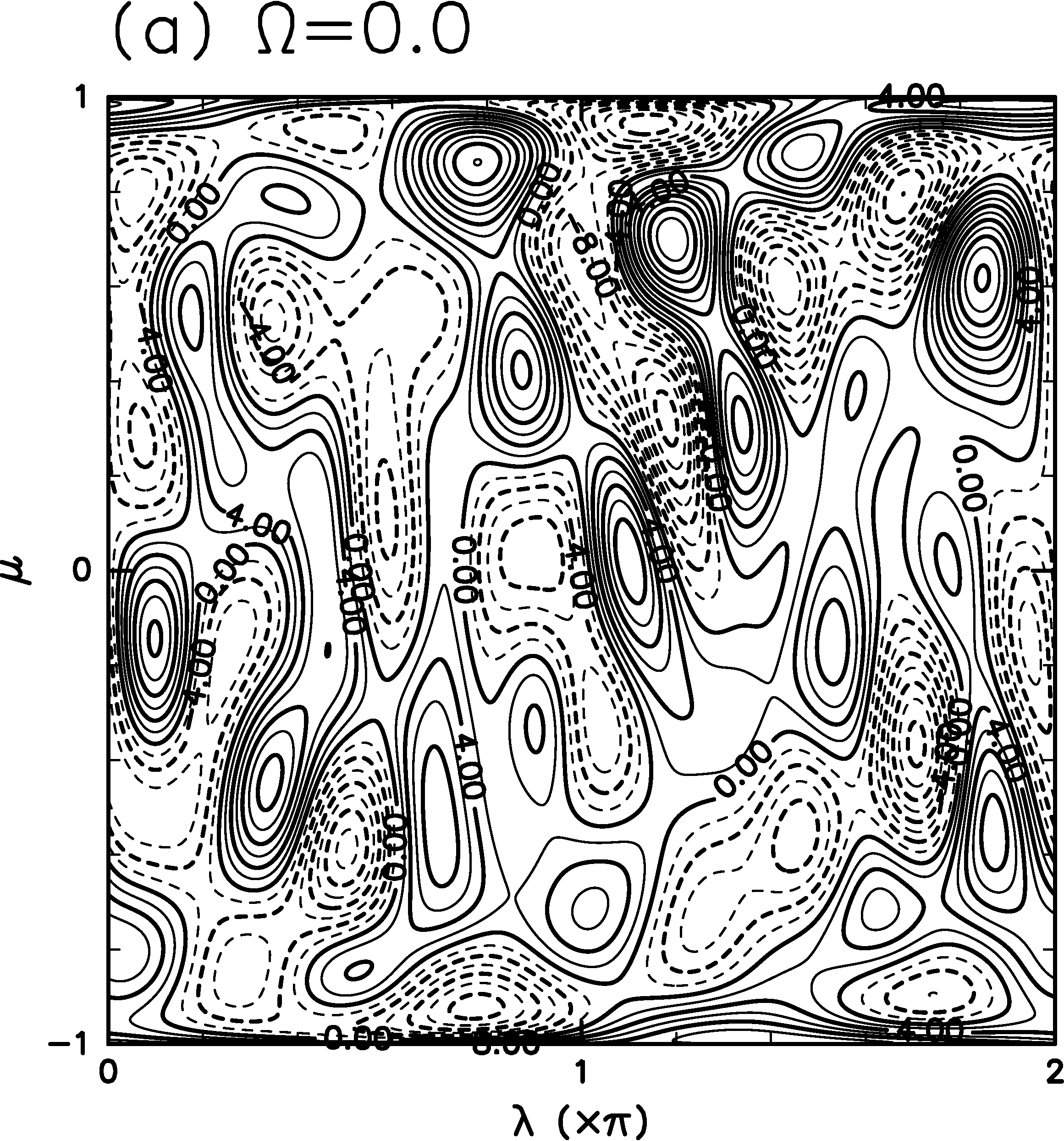}       
            \centering
        \end{minipage}
        &
        \begin{minipage}[t]{0.45\hsize}
            \includegraphics[scale=0.25]{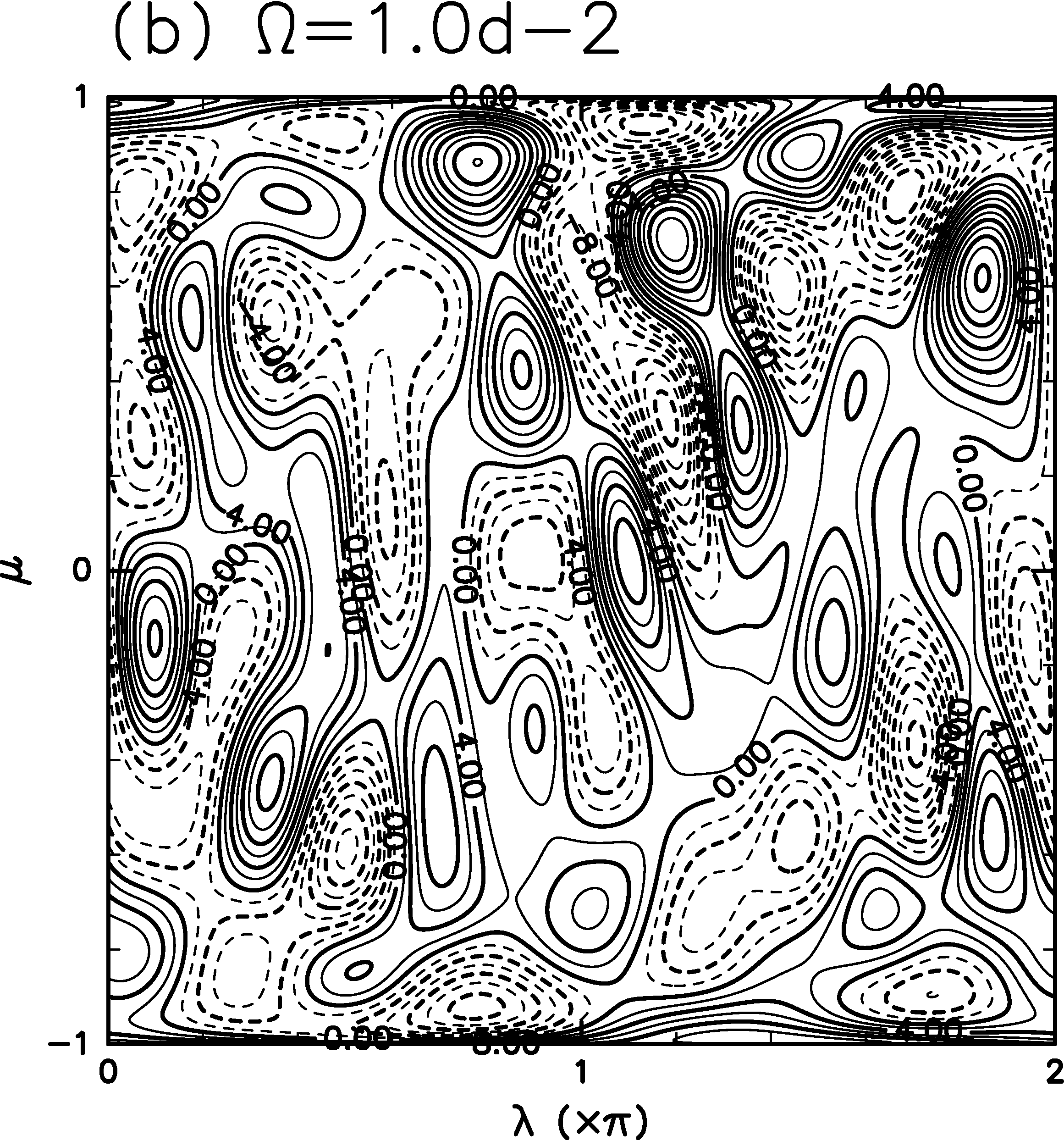}       
            \centering
        \end{minipage}
        \\
        \begin{minipage}[t]{0.45\hsize}
            \includegraphics[scale=0.25]{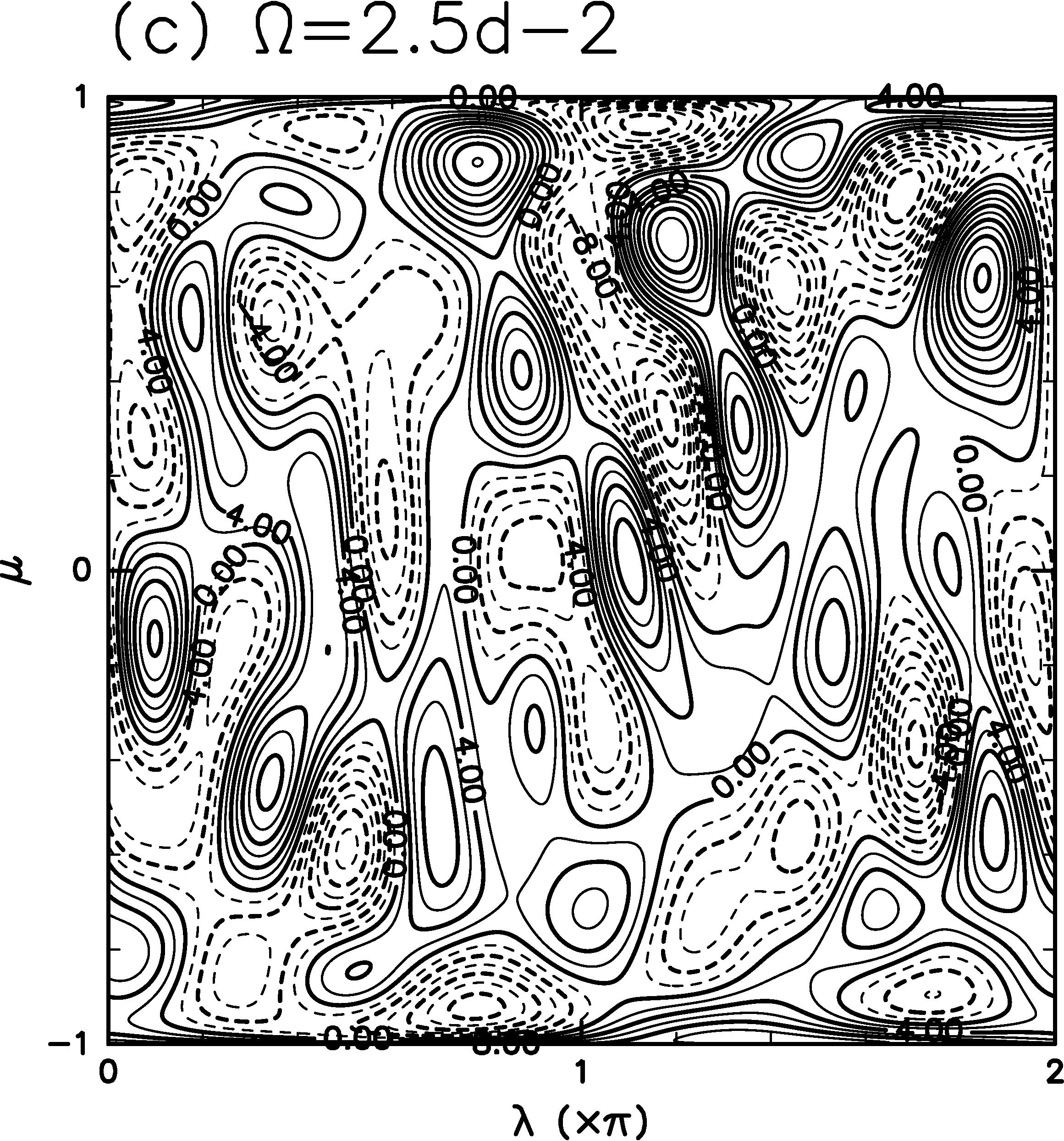}       
            \centering
        \end{minipage}
        &
        \begin{minipage}[t]{0.45\hsize}
            \includegraphics[scale=0.25]{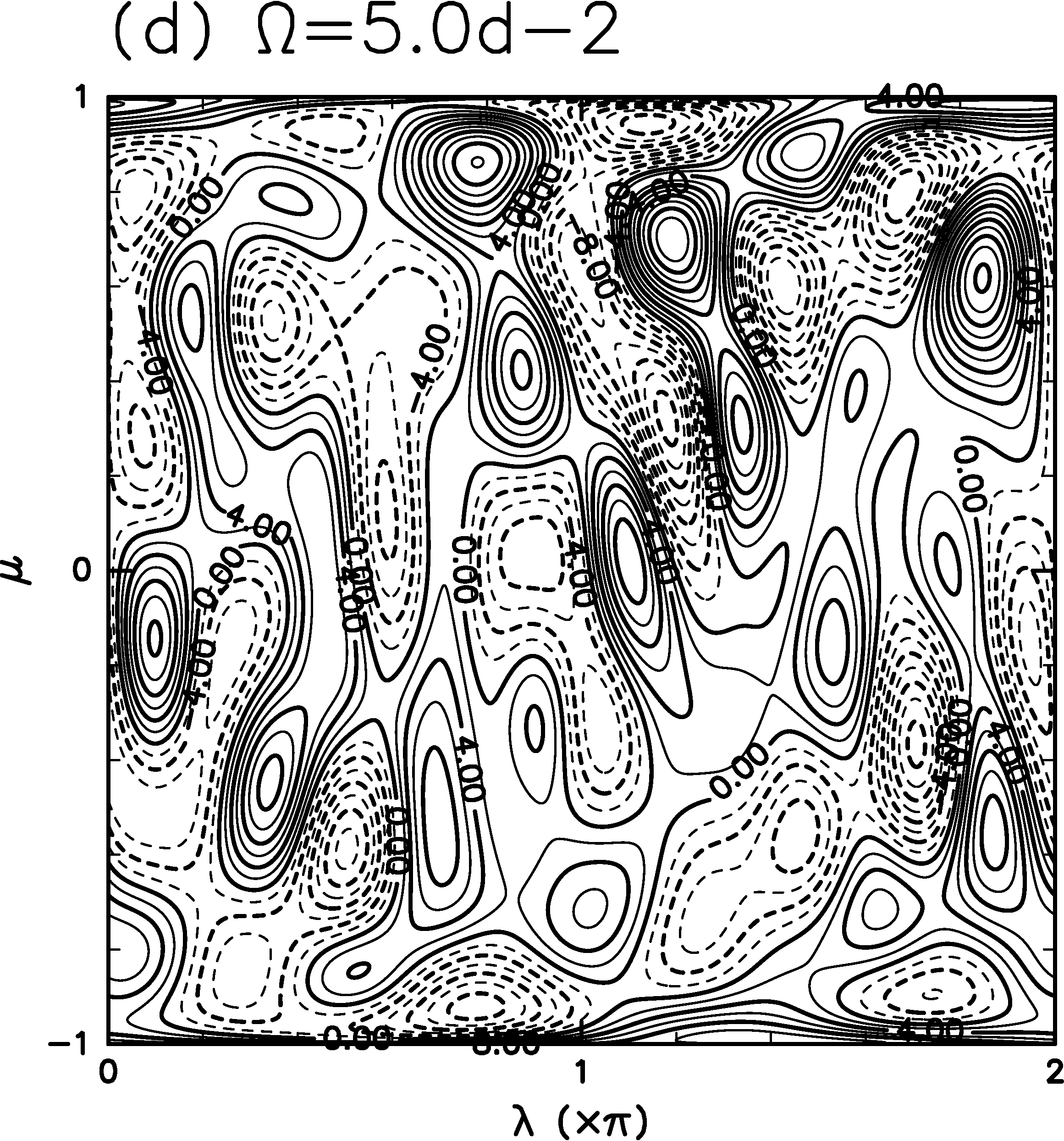}       
            \centering
        \end{minipage}
    \end{tabular}
    \caption{Initial vorticity fields \(q_{\rm ini}\) for the eight values of \(\Omega\). (a) \(\Omega=0\), (b) \(\Omega=1.0\times 10^{-2}\), (c) \(\Omega=2.5\times 10^{-2}\), (d) \(\Omega=5.0\times 10^{-2}\), (e) \(\Omega=7.5\times 10^{-2}\), (f) \(\Omega=1.0\times 10^{-1}\), (g) \(\Omega=1.8\), and (h) \(\Omega=2.5\). Panels for (e), (f), (g), and (h) are shown in the continued figure. In each panel, the corresponding \(\Omega\) value is shown at the top of the panel, the horizontal axis is the longitude \(\lambda/(2\pi)\), and the vertical axis is the sine-latitude \(\mu\). Contour interval is set to 2.0.}
    \label{initial_a-d}
\end{figure}
\addtocounter{figure}{-1}

\begin{figure}[htbp]
    \centering
    \begin{tabular}{cc}
        \begin{minipage}[t]{0.40\hsize}
            \includegraphics[scale=0.25]{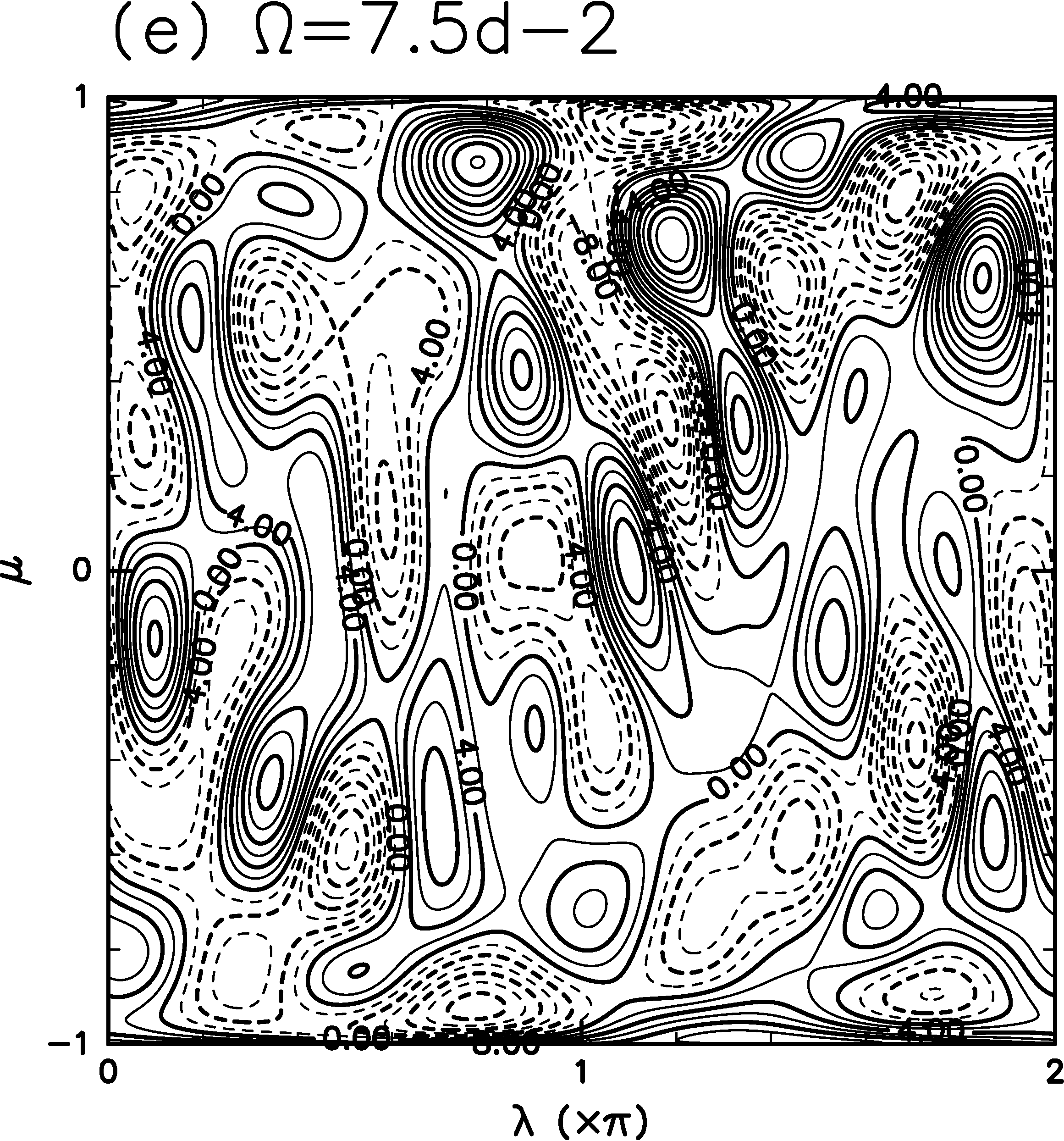}     
            \centering
        \end{minipage}
        &
        \begin{minipage}[t]{0.40\hsize}
            \includegraphics[scale=0.25]{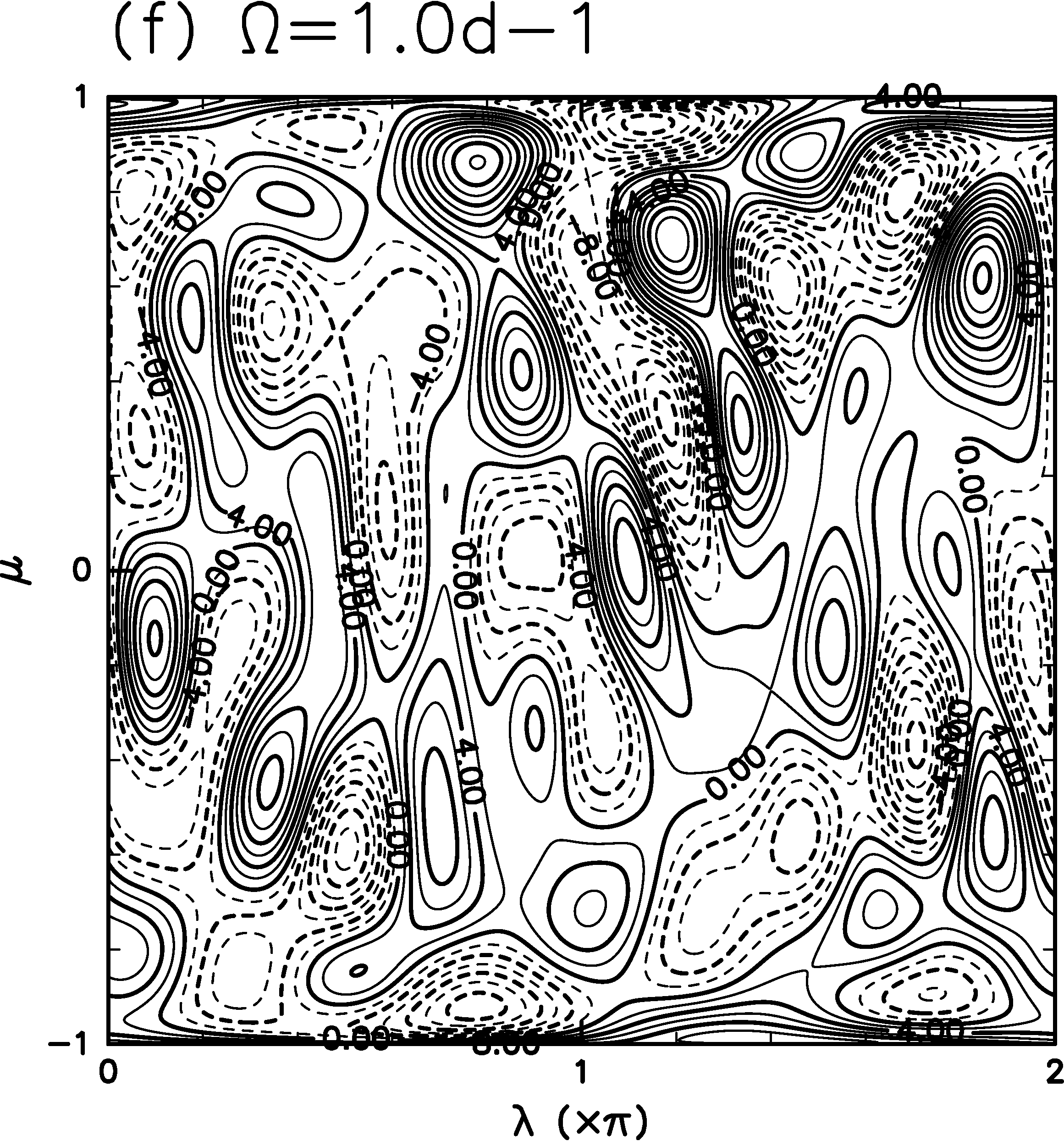}       
            \centering
        \end{minipage}
        \\
        \begin{minipage}[t]{0.40\hsize}
            \includegraphics[scale=0.25]{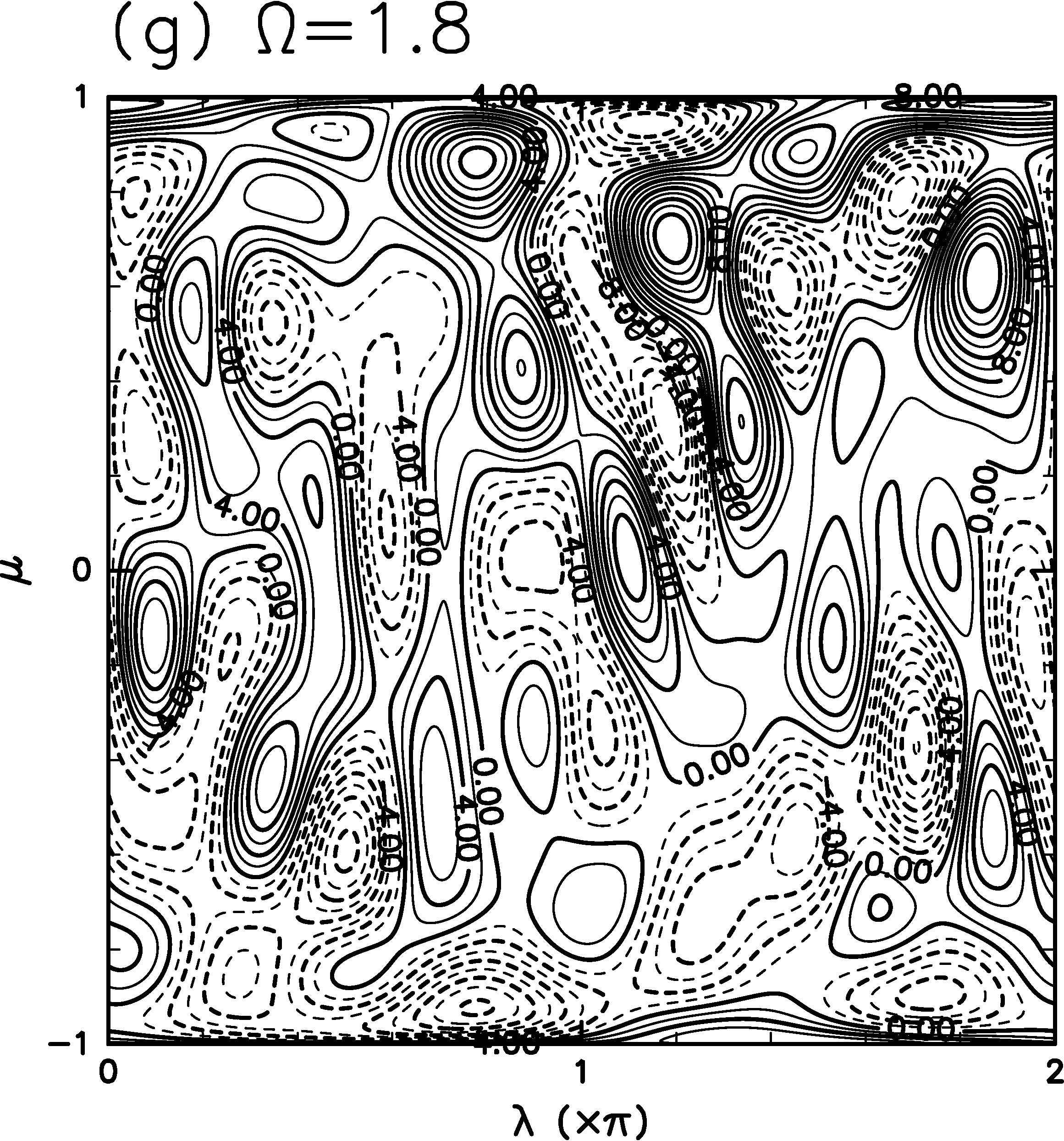}      
            \centering
        \end{minipage}
        &
        \begin{minipage}[t]{0.40\hsize}
            \includegraphics[scale=0.25]{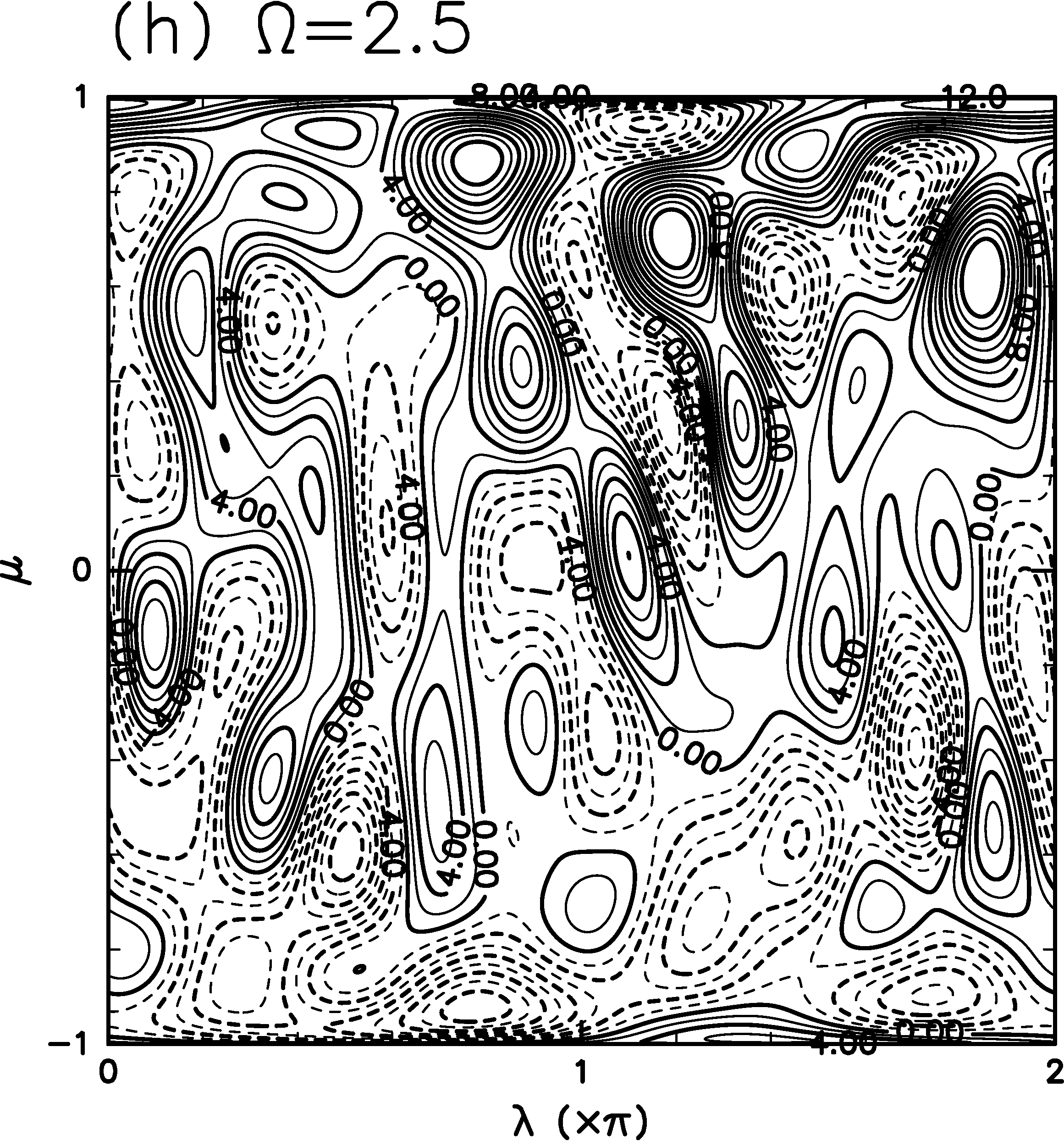}      
            \centering
        \end{minipage}

    \end{tabular}
    \caption{(continued.)}
    \label{initial_e-h}
\end{figure}

\subsection{Time integration from the initial vorticity fields}\label{subsection_timeevo}
To compare with statistical equilibria, we perform long-time integrations from the initial vorticity fields. We numerically integrate the Navier-Stokes equation with a small viscosity term:
\begin{align}
    \frac{\partial q}{\partial t} + J(\psi,q) = \nu (\Delta +2) q. \label{NS_equation}
\end{align}
The parameter \(\nu\) is the viscosity coefficient. The addition of a viscosity term is necessary in numerical integration because without it enstrophy would accumulate at the grid scale. We take the viscosity term to be of the form \(\nu (\Delta +2) q\) so that the conservation of angular momentum is exactly satisfied. The equation \eqref{NS_equation} is discretized by using the spectral method with the truncation wavenumber \(N=682\). The nonlinear term \(J(\psi, q)\) is evaluated via the transform method using \(2048 \times 1024\) grids on the sphere. The viscosity coefficient is set to \(\nu = 1/[N(N+1)-2]\). The classical 4th-order Runge-Kutta method is used for the time integration. The time step \(\Delta t\) is \(1/1000\), except for the case \(\Omega=2.5\) where \(\Delta t\) is set to \(1/2000\). 

{Since the time integrations are performed with a small viscosity, the symplectic structure of the Euler system is violated. We show the values of the ratio of remaining energy and enstrophy at \(t=200\) for each case in Table \ref{tab:energy_enstrophy}. In each case, at least 99\% of the initial energy is conserved, but the enstrophy is well dissipated and only 19\%--30\% of the initial value remains. The loss of the enstrophy is inevitable. It is ``fragile'' \citep[e.g.,][]{naso2010statistical} in the sense that it is sensitive to local high wavenumber fluctuations of the vorticity field and decays with viscosity. However, we believe that qualitative results are robust in the presence of a small viscosity. The results of \cite{modin2020casimir}, which performed Casimir-conserving time integrations, support this. Furthermore, the observation that local fluctuations of the vorticity are smoothed out by small viscosity is consistent with the coarse-graining process adopted in the MRS theory \citep{sommeria1991final}. }

\begin{table}[htbp]
    
    \centering
    \caption{The values of the ratio of the energy and the enstrophy at \(t=200\) to the initial values in the time integrations for the eight values of \(\Omega\). }
	\begin{tabular}{@{}ccc}

          \hline
          \(\Omega\)   &  energy & enstrophy \\ 
          \hline\hline
         \(0\)  & 0.9906 & 0.2014 \\
         \(1.0\times 10^{-2}\) & 0.9906 & 0.2007 \\
         \(2.5\times 10^{-2}\) & 0.9904 & 0.2047 \\
         \(5.0\times 10^{-2}\) & 0.9907 & 0.1995 \\
         \(7.5\times 10^{-2}\) & 0.9906 & 0.2017 \\
         \(1.0\times 10^{-1}\) & 0.9909 & 0.1948 \\
         \(1.8\) & 0.9982 & 0.2406 \\
         \(2.5\) & 0.9990 & 0.3069 \\ 
         \hline
    \end{tabular}
    \label{tab:energy_enstrophy}
\end{table}

We show the vorticity fields at \(t=200\), which are obtained by the numerical time integration from the initial vorticity fields, in figure \ref{t200_a-d}. For the initial vorticity field with \(\Omega=0\), the quadrupole state consisting of two positive vortices and two negative ones appeared (figure \ref{t200_a-d}(a)). This quadrupole state has been confirmed to be stable, but with a quasi-periodic motion, by the longer time integrations of \cite{dritschel2015late} and \cite{modin2020casimir}. For the initial vorticity field with non-zero angular momentum, the initial vorticity field with \(\Omega=1.0\times 10^{-2}, 2.5\times 10^{-2}, 5.0\times 10^{-2}\) and \(7.5\times 10^{-2}\) evolved into quadrupole states (figure \ref{t200_a-d}(b)--\ref{t200_a-d}(e)). The initial vorticity field with \(\Omega=1.0\times 10^{-1}\) and \(1.8\) evolved to the state with three coherent vortices (figure \ref{t200_a-d}(f) and \ref{t200_a-d}(g)), and the initial vorticity field with \(\Omega=2.5\) evolved to the state with two vortices (figure \ref{t200_a-d}(h)). This is consistent with the finding of \cite{modin2020casimir} that the number of coherent vortices decreases from four to three, and to two as the ratio of the angular momentum to the square root of the enstrophy increases. {Note that, the above results are not sensitive to the values of viscosity coefficient at least in the range from \(0.1\nu\) to \(10\nu\).}

\begin{figure}[htbp]
    \centering
    \begin{tabular}{cc}
        \begin{minipage}[t]{0.45\hsize}
            \includegraphics[scale=0.25]{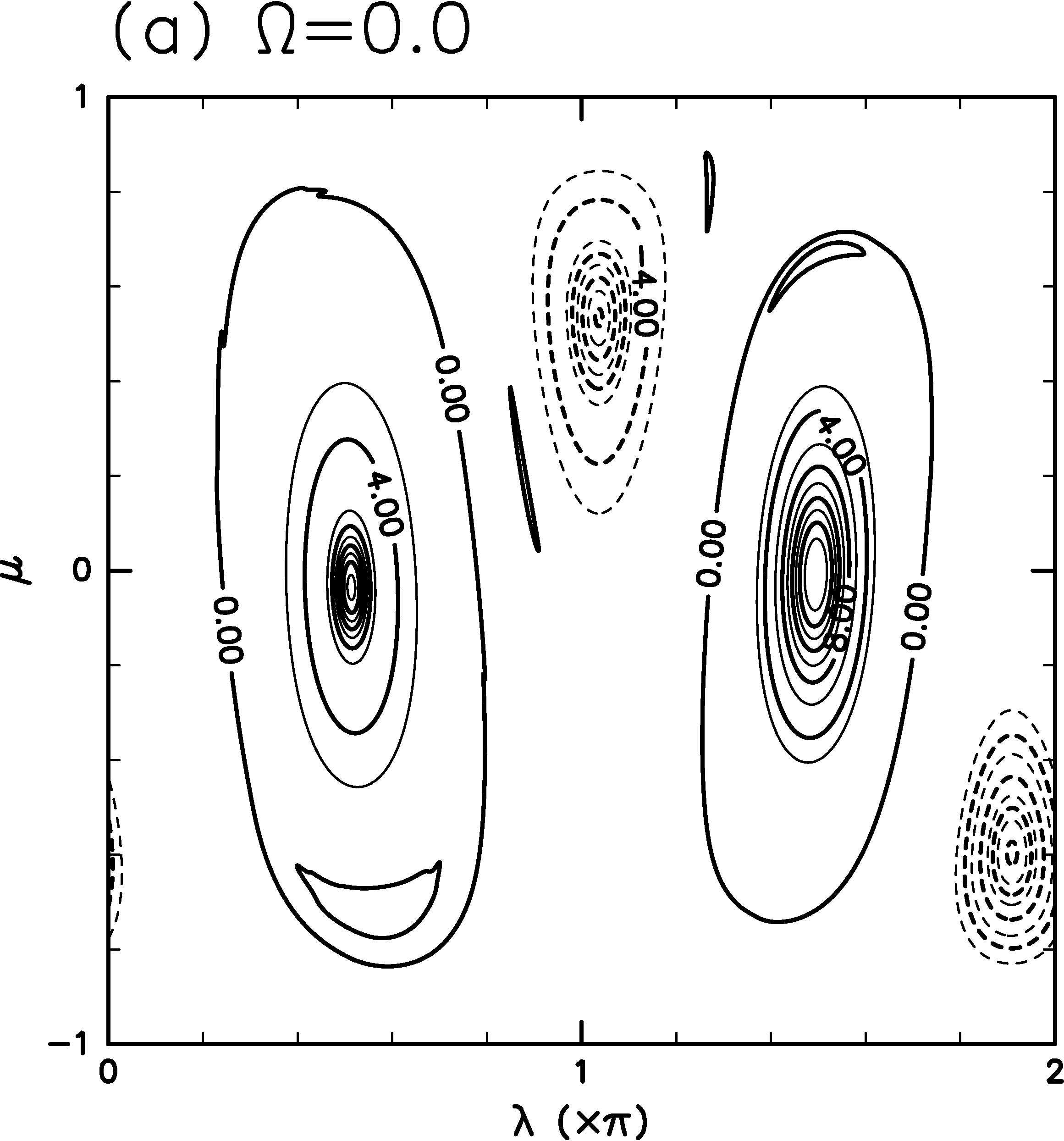}       
            \centering
        \end{minipage}
        &
        \begin{minipage}[t]{0.45\hsize}
            \includegraphics[scale=0.25]{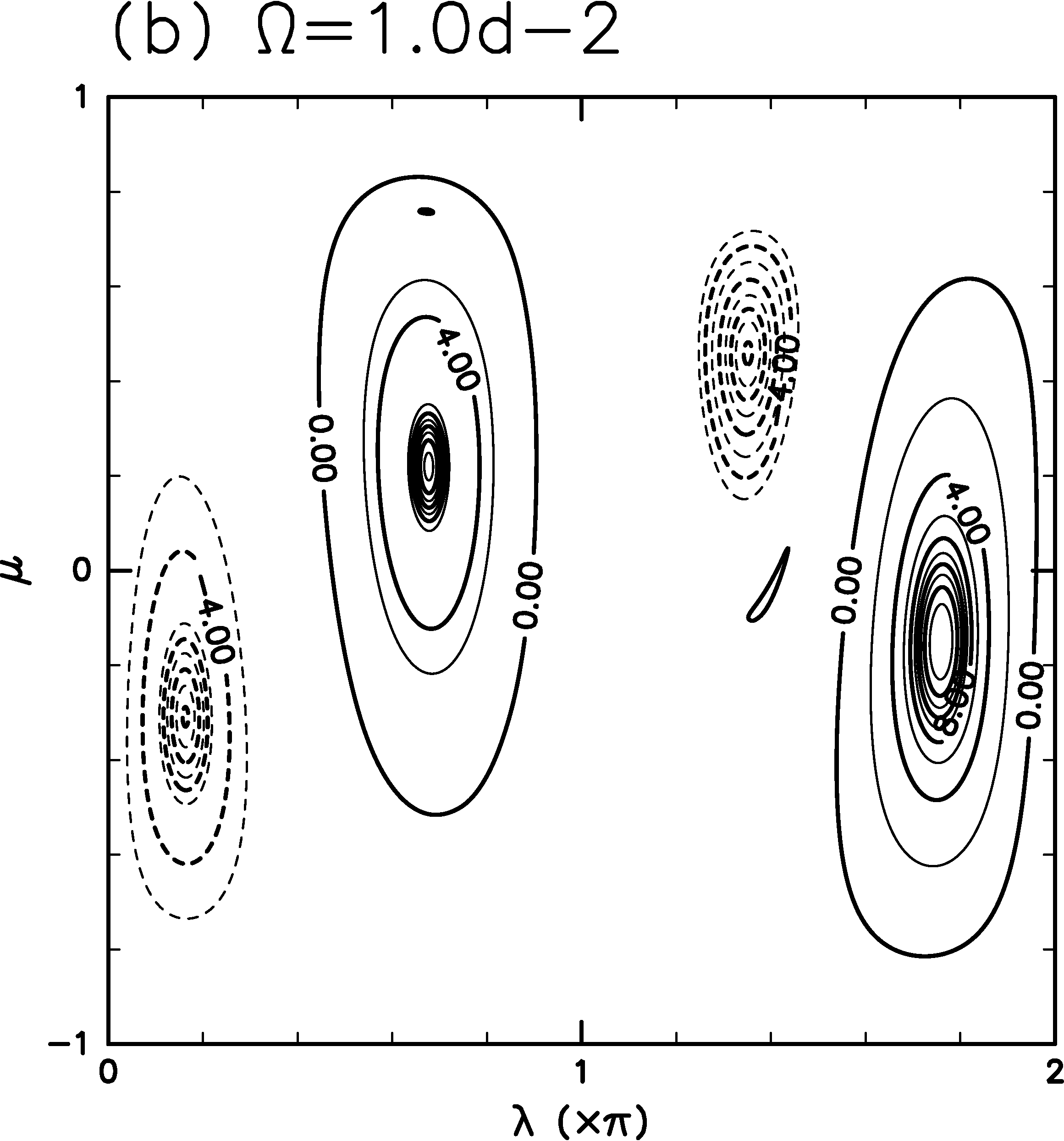}       
            \centering
        \end{minipage}
        \\
        \begin{minipage}[t]{0.45\hsize}
            \includegraphics[scale=0.25]{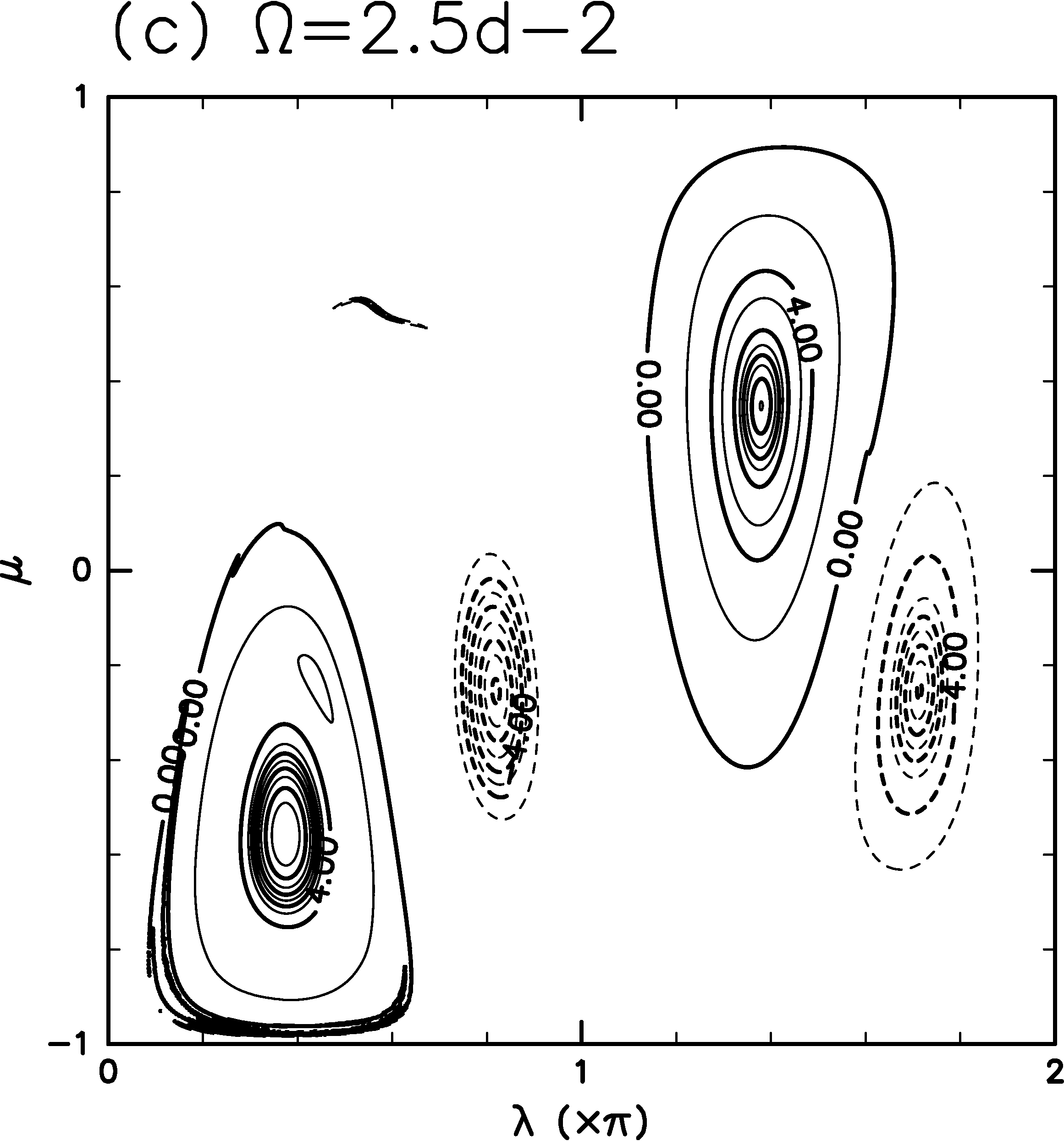}       
            \centering
        \end{minipage}
        &
        \begin{minipage}[t]{0.45\hsize}
            \includegraphics[scale=0.25]{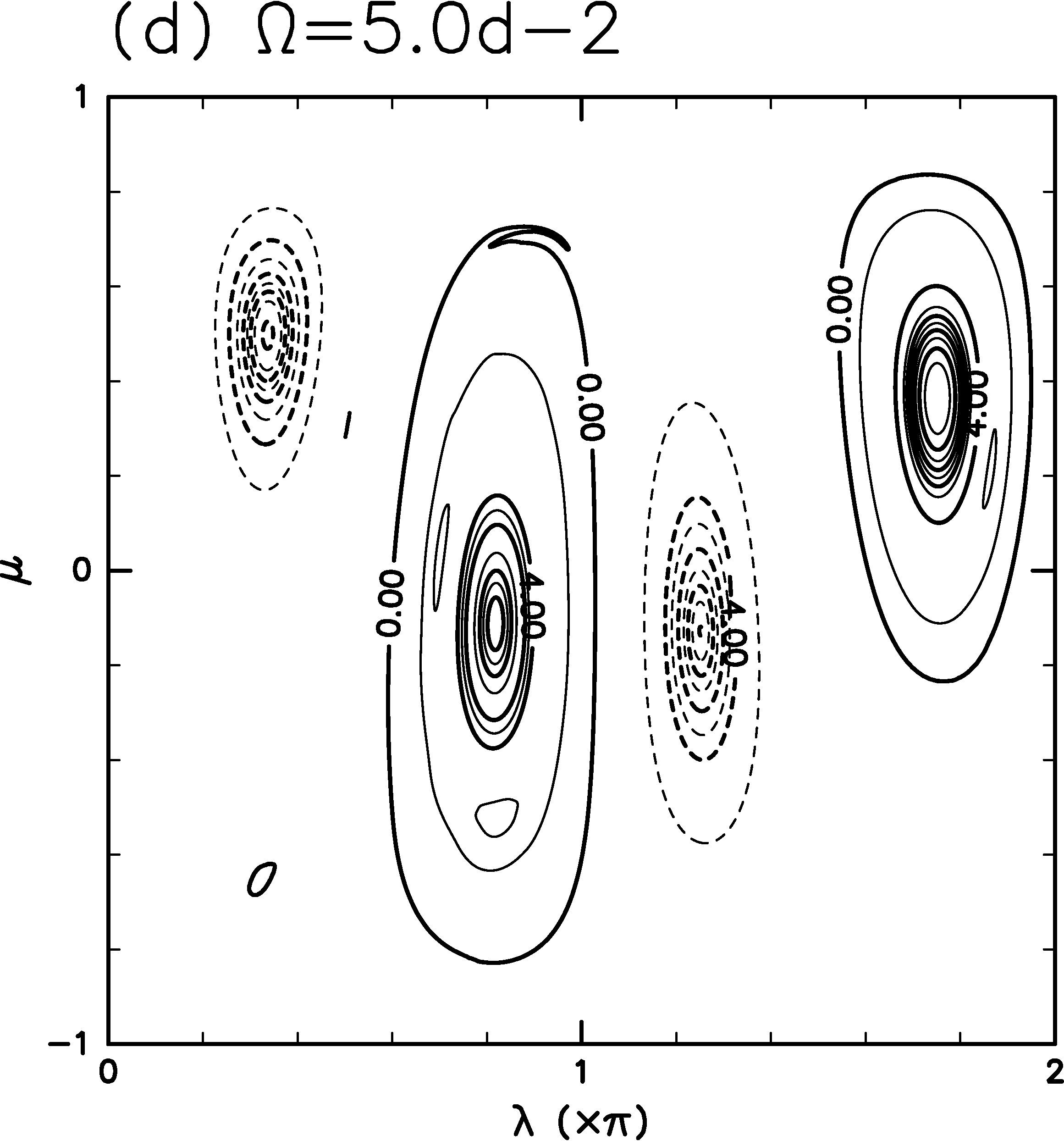}       
            \centering
        \end{minipage}
    \end{tabular}
    \caption{Vorticity fields \(q\) at \(t=200\) obtained by numerical integration from the initial vorticity fields corresponding to \(\Omega=0, 1.0\times 10^{-2}, 2.5\times 10^{-2}\), and \(5.0\times 10^{-2}\) (the cases of \(\Omega=7.5\times 10^{-2}, 1.0\times 10^{-1}\), and \(2.5\) are shown in the continued figure). In each panel, the corresponding \(\Omega\) is shown at the top of the panel, the horizontal axis is the longitude \(\lambda/(2\pi)\), and the vertical axis is the sine-latitude \(\mu\). Contour interval is set to 2.0.}
    \label{t200_a-d}
\end{figure}

\addtocounter{figure}{-1}
\begin{figure}[htbp]
    \centering
    \begin{tabular}{cc}
        \begin{minipage}[t]{0.45\hsize}
            \includegraphics[scale=0.25]{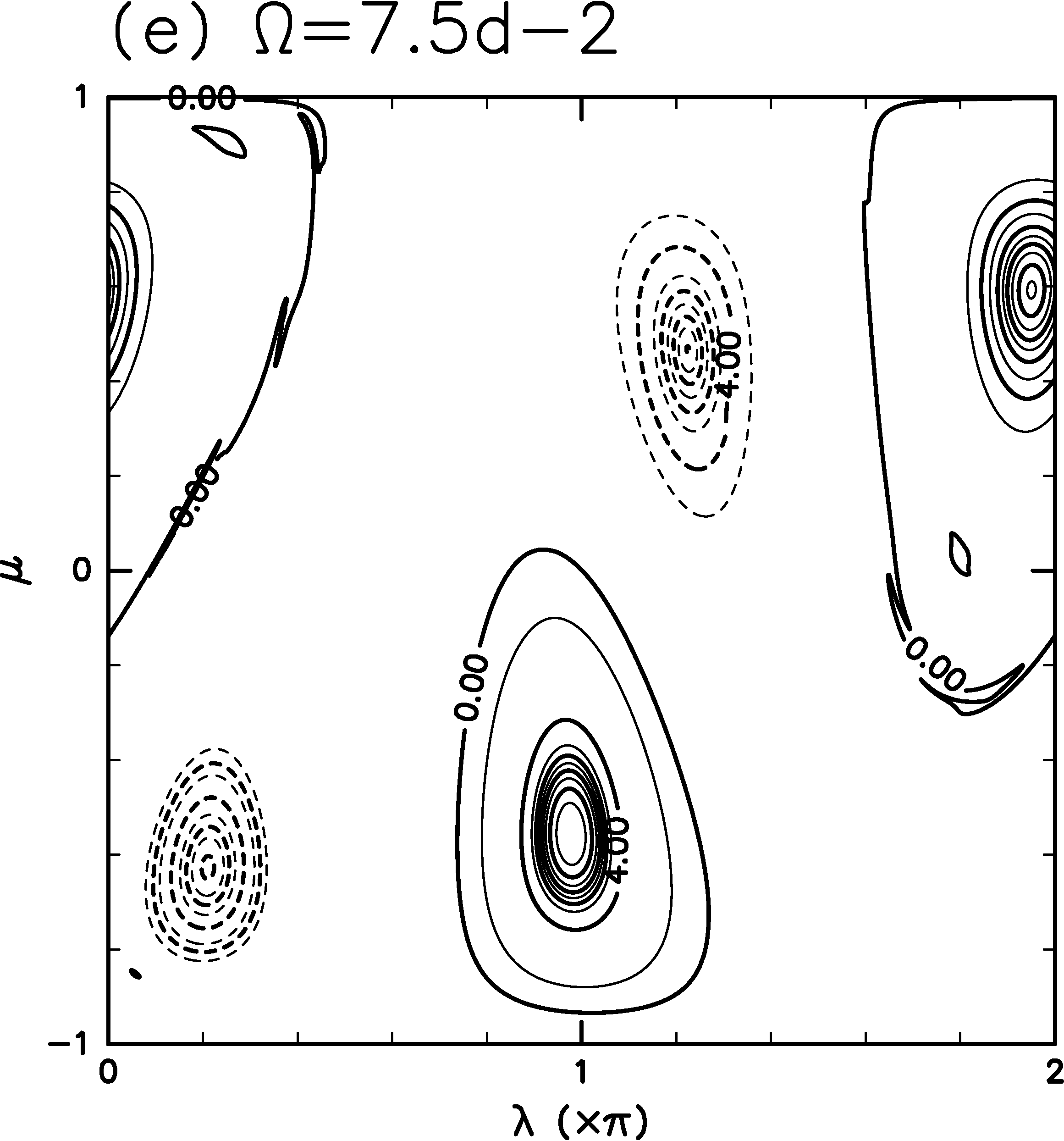}       
            \centering
        \end{minipage}
        &
        \begin{minipage}[t]{0.45\hsize}
            \includegraphics[scale=0.25]{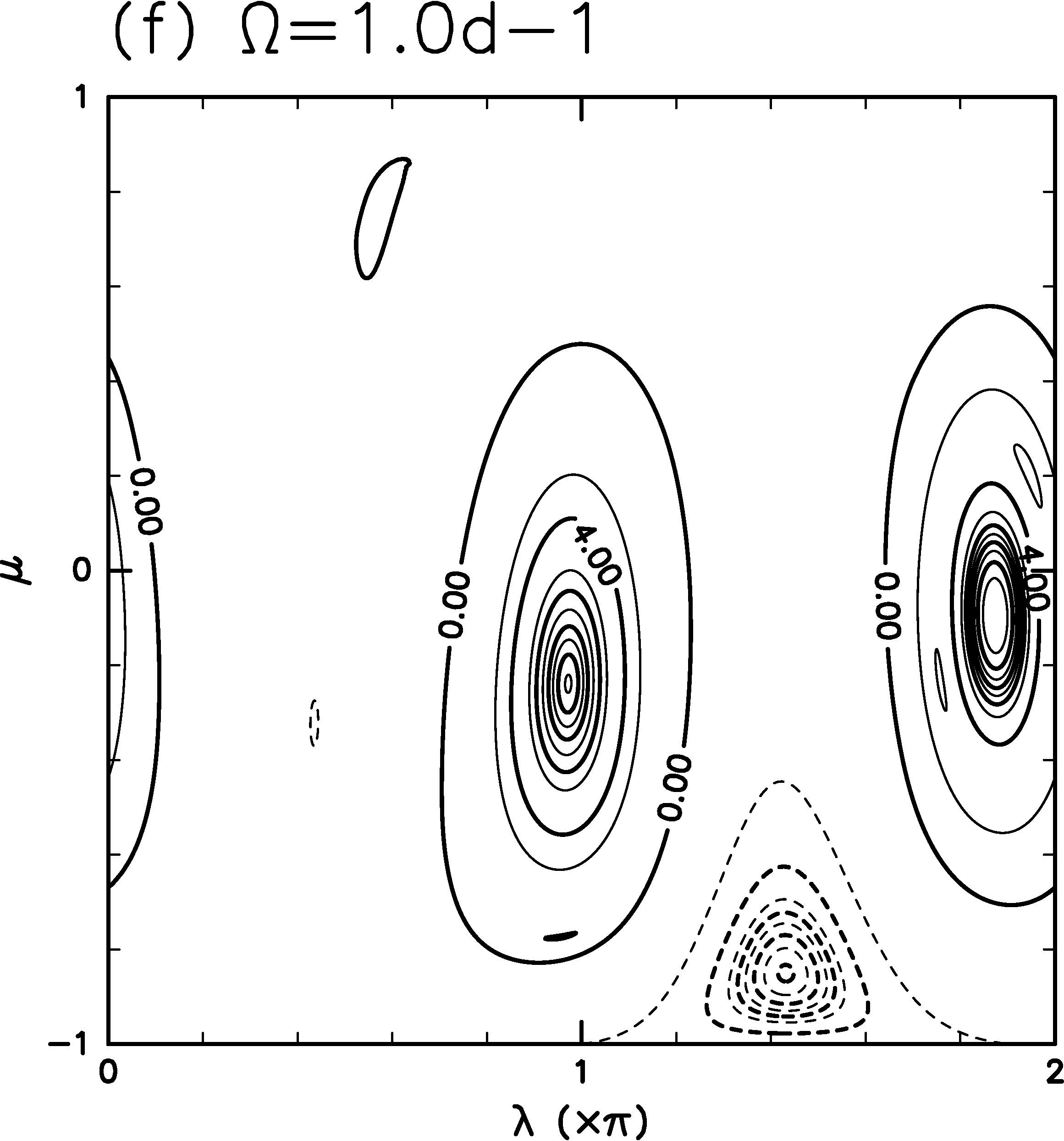}       
            \centering
        \end{minipage}
        \\
        \begin{minipage}[t]{0.45\hsize}
            \includegraphics[scale=0.25]{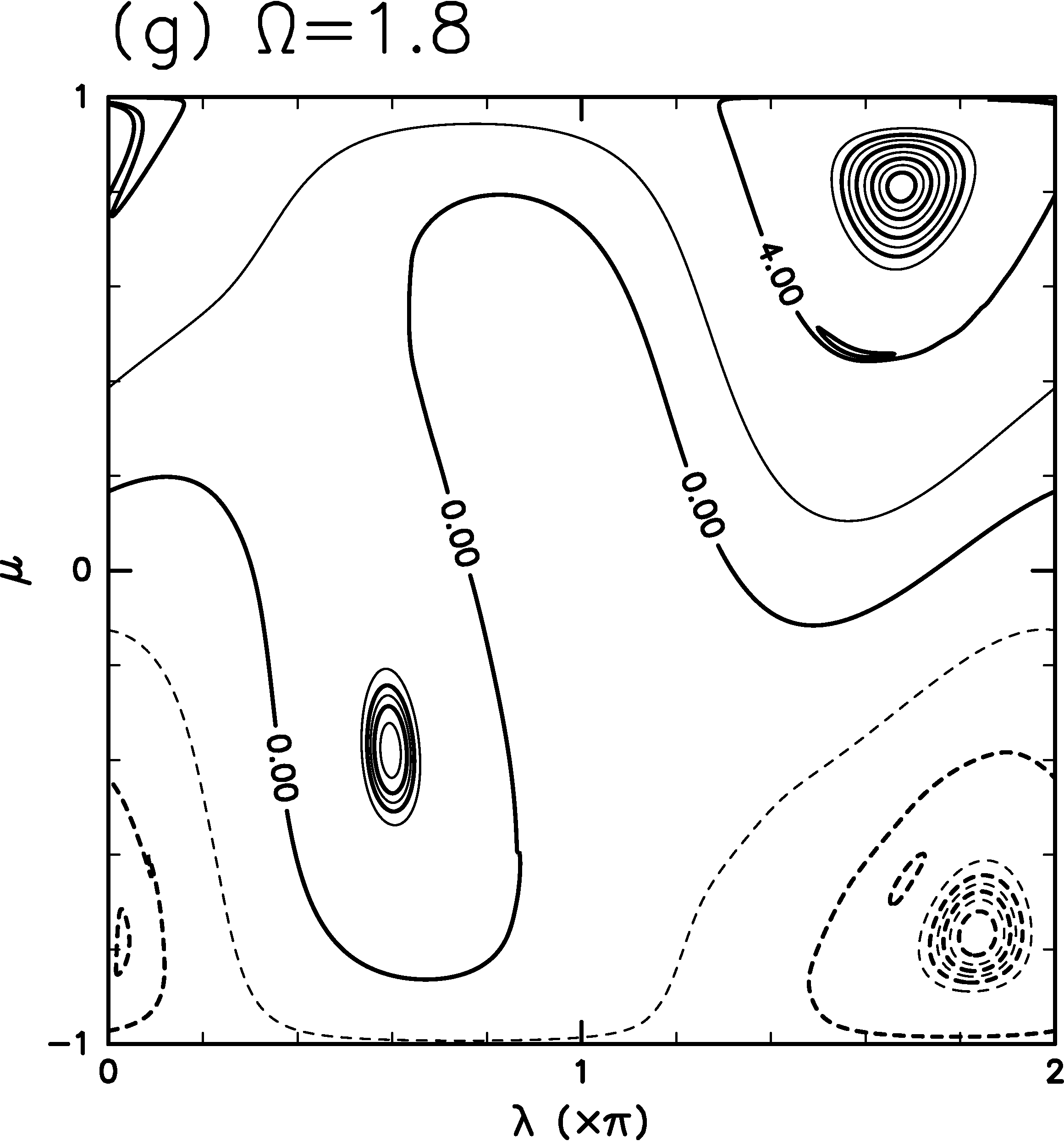}       
            \centering
        \end{minipage}
        &
        \begin{minipage}[t]{0.45\hsize}
            \includegraphics[scale=0.25]{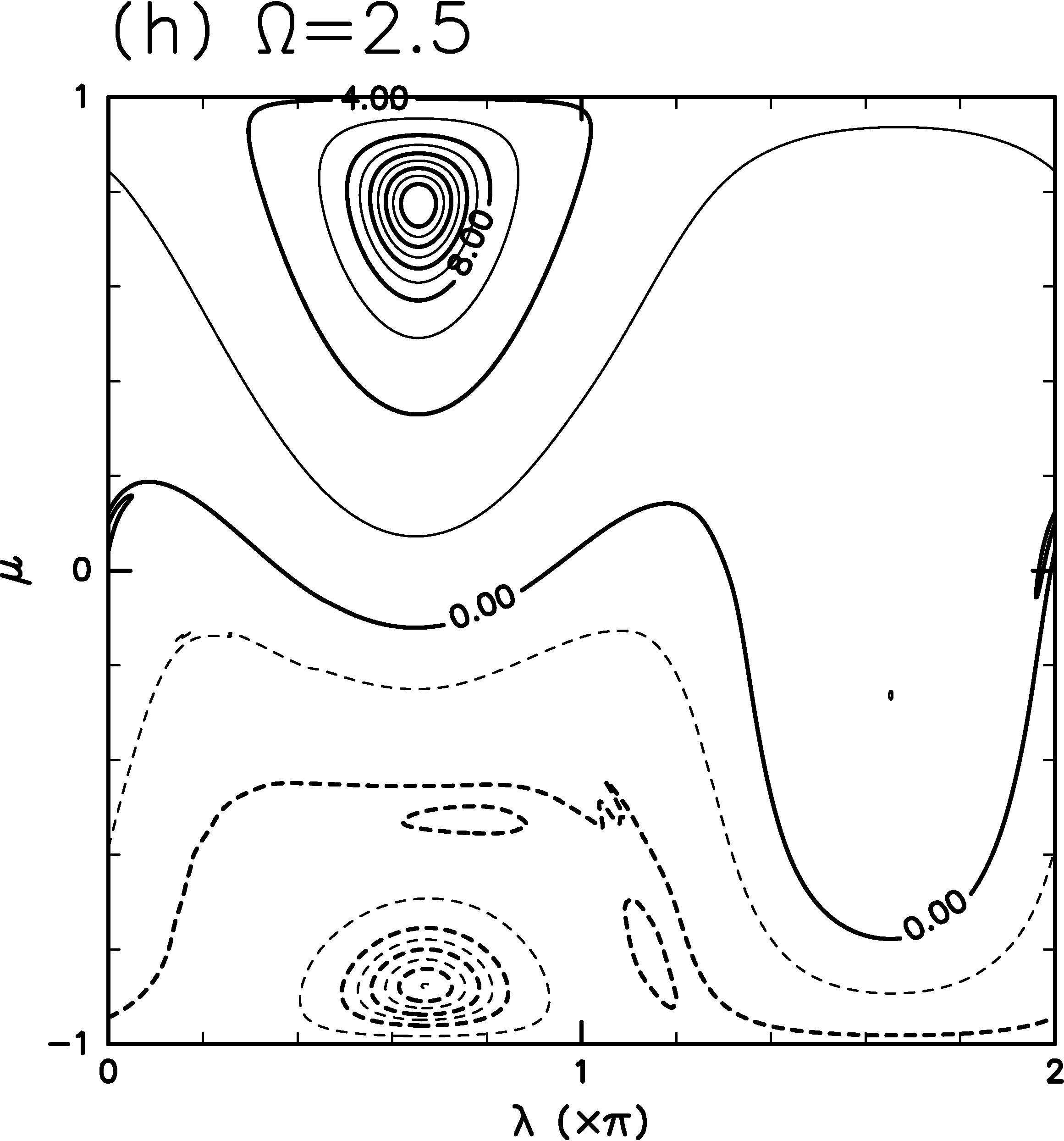}       
            \centering
        \end{minipage}
    \end{tabular}
    \caption{(continued.)}
    \label{t200_e-h}
\end{figure}

\subsection{Settings of statistical equilibria computation}\label{subsec_setting_computation}
The number of the computational grids for the computation of the statistical equilibria is set as \(I=64\) and \(J=32\), and the truncation wavenumber of the spherical harmonic expansion is set as \(N=21\). The vorticity patch approximation for each case is determined by the method of section \ref{subsec_patchgenerate}. {The number of vorticity patches is set to \(K=32\).}
The small parameters \(\alpha\) and \(\varepsilon\) are set as \(\alpha =1.0\times 10^{-4}\) and \(\varepsilon = 1.0\times 10^{-10}\), respectively, except for the case of \(\Omega=5.0\times 10^{-2}\), where \(\varepsilon\) is set as \(\varepsilon = 1.0\times 10^{-11}\) because {in this case the initial search point fails to be in the interior of the feasible region if we set \(\varepsilon=1.0\times 10^{-10}\), and smaller value of \(\varepsilon\) is presumed.} 
The determined areas and vorticity values corresponding to the vorticity patch approximation for each case are shown in figure \ref{fig_vorticitypatch}.

\begin{figure}[htbp]
    \centering
    \begin{tabular}{cc}
        \begin{minipage}[t]{0.50\hsize}
            \includegraphics[scale=0.32, angle=270]{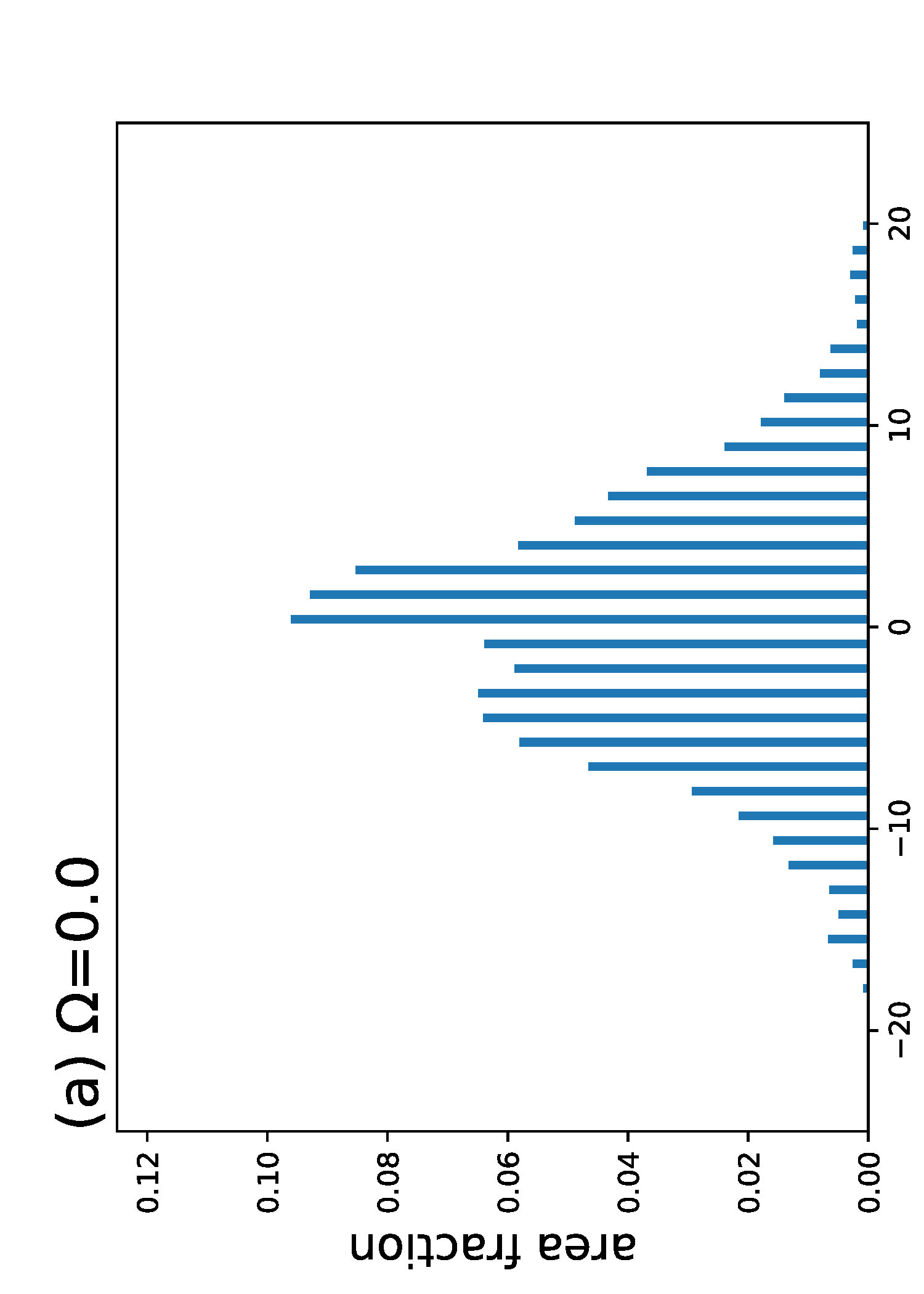}       
            \centering
        \end{minipage}
         &  
        \begin{minipage}[t]{0.50\hsize}
            \includegraphics[scale=0.32, angle=270]{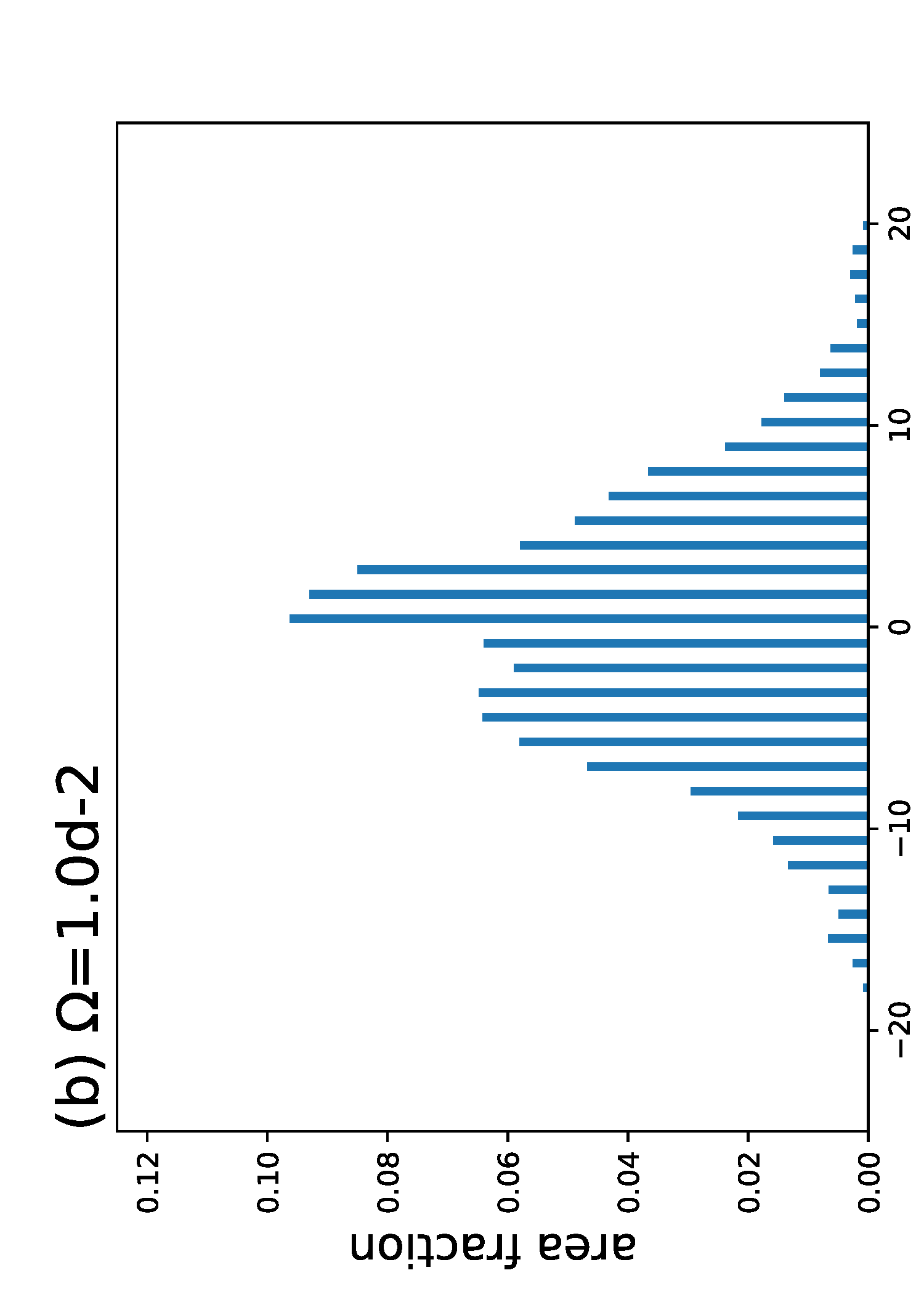}       
            \centering
        \end{minipage} \\

        \begin{minipage}[t]{0.50\hsize}
            \includegraphics[scale=0.32, angle=270]{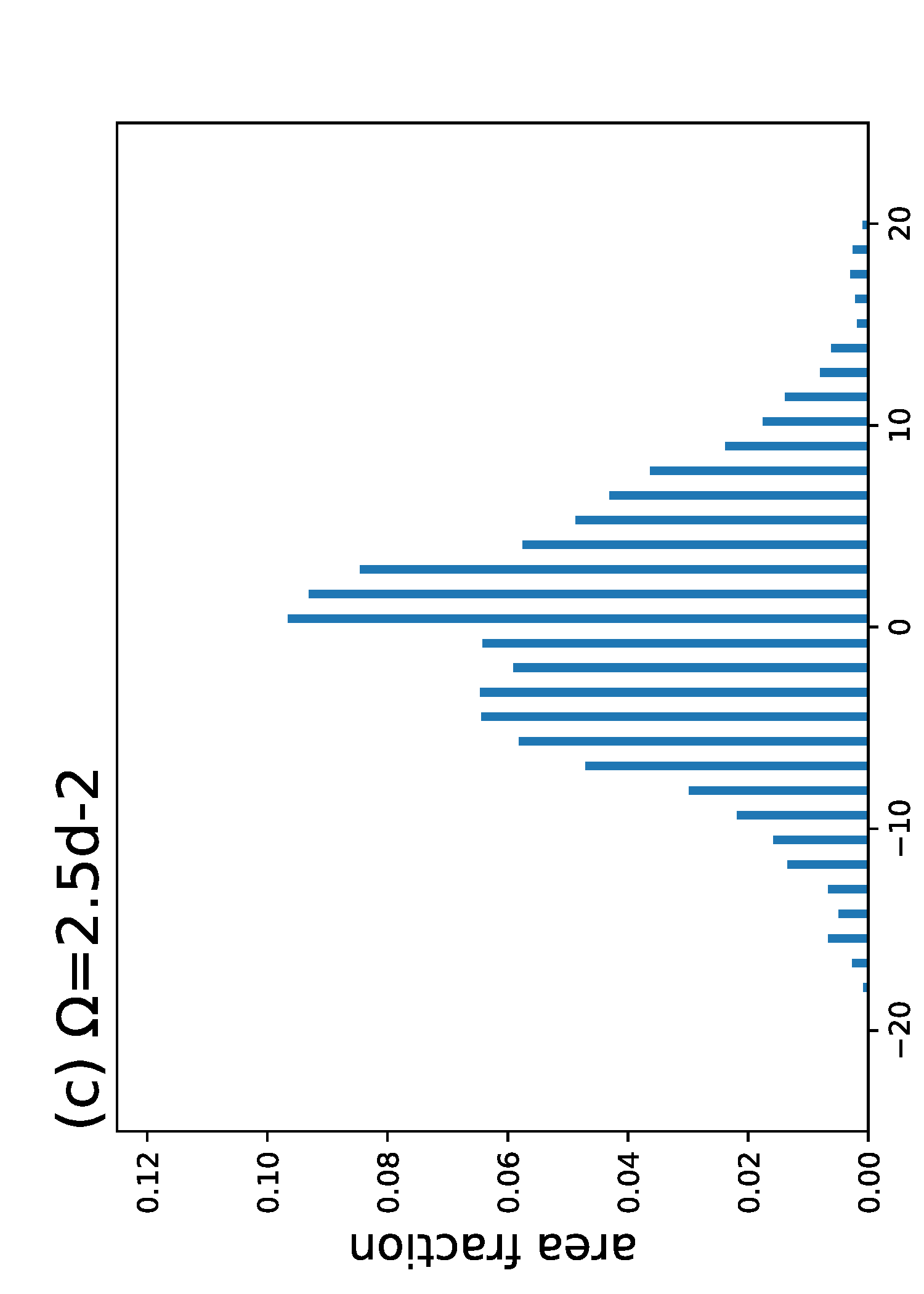}       
            \centering
        \end{minipage}
         &  
        \begin{minipage}[t]{0.50\hsize}
            \includegraphics[scale=0.32, angle=270]{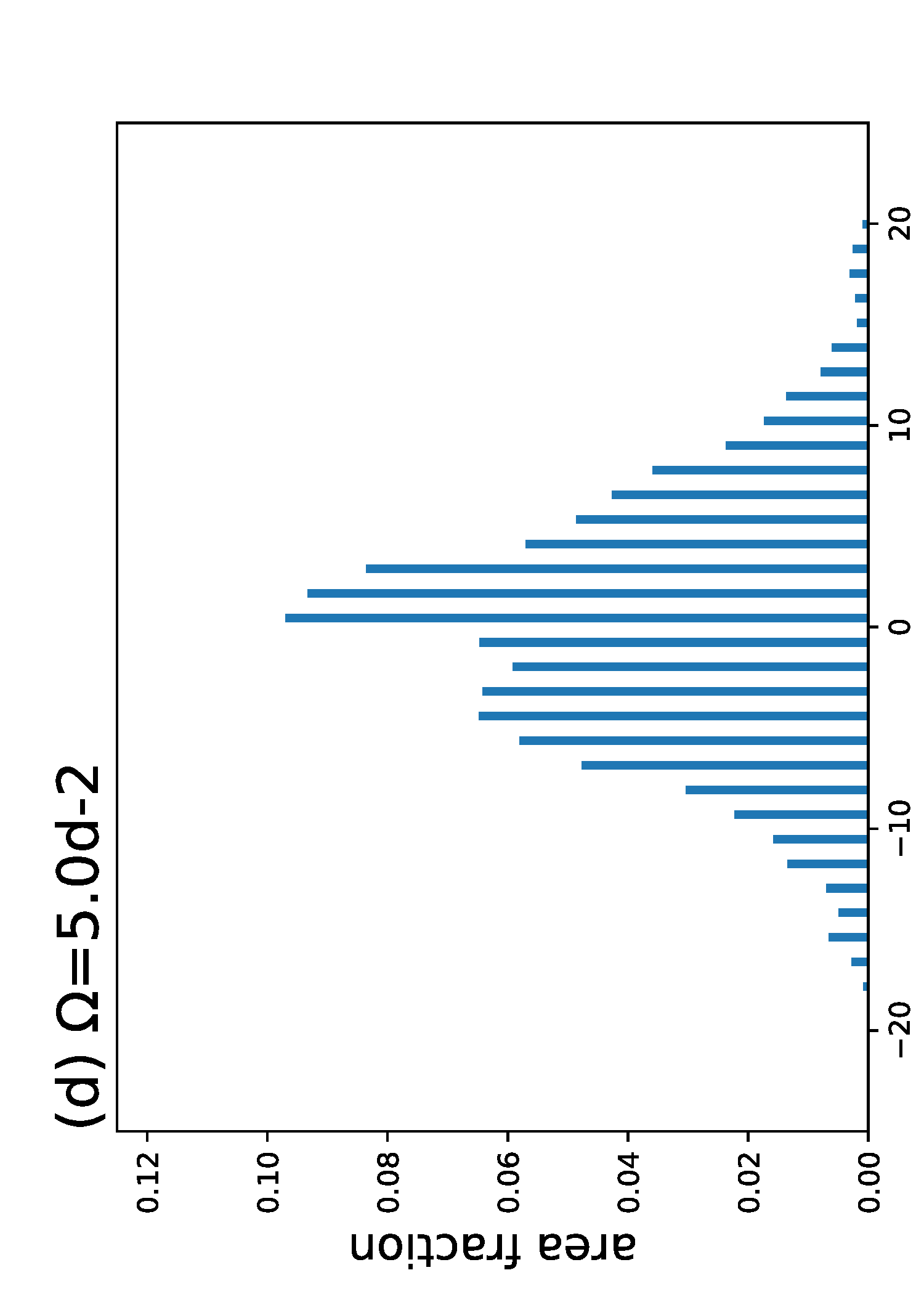}       
            \centering
        \end{minipage} \\
        
        \begin{minipage}[t]{0.50\hsize}
            \includegraphics[scale=0.32, angle=270]{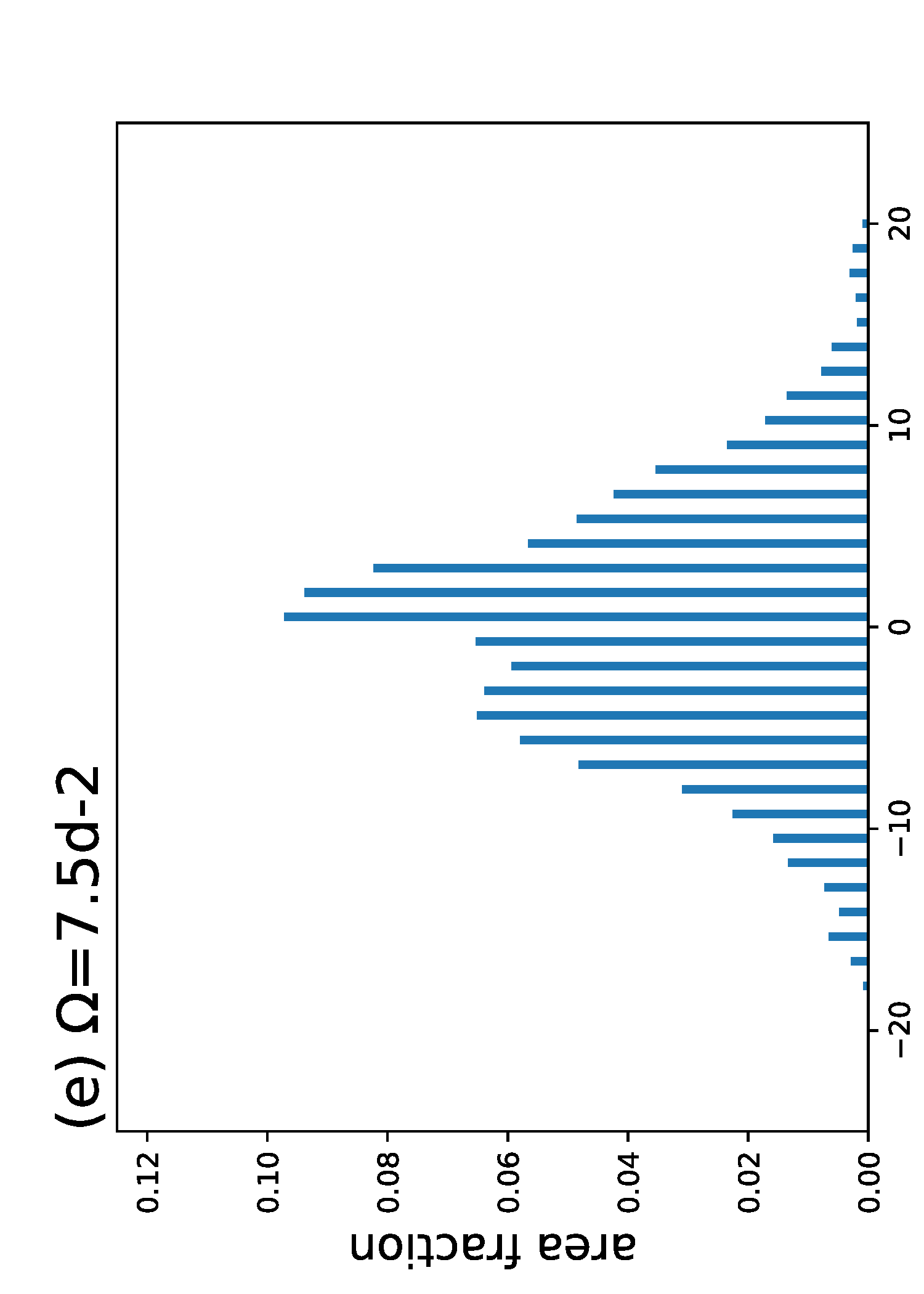}       
            \centering
        \end{minipage}
         &  
        \begin{minipage}[t]{0.50\hsize}
            \includegraphics[scale=0.32, angle=270]{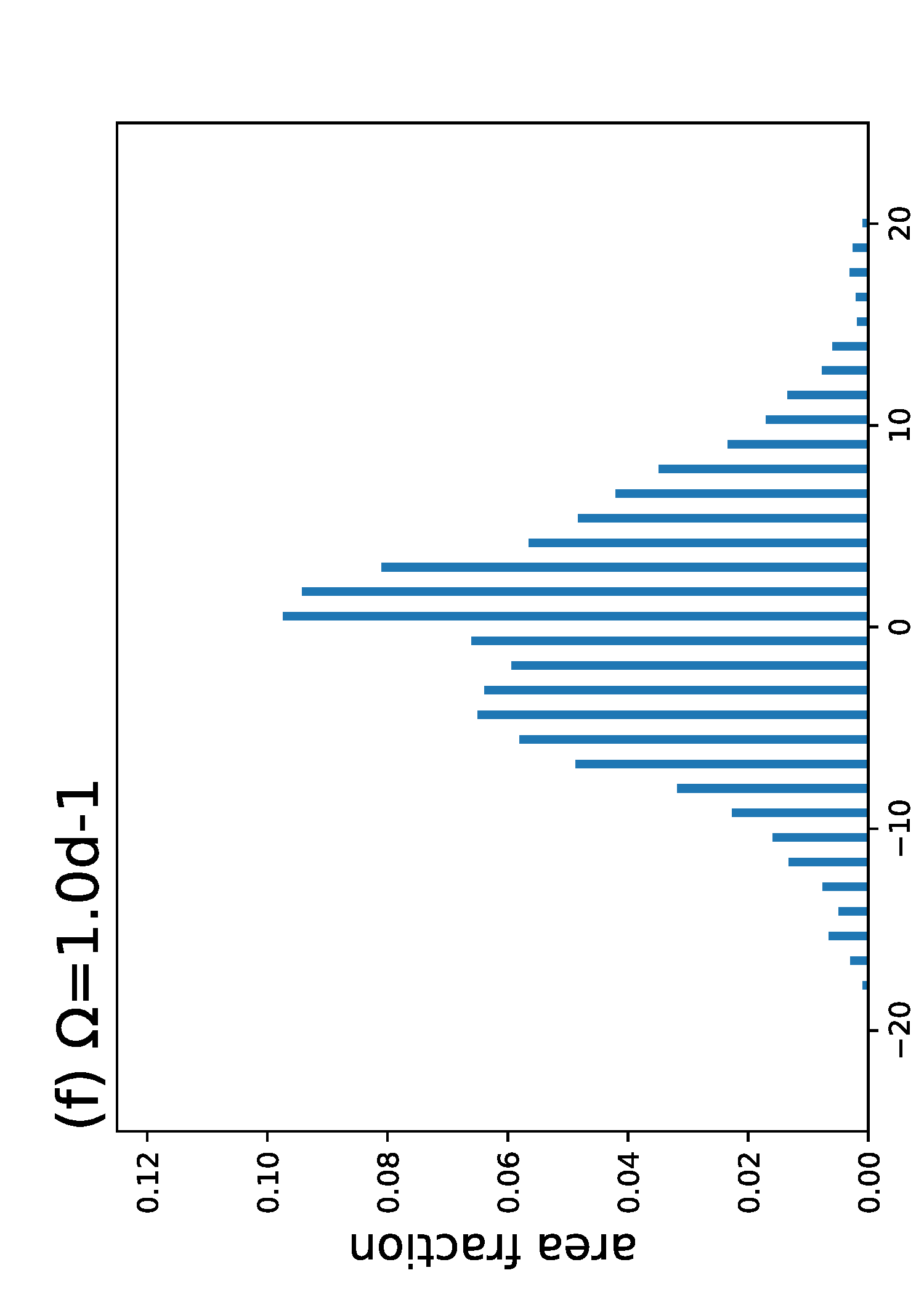}       
            \centering
        \end{minipage} \\

        \begin{minipage}[t]{0.50\hsize}
            \includegraphics[scale=0.32, angle=270]{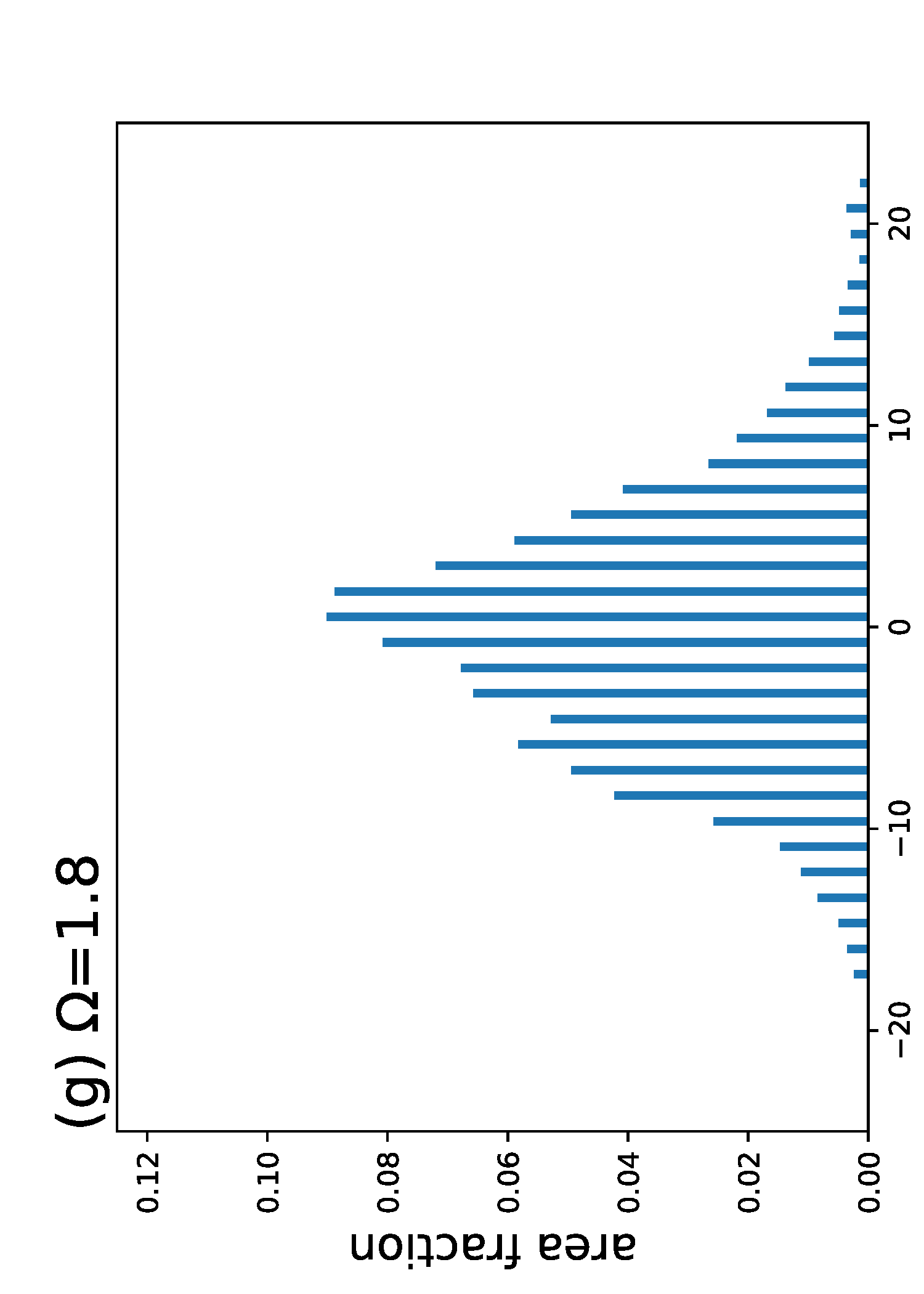}       
            \centering
        \end{minipage}
         &  
        \begin{minipage}[t]{0.50\hsize}
            \includegraphics[scale=0.32, angle=270]{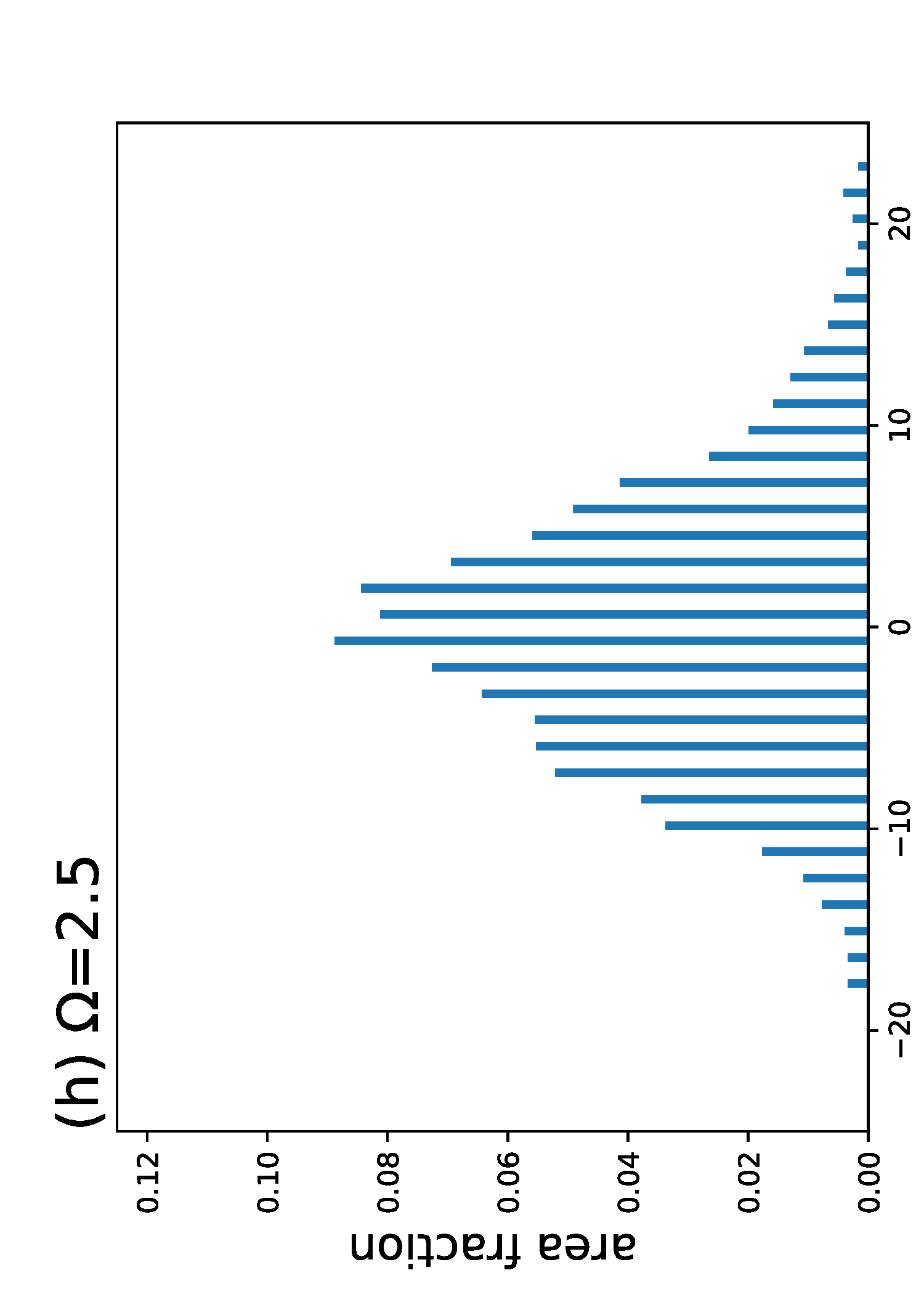}       
            \centering
        \end{minipage} 
    \end{tabular}
    \caption{The areas and vorticity values corresponding to the vorticity patch approximation for eight values of \(\Omega\). In each panel, the corresponding \(\Omega\) value is shown at the top of the panel, the horizontal axis is the value of vorticity patch \(Q_k\), and the vertical axis is the area fraction \(S_k/(4\pi)\) of the vorticity patch.}
    \label{fig_vorticitypatch}
\end{figure}

With the above settings, we search for statistical equilibria by using the gradient method on the energy level surface. The starting point for the search is set to \(Z_{\rm ini}\), as described in section \ref{subsec_patchgenerate}. {This is a natural choice for the starting point, because \(Z_{\rm ini}\) does not have particular structures such as low-wavenumber structure of vorticity field. However, we should note it is difficult to prove that the obtained equilibria are the global maximizers of the mixing entropy because the maximization problem is nonlinear and non-convex.} We also compute zonal statistical equilibria by restricting the search to the subset of zonal vorticity fields: \(P\cap \{Z=(\hat{\xi}_{0,2}, \cdots, \hat{\xi}_{N,N}, \hat{\eta}_{N,N})\in \mathbf{R}^{(N+1)^2-{ 4}}\,|\, \hat{\xi}_{m,n}=\hat{\eta}_{m,n}=0\,(m\neq 0)\}\), to trace the {branch} of {zonal statistical equilibria}. {We note that in MRS theory, the emergence of zonal symmetry generally occurs by changing some parameters like the values of conserved quantities or the inverse temperature parameter. The fundamental fact supporting this is that, if the inverse temperature \(\beta\) is larger than some value \(\beta_c>0\), the statistical equilibrium is unique and it is zonal. {(The inverse temperature \(\beta\) is defined as \(\beta = \partial S_{\rm mix}/\partial E\) for a branch of statistical equilibria, and it can be viewed as the Lagrange multiplier of the energy constraint. Note that in the statistical mechanics of two-dimensional turbulence, \(\beta\) can be negative. \cite{chorin2013vorticity} summarizes the notion of inverse temperature in the context of two-dimensional turbulence.)} In fact, we typically observe such changes of statistical equilibra when we change the value of the conserved energy. Furthermore, a statistical equilibrium obtained by a zonally restricted search is at least a critical point of the mixing entropy even when the search is not restricted.} {The step size of the gradient method we use is from \(1.0\times 10^{-5}\) to \(1.0\times 10^{-3}\), depending on the progress of the search. Furthermore, when the search point comes in the neighborhood of the equilibria, where the norm of the gradient vector is very small, we use the gradient method on a local coordinate of the energy surface as done in \cite{ryono2022new}, and we set the step size to 1.0. However, the step size is not critical for the results unless it is too large. The computations are done on a PC whose CPU was
Intel Xeon W-2135 which had 6 cores of 3.7 GHz, and each search of the statistical equilibria is completed in up to a few tens of hours.}

\section{Results}\label{section_results}
Computed statistical equilibria for the eight initial vorticity fields are shown in figure \ref{equilibria_a-d}. For the initial vorticity field with \(\Omega=0\), which is equivalent to that of \cite{dritschel2015late}, we obtained a statistical equilibrium of quadrupole form, with two positive vortices and two negative ones (figure \ref{equilibria_a-d}(a)). Note that other equivalent equilibria can be obtained due to \(O(3)\) symmetry of the sphere, since the initial angular momentum is zero. For initial vorticity fields with non-zero angular momentum, three types of equilibria are obtained: quadrupole states with two positive vortices fixed at the both poles because of the non-zero angular momentum (for \(\Omega=1.0\times 10^{-2}, 2.5\times 10^{-2}\), and \(5.0\times 10^{-2}\) shown in figure \ref{equilibria_a-d}(b)--\ref{equilibria_a-d}(d), respectively), a wavy state of wavenumber one (for \(\Omega=7.5\times 10^{-2}\) shown in figure \ref{equilibria_a-d}(e)), and zonal states (for \(\Omega=1.0\times 10^{-1}, 1.8\), and \(2.5\) shown in figure \ref{equilibria_a-d}(f)--\ref{equilibria_a-d}(h), respectively). The dependence of the type of equilibrium on the parameter \(\Omega\) is summarized as follows; as the parameter \(\Omega\) increases, the two negative vortices of the quadrupole states become weaker and eventually disappear, leading to a zonally symmetric equilibrium. {Note that, although the above statistical equilibria are at least local maximizers of the mixing entropy, it is difficult in principle to prove that they are also the global maximizers. Nevertheless, it is highly implied that they are also global maximizers, because they are very close to the global maximizers of the MRS--2 problem, whose constraint is looser than our problem, as we see later. }

\begin{figure}[htbp]
    \centering
    \begin{tabular}{cc}
        \begin{minipage}[t]{0.45\hsize}
            \includegraphics[scale=0.30]{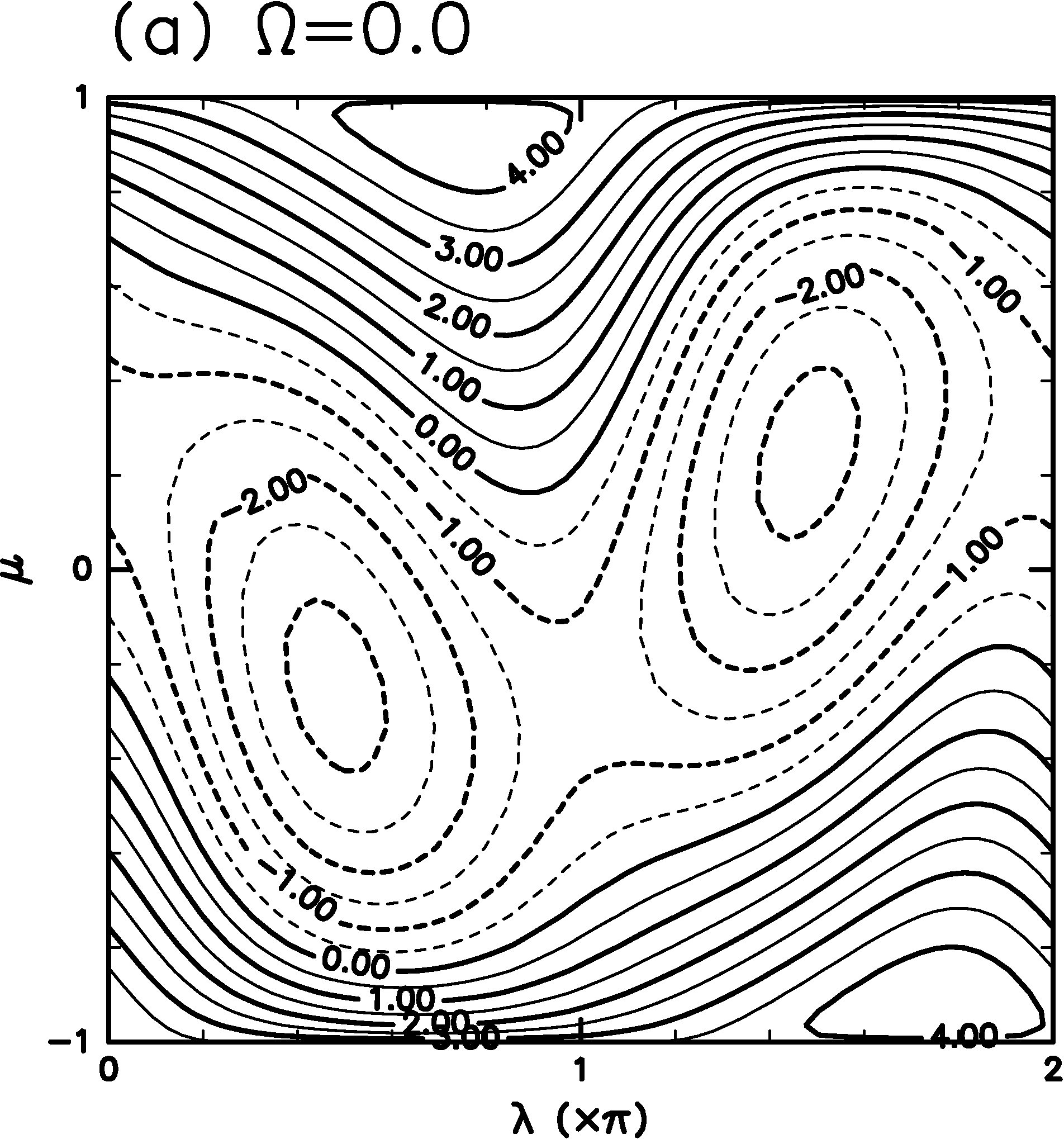}       
            \centering
        \end{minipage}
        &
        \begin{minipage}[t]{0.45\hsize}
            \includegraphics[scale=0.30]{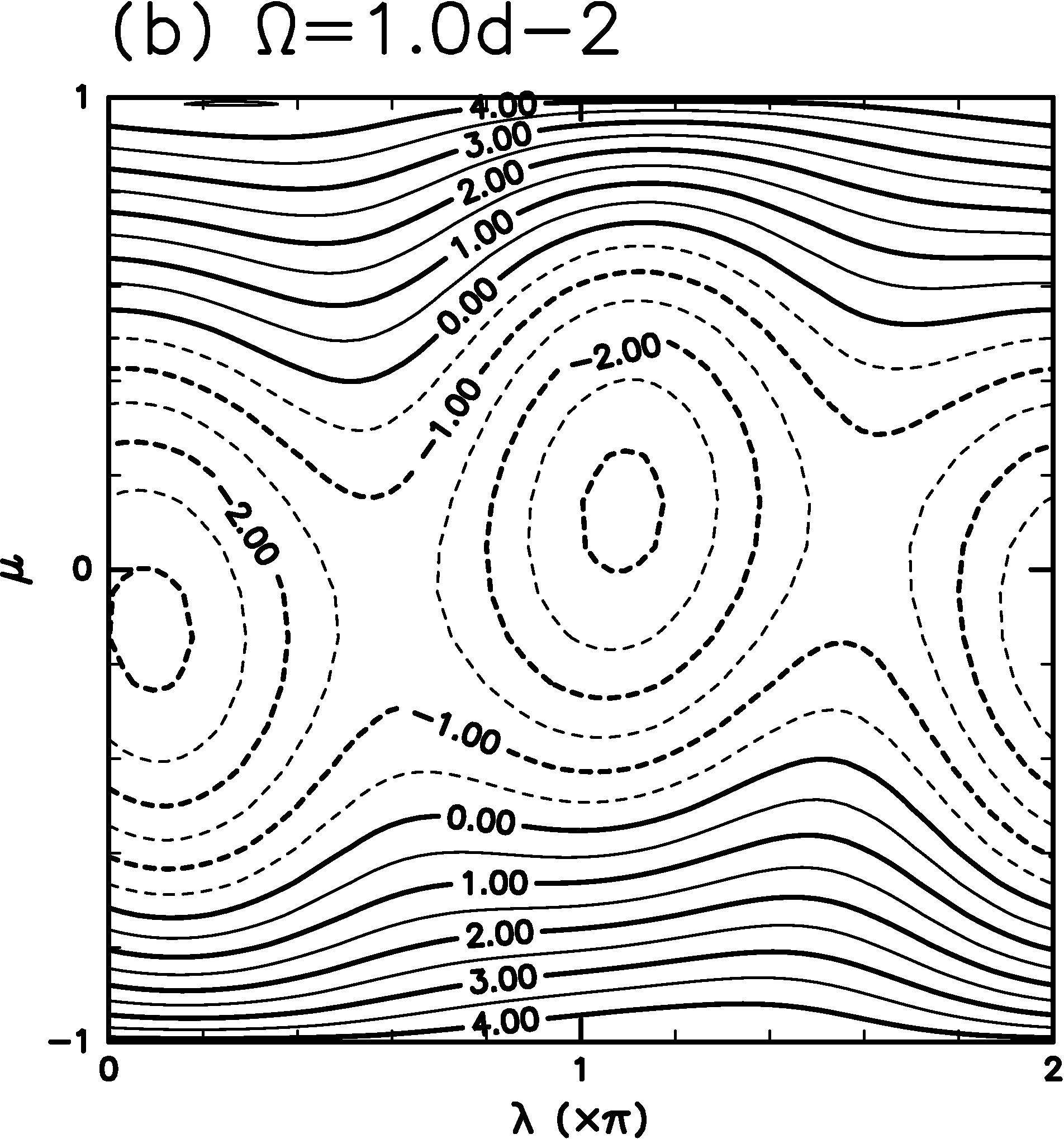}       
            \centering
        \end{minipage}
        \\
        \begin{minipage}[t]{0.45\hsize}
            \includegraphics[scale=0.30]{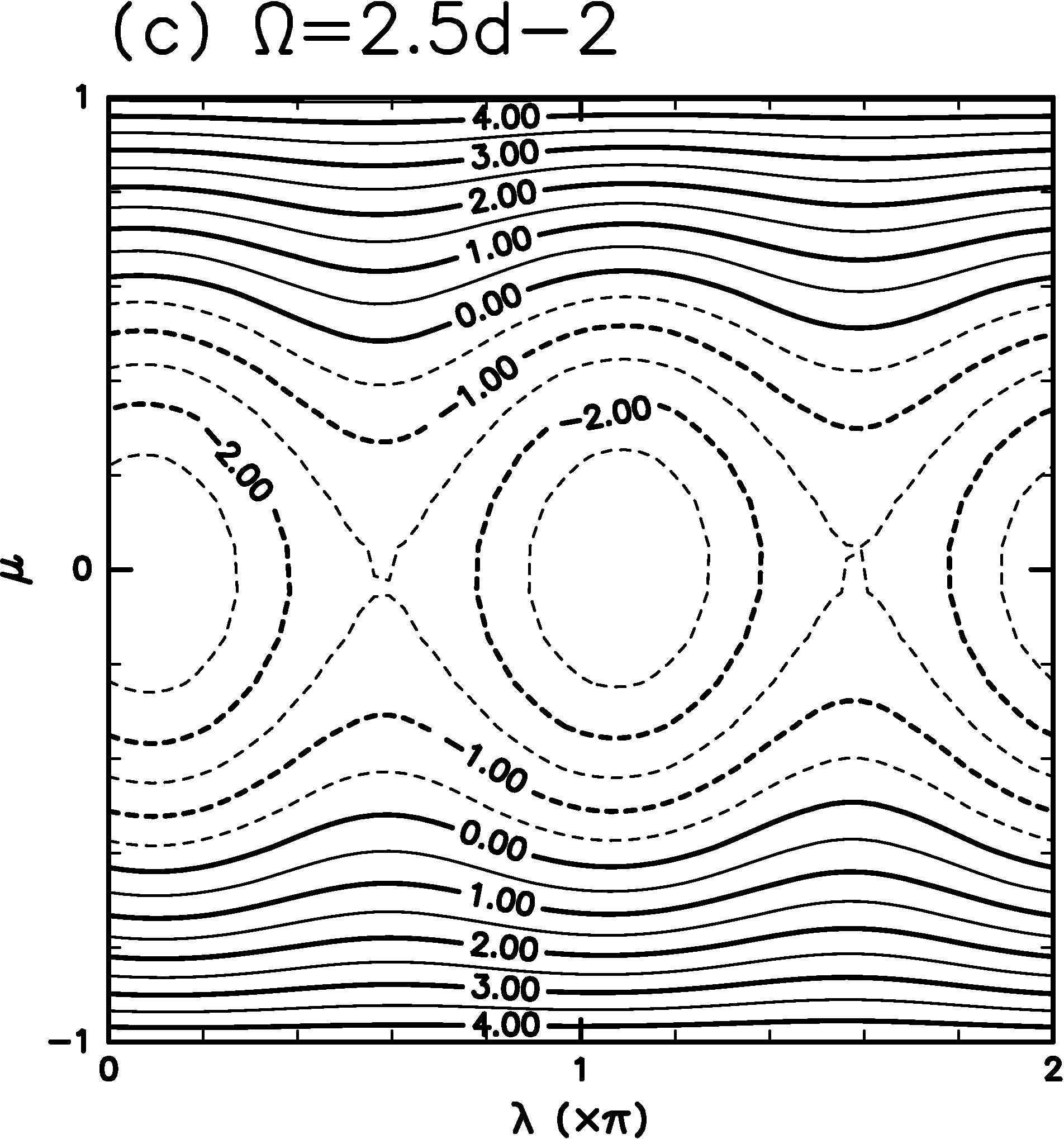}       
            \centering
        \end{minipage}
        &
        \begin{minipage}[t]{0.45\hsize}
            \includegraphics[scale=0.30]{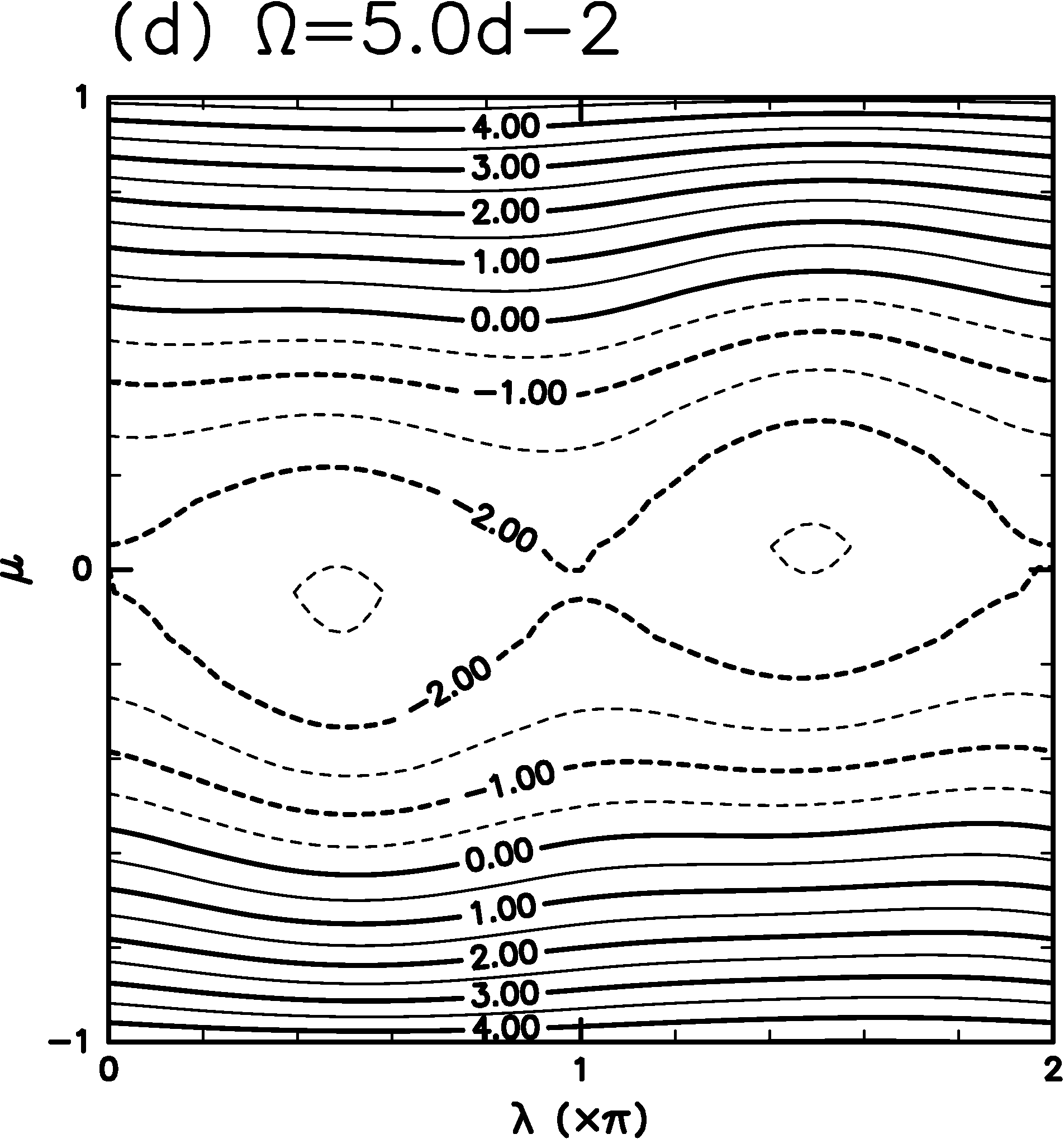}       
            \centering
        \end{minipage}
    \end{tabular}
    \caption{Macroscopic vorticity fields \(\overline{q}\) corresponding to the computed statistical equilibria for \(\Omega=0, 1.0\times 10^{-2}, 2.5\times 10^{-2}, 5.0\times 10^{-2}\) (the cases of \(\Omega=7.5\times 10^{-2}, 1.0\times 10^{-1}, 1.8\), and \(2.5\) are shown in the continued figure). In each panel, the corresponding \(\Omega\) is shown at the top of the panel, the horizontal axis is the longitude \(\lambda/(2\pi)\), and the vertical axis is the sine-latitude \(\mu\). Contour interval is set to 0.5.}
    \label{equilibria_a-d}
\end{figure}
\addtocounter{figure}{-1}

\begin{figure}[htbp]
    \centering
    \begin{tabular}{cc}
        \begin{minipage}[t]{0.40\hsize}
            \includegraphics[scale=0.3]{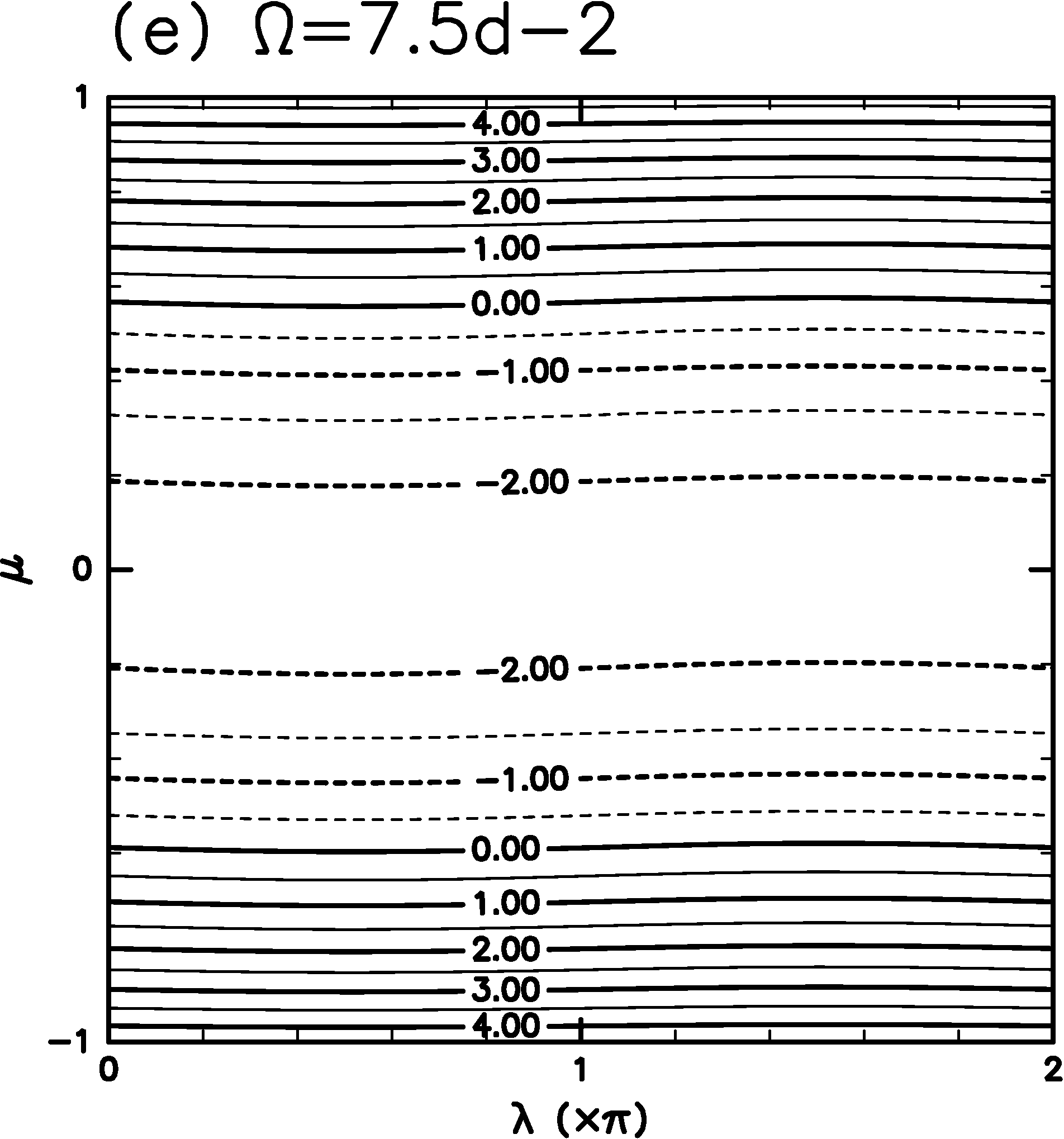}      
            \centering
        \end{minipage}
        &
        \begin{minipage}[t]{0.40\hsize}
            \includegraphics[scale=0.3]{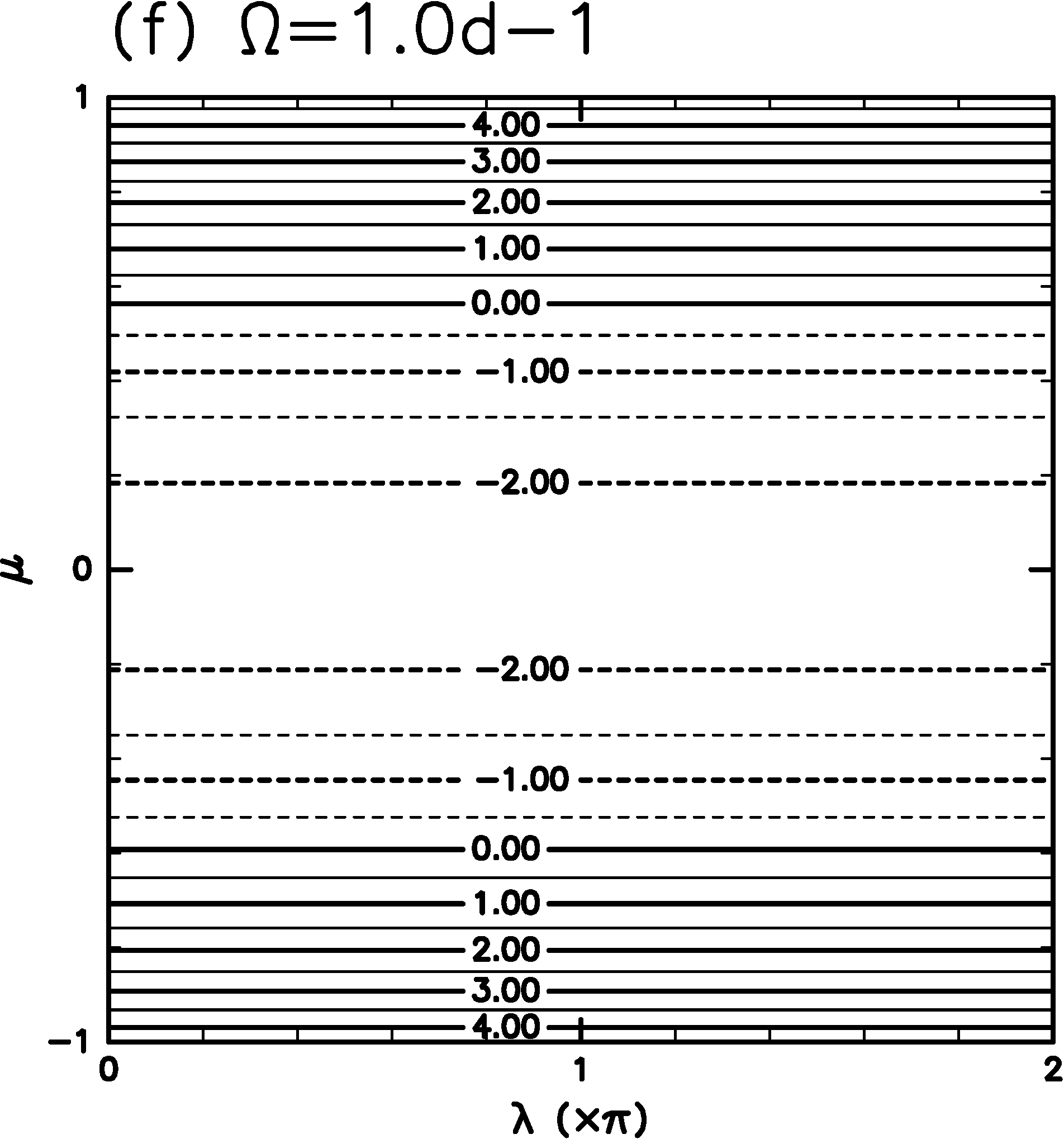}       
            \centering
        \end{minipage}
        \\
        \begin{minipage}[t]{0.40\hsize}
            \includegraphics[scale=0.3]{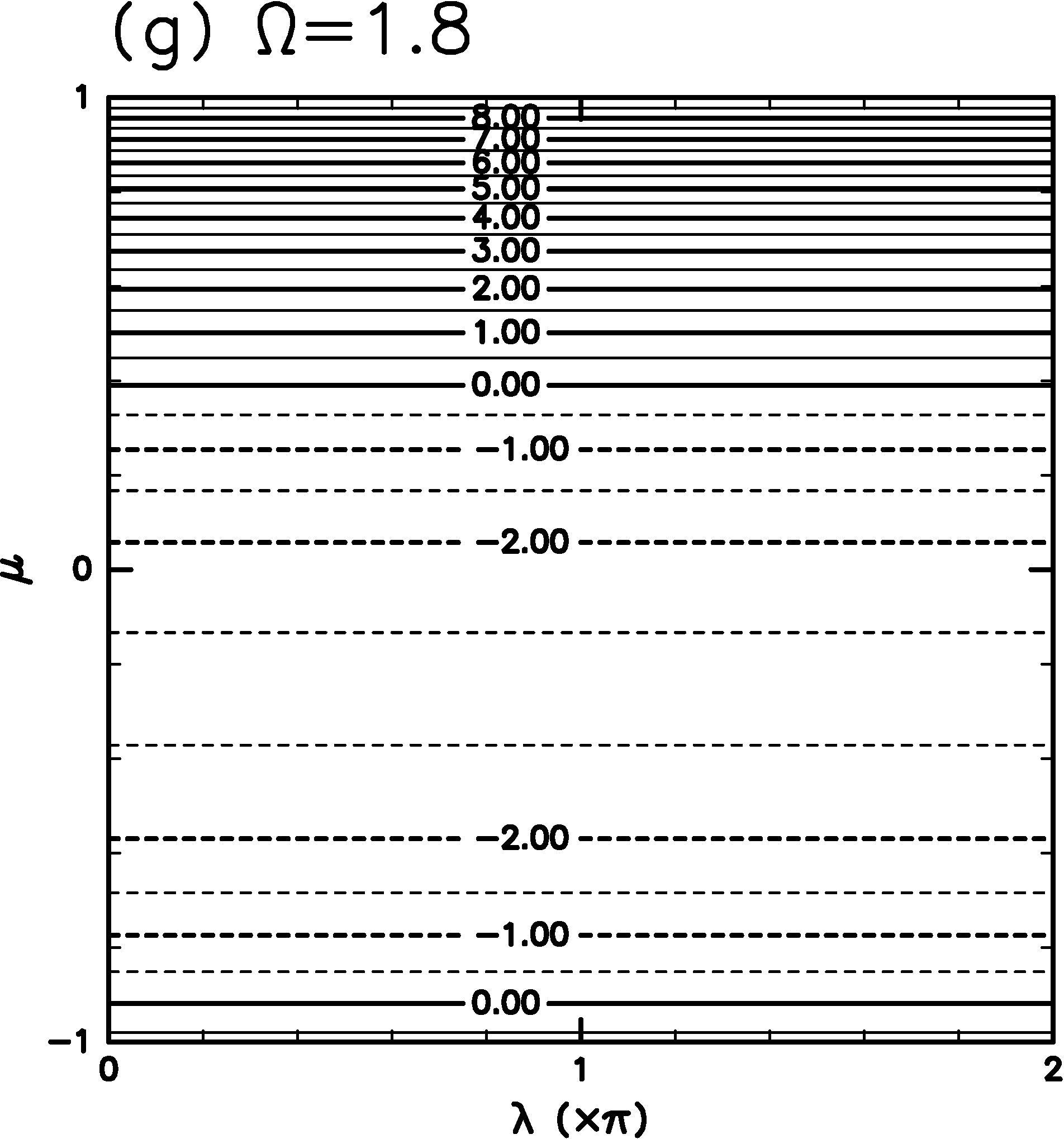}       
            \centering
        \end{minipage}
        &
        \begin{minipage}[t]{0.40\hsize}
            \includegraphics[scale=0.3]{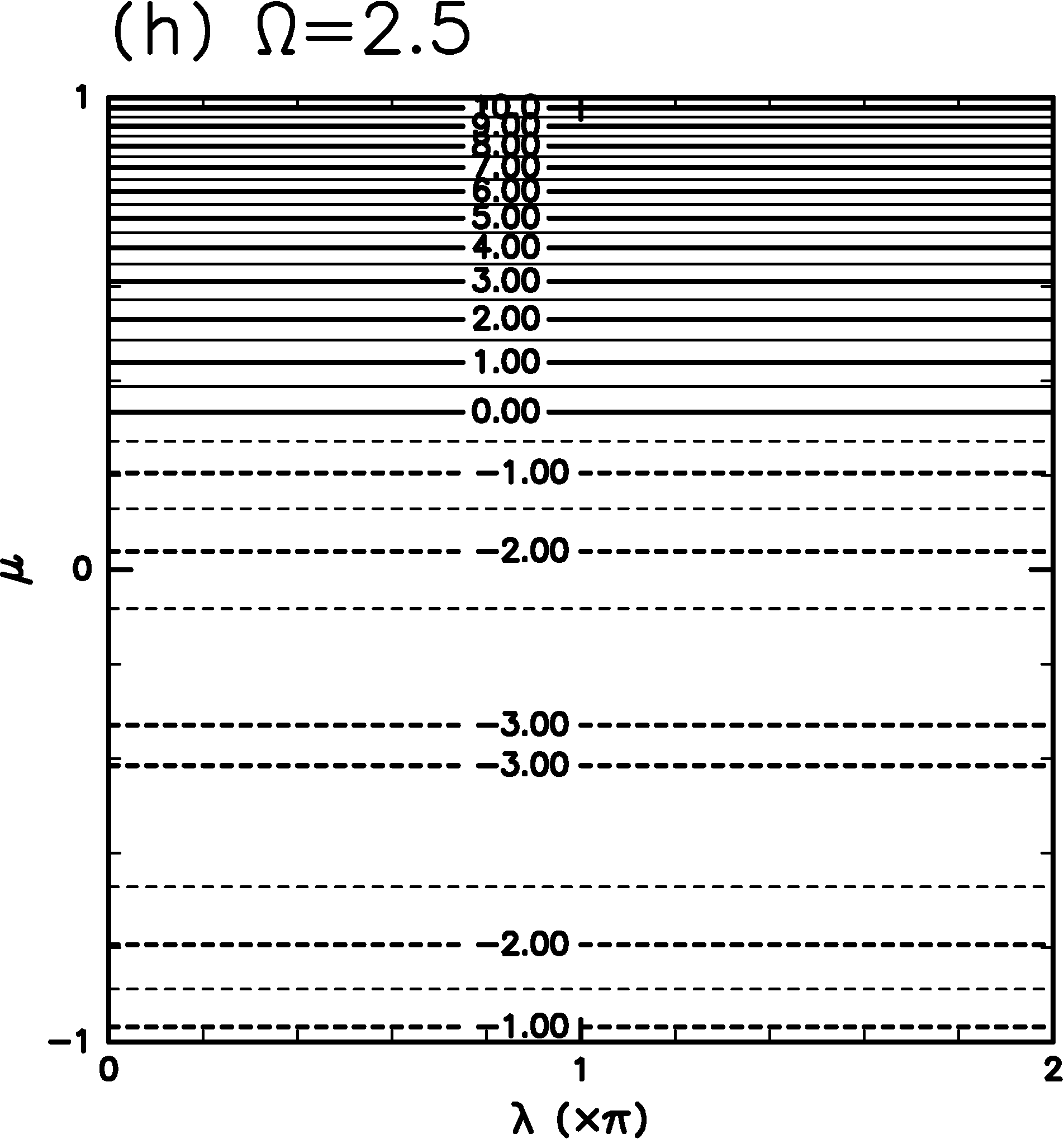}       
            \centering
        \end{minipage}

    \end{tabular}
    \caption{(continued.)}
    \label{equilibria_e-h}
\end{figure}

{For the statistical equilibria computed shown in figure \ref{equilibria_a-d}}, we examine the relationship between the macroscopic vorticity \(\overline{q}\) and the macroscopic stream function \(\overline{\psi}\). Recall that, when the initial angular momentum vector is parallel to the axis connecting the north and south poles, the statistical equilibrium has a functional relationship written as \(\overline{q} = f(\overline{\psi}+\Omega_1 \mu)\) for some constant \(\Omega_1\) (\citealp{qi2014hyperviscosity}, \citealp{ryono2022new}). This relationship means that the statistical equilibrium is a steady solution of the Euler equation for an observer rotating at the angular velocity \(\Omega_1\). 
{Indeed, if we introduce a new rotating frame of reference by \(\lambda'=\lambda-\Omega_1 t; \mu'=\mu; t'=t\), then the Euler equation on this frame of reference becomes}
\begin{align*}
    {
    \frac{\partial q}{\partial t'} - \Omega_1 \frac{\partial q}{\partial\lambda'} + \frac{\partial\psi}{\partial\lambda'}\frac{\partial q}{\partial \mu'}-\frac{\partial q}{\partial\lambda'}\frac{\partial \psi}{\partial \mu'}=0.
    }
\end{align*}
{Substituting \(q=\overline{q}=f(\overline{\psi}+\Omega_1\mu)=f(\overline{\psi} +\Omega_1\mu')\) and \(\psi = \overline{\psi}\) into the above equation, we obtain \(\partial \overline{q}/\partial t'=0\). Therefore, the statistical equilibrium looks steady for an observer on the rotating frame of reference.}
For the quadrupole statistical equilibria we computed, we plot the modified stream function \(\overline{\psi}+\Omega_1 \mu\) versus the macroscopic vorticity \(\overline{q}\) (figure \ref{qpsirel_a-d}). Here, the values of \(\Omega_1\) are determined by the least squares method introduced by \cite{ryono2022new} {for non-zonal equilibria (figure \ref{qpsirel_a-d}(a)-(e))}. {However, the method cannot be applied for zonal equilibria (figure \ref{qpsirel_a-d}(f)-(h)) because it needs non-zero \(L^2\) norm of \(\partial \overline{q}/\partial \lambda\), and therefore we determine the values of \(\Omega_1\) by another method described in \ref{sec_appendix_omega1}.} For all the cases, the functional relationships are almost linear. Linear \(\overline{q}\)--\(\overline{\psi}\) relation can be derived from the MRS--2 theory, which considers the conservation of Casimir invariants only up to the second order (\citealp{naso2010statistical}, \citealp{herbert2013additional}, etc). Note that, however, not all energy condenses to total wavenumber two modes of the vorticity field completely as predicted by the MRS--2. For the \(\Omega=0\) case in this paper, the total wavenumber four modes have about \(3.5850\times 10^{-6}\) times the energy of the total wavenumber two modes. This ratio increases as the value of \(\Omega\) increases{, which realizes as the subtle curvature of \(\overline{q}\)-\(\overline{\psi}\) curve in figure \ref{qpsirel_a-d}(g) and (h)}. For example, in the case of \(\Omega=1.8\), the total wavenumber four modes have about \(7.4750\times 10^{-4}\) times the energy of the total wavenumber two modes. {Although the small values may not seem physically significant, we address that this is highly a non-trivial result from the mathematical point of view of the entropy maximization problem, since we do not truncate any Casimir invariants as in MRS--2 and there is no a priori guarantee that the amplitudes of modes with large total wave numbers are small. Furthermore, the condensation of the energy into \(n=2\) modes leads to the dynamical stability of the statistical equilibria. Indeed, if there were an instability from the statistical equilibrium, then the enstrophy would be transported to the modes with higher total wavenumbers. According to the well known discussion of \cite{fjortoft1953changes}, such transport must be accompanied with an energy transport to modes with lower wavenumbers. This is impossible since the modes with total wavenumber \(n=1\) are fixed due to the angular momentum conservation.}

\begin{figure}[htbp]
    \centering
    \begin{tabular}{cc}
        \begin{minipage}[t]{0.45\hsize}
            \includegraphics[scale=0.21]{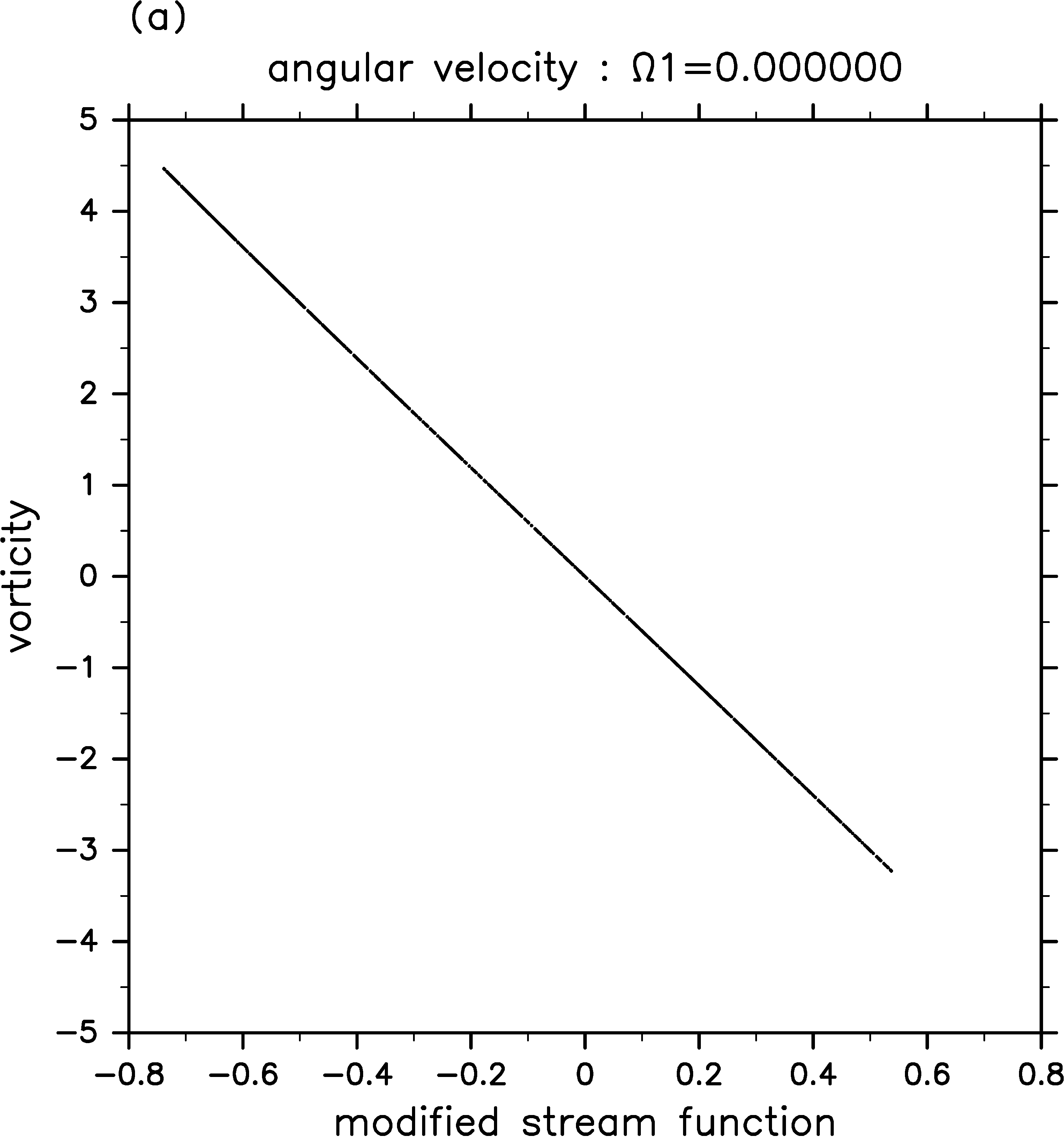}       
            \centering
        \end{minipage}
        &
        \begin{minipage}[t]{0.45\hsize}
            \includegraphics[scale=0.21]{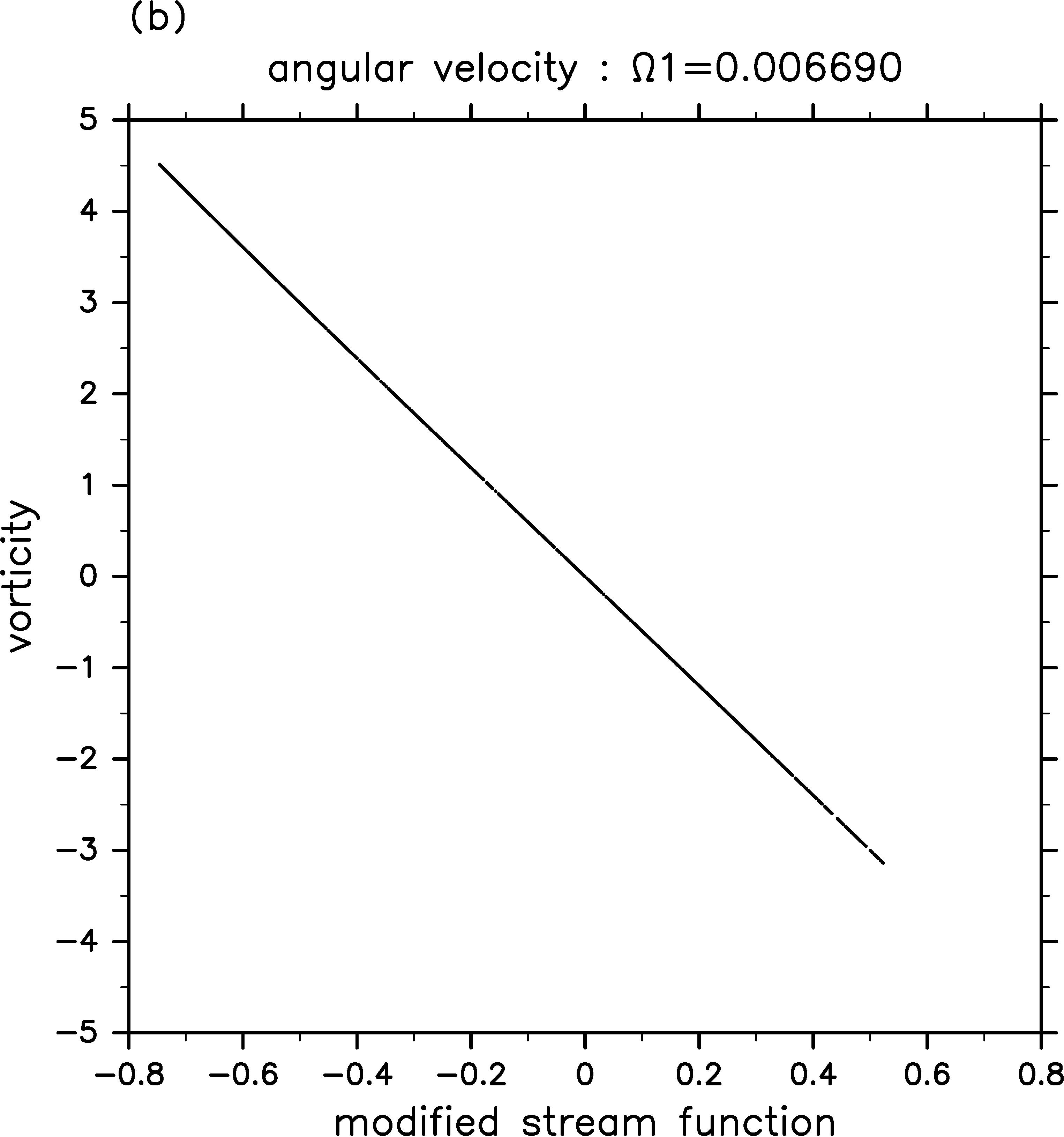}       
            \centering
        \end{minipage}
        \\
        \begin{minipage}[t]{0.45\hsize}
            \includegraphics[scale=0.21]{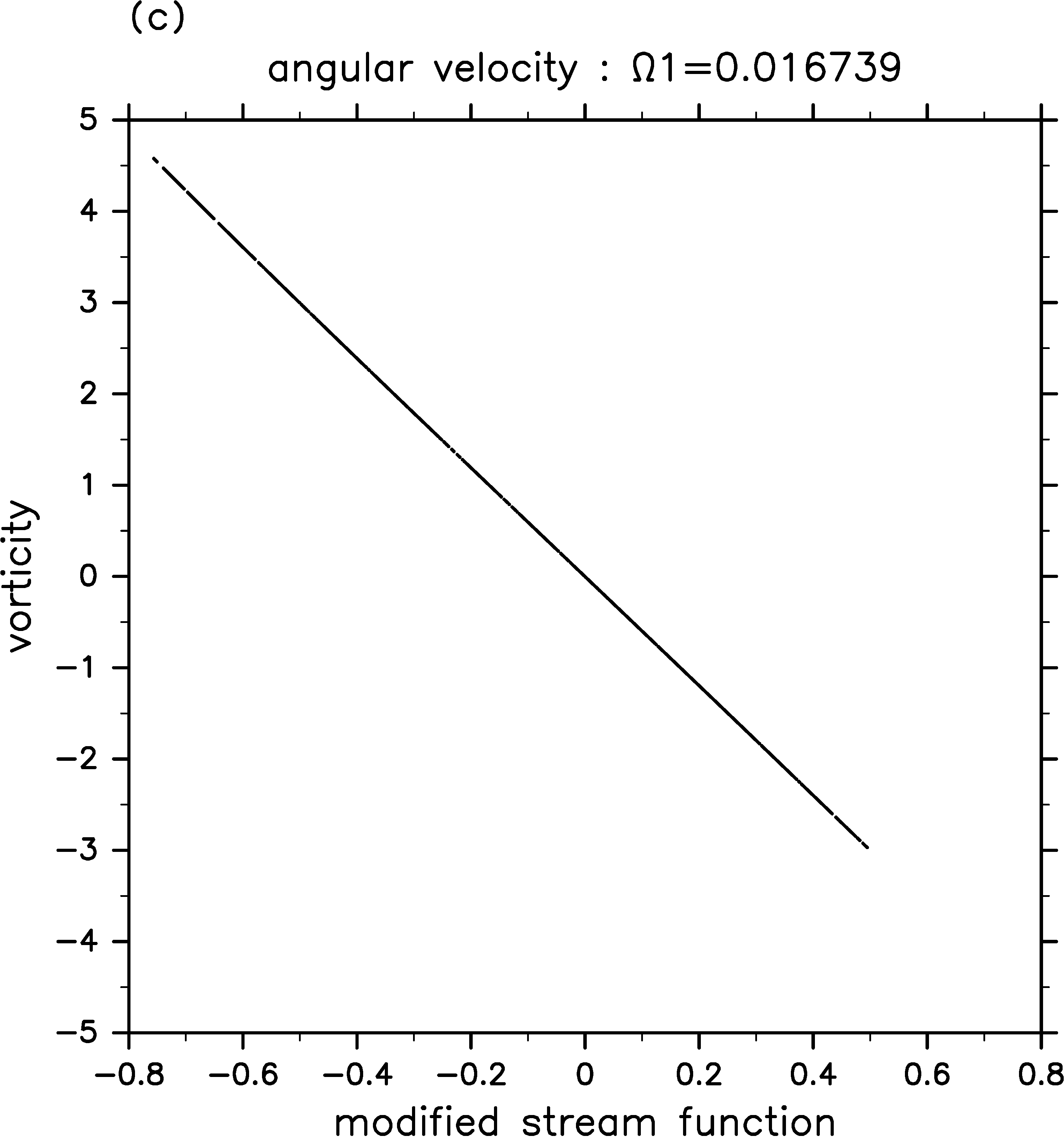}       
            \centering
        \end{minipage}
        &
        \begin{minipage}[t]{0.45\hsize}
            \includegraphics[scale=0.21]{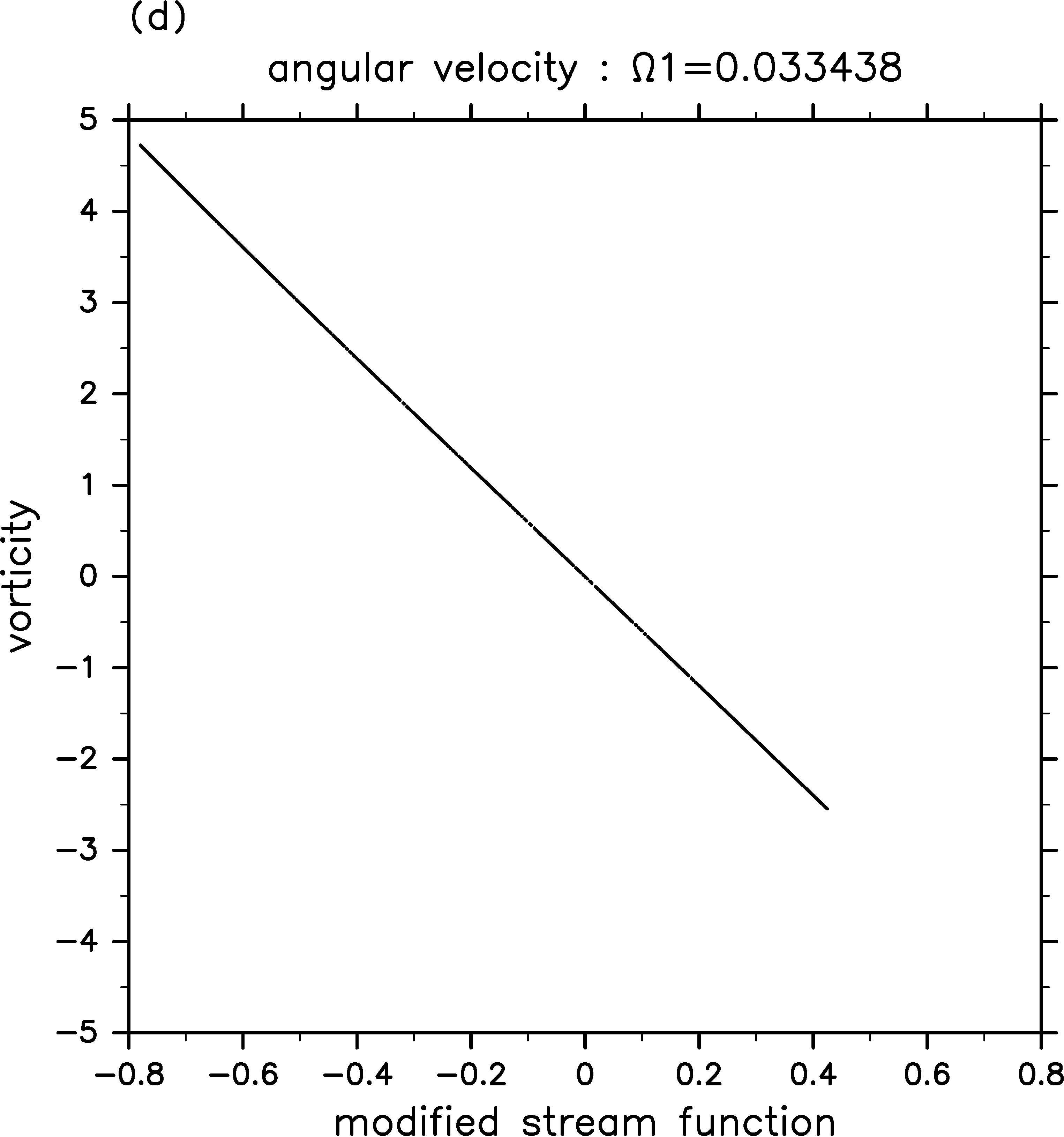}       
            \centering
        \end{minipage}
    \end{tabular}
    \caption{Plots of modified stream function \(\overline{\psi}+\Omega_1 \mu\) versus macroscopic vorticity \(\overline{q}\) for the initial vorticity fields corresponding to {\(\Omega = 0, 1.0\times 10^{-2}, 2.5\times 10^{-2}\), and \(5.0\times 10^{-2}\) (the cases for \(\Omega=7.5\times 10^{-2}, 1.0\times 10^{-1}, 1.8\), and \(2.5\) are shown in the continued figure)}. In each panel, the value of \(\Omega_1\) is shown at the top of the panel, the horizontal axis is the modified stream function \(\overline{\psi}+\Omega_1 \mu\), and the vertical axis is the macroscopic vorticity \(\overline{q}\).}
    \label{qpsirel_a-d}
\end{figure}
\addtocounter{figure}{-1}
\begin{figure}[htbp]
    \centering
    \begin{tabular}{cc}
        \begin{minipage}[t]{0.45\hsize}
            \includegraphics[scale=0.21]{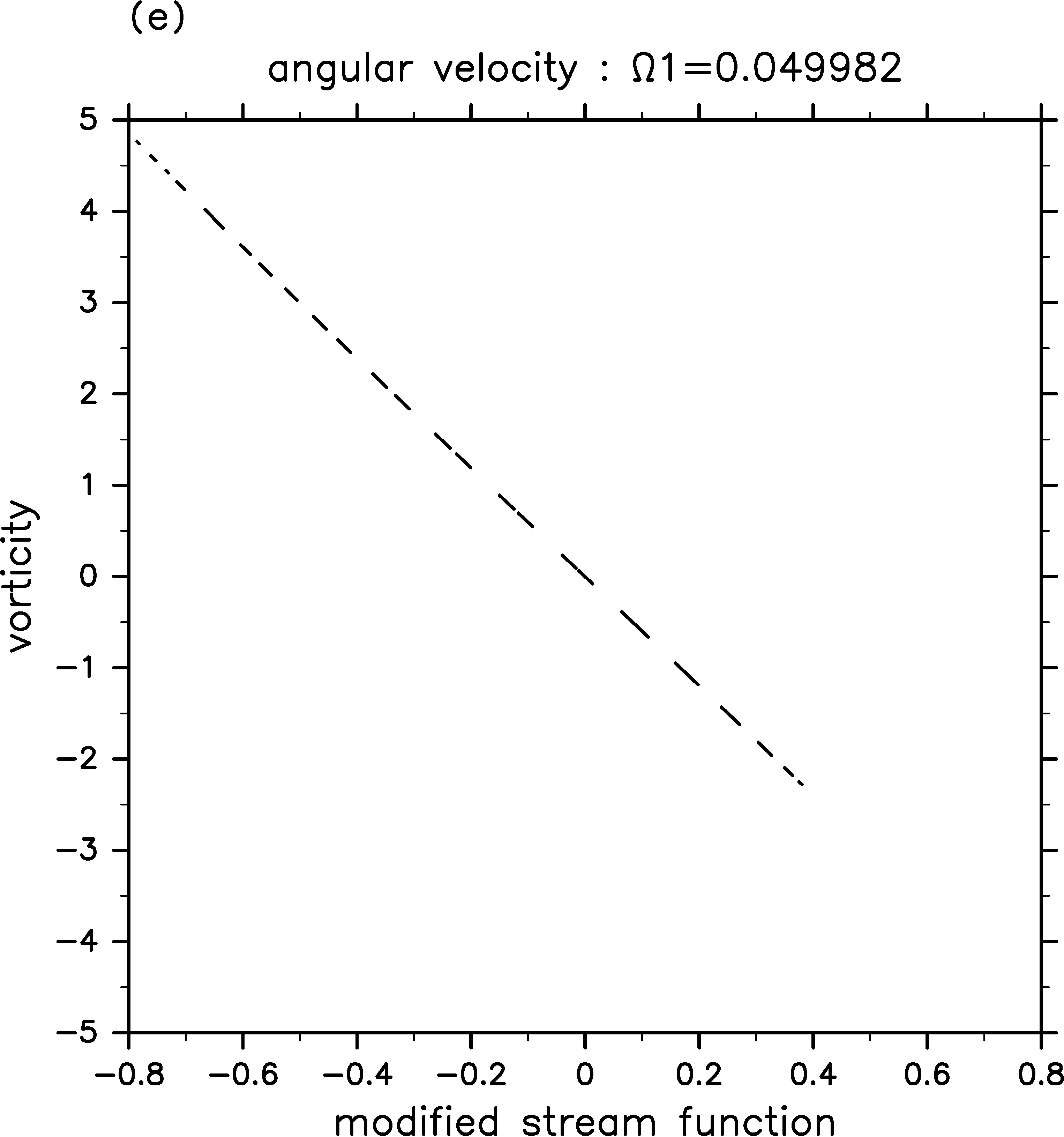}       
            \centering
        \end{minipage}
        &
        \begin{minipage}[t]{0.45\hsize}
            \includegraphics[scale=0.21]{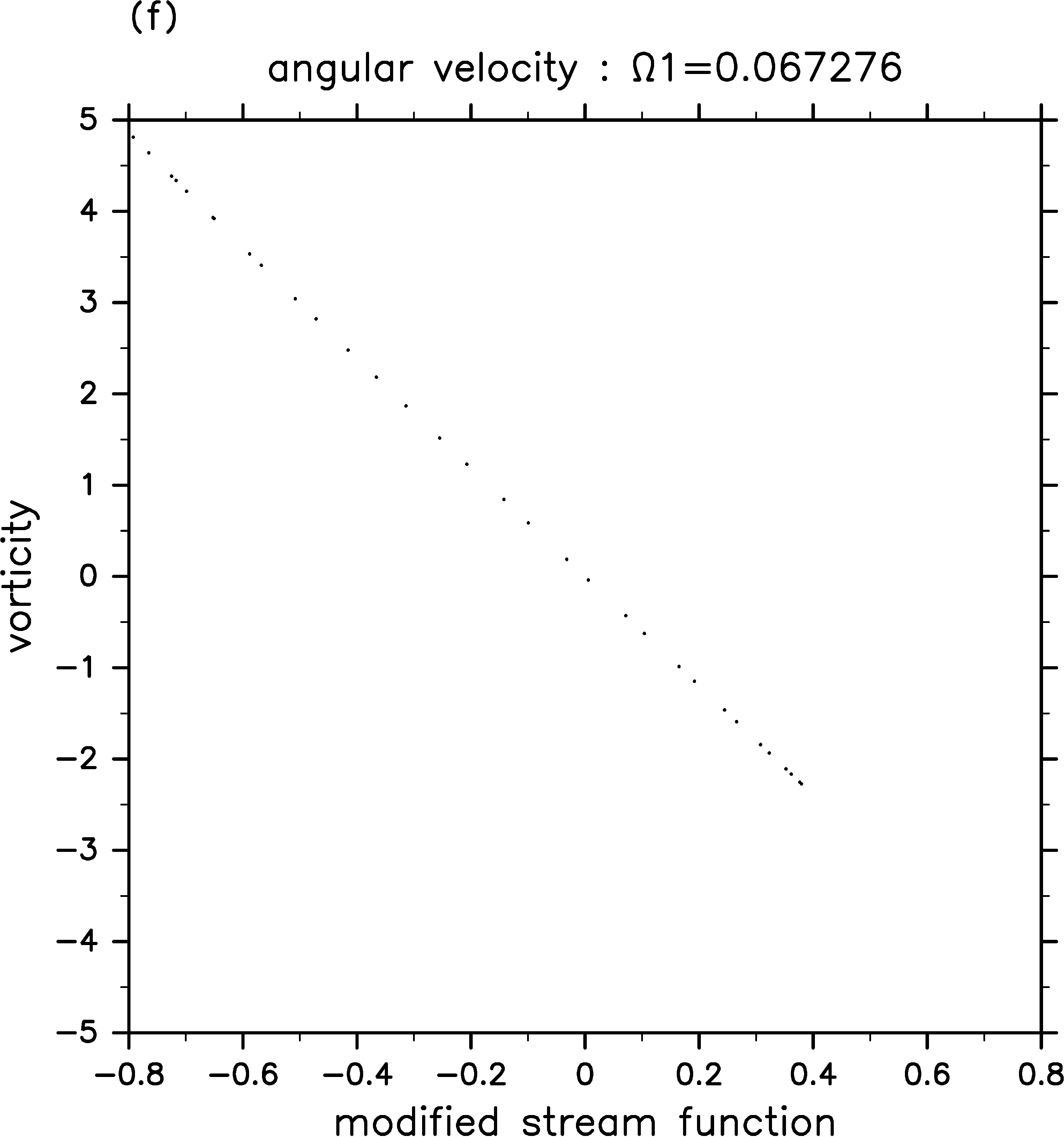}       
            \centering
        \end{minipage}
        \\
        \begin{minipage}[t]{0.45\hsize}
            \includegraphics[scale=0.21]{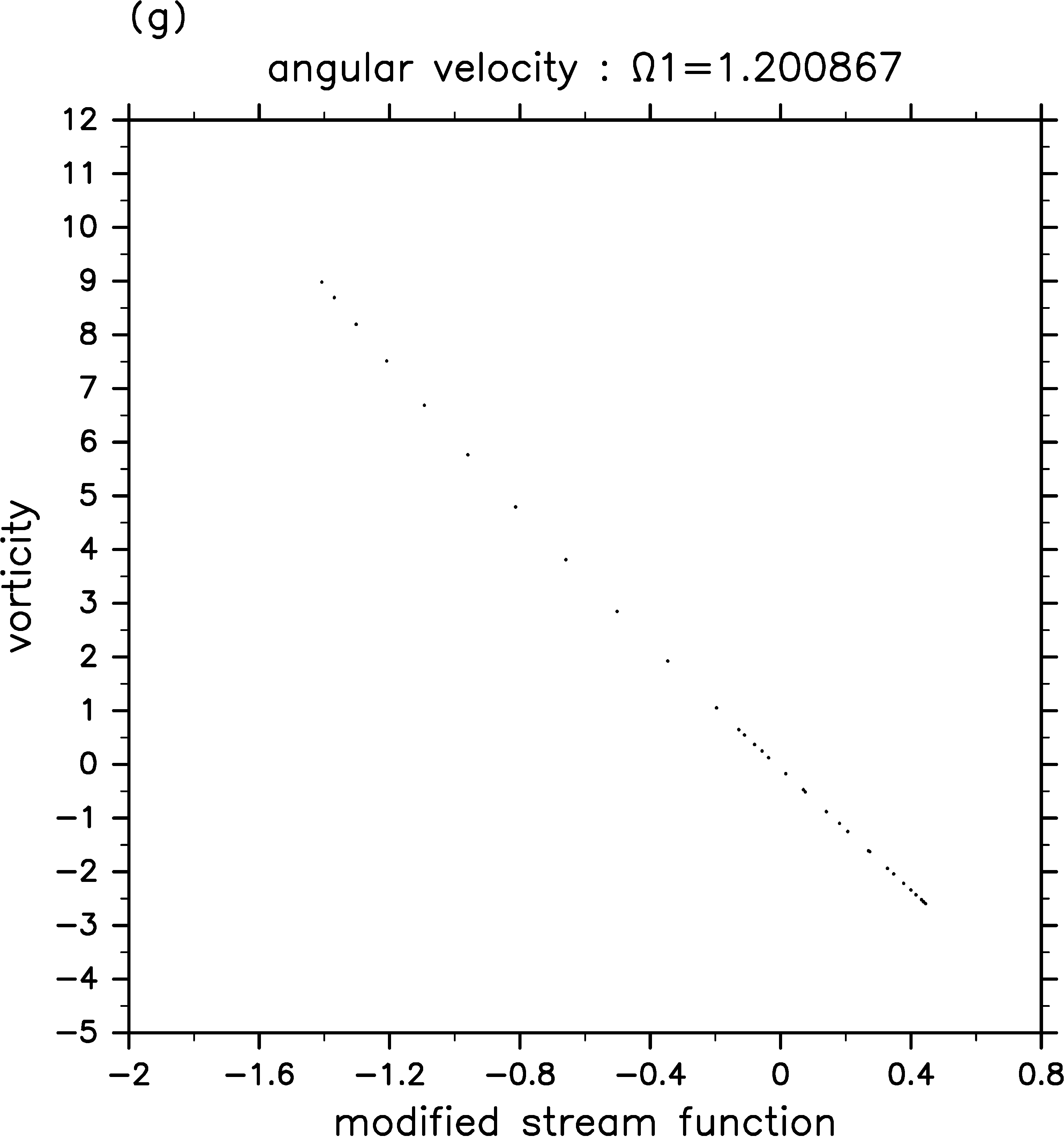}       
            \centering
        \end{minipage}
        &
        \begin{minipage}[t]{0.45\hsize}
            \includegraphics[scale=0.21]{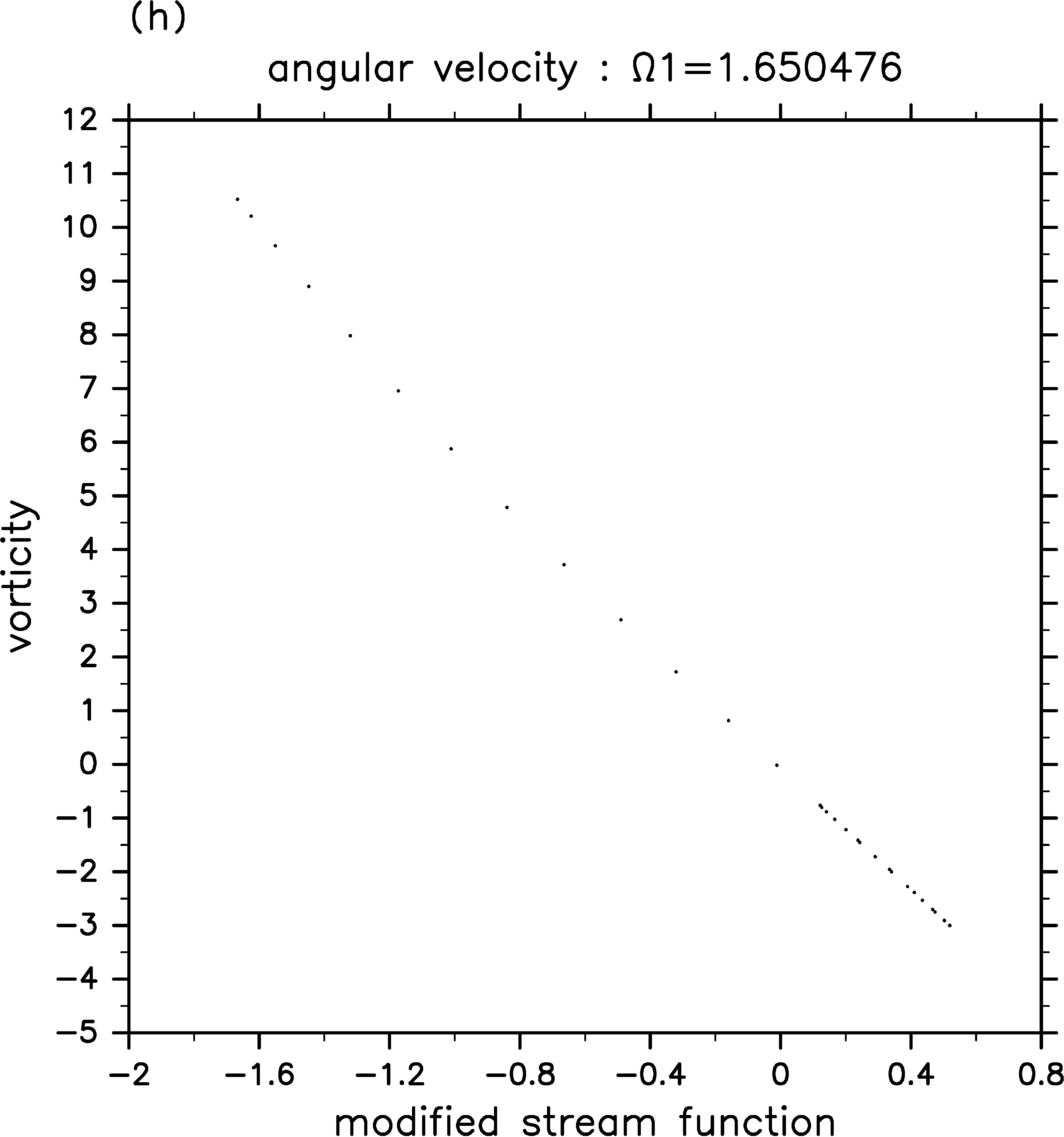}       
            \centering
        \end{minipage}
    \end{tabular}
    \caption{{(Continued.)}}
    \label{qpsirel_g-h}
\end{figure}

Now that we have described the properties of the computed statistical equilibria, let us {compare} these statistical equilibria with the time evolution results (shown in section \ref{subsection_timeevo}). For the initial vorticity field of \cite{dritschel2015late} (\(\Omega=0\) case), the MRS theory predicts a quadrupole state (figure \ref{equilibria_a-d}(a)) as the statistical equilibrium. In this case, a quadrupole state is also obtained as the final state of time evolution (figure \ref{t200_a-d}(a)), so although the shape of the vorticity distribution itself is different, it is topologically consistent in the sense of a quadrupole. The existence of quadrupole statistical equilibria was also noted by \cite{qi2014hyperviscosity}, where the MRS--2 and the perturbative MRS--4 theory were used. Note that, however, these theories do not necessarily predict a quadrupole state as a unique equilibrium. The present study provides the first rigorous confirmation that the quadrupole state is {at least a local maximum of the mixing entropy,} strictly based on the original MRS theory. On the other hand, we also confirmed that the MRS theory did not predict the concentrated vortex cores which appeared in the time evolution result (compare figure \ref{t200_a-d}(a) and \ref{equilibria_a-d}(a)). Indeed, the statistical equilibrium had a maximum \(\overline{q}\) value of \(4.4668\) and a minimum value of \(-3.2267\) (figure \ref{equilibria_a-d}(a)), while the vorticity field \(q\) at the final state of the time evolution (figure \ref{t200_a-d}(a)) had a maximum value of \(19.836\) and a minimum value of \(-16.283\). The vortices that appeared in the time evolution had concentrated cores of vorticity where large vorticity was confined to a small radius, in contrast to the spread vortices in the statistical equilibrium. For the cases of \(\Omega= 1.0\times 10^{-2}, 2.5\times 10^{-2}\), and \(5.0\times 10^{-2}\), we obtained the statistical equilibria of quadrupole states (figure \ref{equilibria_a-d}(b), \ref{equilibria_a-d}(c), and \ref{equilibria_a-d}(d), respectively), which were topologically consistent with the time integration results (figure \ref{t200_a-d}(b), \ref{t200_a-d}(c), and \ref{t200_a-d}(d)), but note that two positive vortices are not necessarily fixed at poles in the time integration results. For the case of \(\Omega=7.5\times 10^{-2}\), the macroscopic vorticity field \(\overline{q}\) corresponding to the statistical equilibrium (figure \ref{equilibria_a-d}(e)) was almost zonal, but with subtle wavenumber one distortion, whereas the time integration result (figure \ref{t200_a-d}(e)) showed a quadrupole state. For the cases of \(\Omega=1.0\times 10^{-1}, 1.8\) and \(2.5\), the macroscopic vorticity fields \(\overline{q}\) corresponding to the statistical equilibria (figure \ref{equilibria_a-d}(f), \ref{equilibria_a-d}(g), and \ref{equilibria_a-d}(h)) were zonal, but the time integration results (figure \ref{t200_a-d}(f), \ref{t200_a-d}(g), and \ref{t200_a-d}(h)) were non-zonal. In each case of \(\Omega=1.0\times 10^{-1}, 1.8\) and \(2.5\), negative vorticity is distributed in low-latitude area in the statistical equilibrium, but, for the time integration result a negative vortex roams around the south pole. The reason why the statistical equilibria and the time integration results had such differences is discussed in the next section.

 {As we have seen so far}, the type of equilibrium changes from non-zonal to zonally symmetric as the parameter \(\Omega\) {become larger than \(7.5\times 10^{-2}\)}. However, we can trace the branch of the zonally symmetric equilibria to small{er} values of \(\Omega\) by restricting the search to the subset of zonal vorticity fields as described at the end of section \ref{subsec_setting_computation}. The values of \(S_{\rm mix}\) for the zonal and non-zonal statistical equilibria are shown in Table \ref{tab:entropy_values}. {Here, we have confirmed that the values of the mixing entropy converged and were accurate up to 13 digits.} For the cases of \(\Omega=1.0\times 10^{-2}, 2.5\times 10^{-2}\), and \(5.0\times 10^{-2}\), the non-zonal (quadrupole) equilibria have higher entropy values than the zonal equilibria, but for the case of \(\Omega=7.5\times 10^{-2}\), the zonal equilibrium has a {slightly} higher {but almost equal} entropy value {to} the non-zonal equilibrium, {which implies that the loss of zonal asymmetry of the equilibrium occurs when \(\Omega\) is slightly larger than \(7.5\times 10^{-2}\). We can confirm this transition by comparing the macroscopic vorticity fields. Figure \ref{zonal_equilibria} shows the zonally symmetric equilibria obtained in this way for \(\Omega=0, 1.0\times 10^{-2}, 2.5\times 10^{-2}, 5.0\times 10^{-2}\), and \(7.5\times 10^{-2}\) cases. For the case of \(\Omega = 7.5\times 10^{-2}\), the zonal and non-zonal equilibria appear almost identical (figure \ref{equilibria_a-d}(e) and \ref{zonal_equilibria}(e))}. {For \(\Omega=0.0, 1.0\times 10^{-2}, 2.5\times 10^{-2}\), and \(5.0\times 10^{-2}\), the difference in entropy values are still small between zonal and non-zonal equilibria, although the appearance of the macroscopic vorticity fields is quite different. However, we note that however small the difference in the mixing entropy of two macroscopic states, there is a large difference in the number of corresponding microscopic states. The concentration theorem of \cite{robert1991maximum} states that, heuristically, the number of microscopic states in the neighborhood of the entropy maximizing state is huge compared to the number of microscopic states corresponding to the other macroscopic states. We discuss this point more quantitatively in the next section.}

\begin{table}[htbp]
    \centering
    \caption{Values of mixing entropy \(S_{\rm mix}\) for the zonal and non-zonal statistical equilibria computed for the initial vorticity fields corresponding to the eight values of \(\Omega\). The values are shown to 12 decimal places. {For \(\Omega=1.0\times 10^{-1}, 1.8\), and \(2.5\), the values for non-zonal statistical equilibria are not available because such equilibria are not found.}}
    \begin{tabular}{@{}ccc}
          \hline
          \(\Omega\)   &  zonal & non-zonal \\ 
          \hline\hline
         \(0.0\) & 2.949248609918 & 2.949253274176 \\
         \(1.0\times 10^{-2}\) & 2.949334821280 & 2.949337898743 \\
         \(2.5\times 10^{-2}\) & 2.949483505214 & 2.949484819565 \\
         \(5.0\times 10^{-2}\) & 2.949824780187 & 2.949824800570 \\
         \(7.5\times 10^{-2}\) & 2.950228945964 & 2.950228945819 \\
         \(1.0\times 10^{-1}\) & 2.950667559451 & not available \\
         \(1.8\) & 2.913956688301 & not available \\
         \(2.5\) & 2.881885872494 & not available \\ 
         \hline
    \end{tabular}
    \label{tab:entropy_values}
\end{table}

\begin{figure}[htbp]
    \centering
    \begin{tabular}{cc}
        \begin{minipage}[t]{0.40\hsize}
            \includegraphics[scale=0.3]{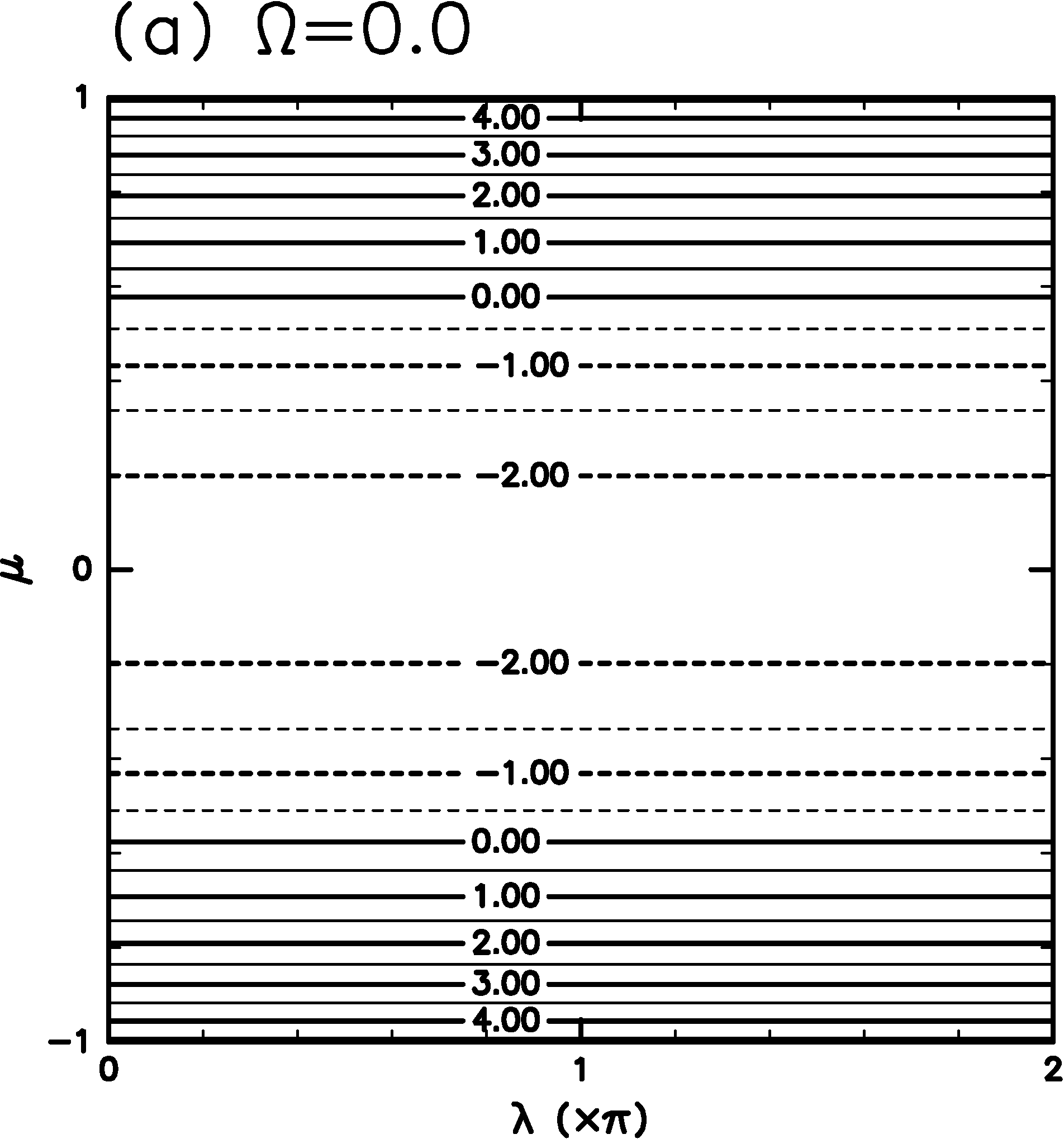}   
            \centering
        \end{minipage}
        &
        \begin{minipage}[t]{0.40\hsize}
            \includegraphics[scale=0.3]{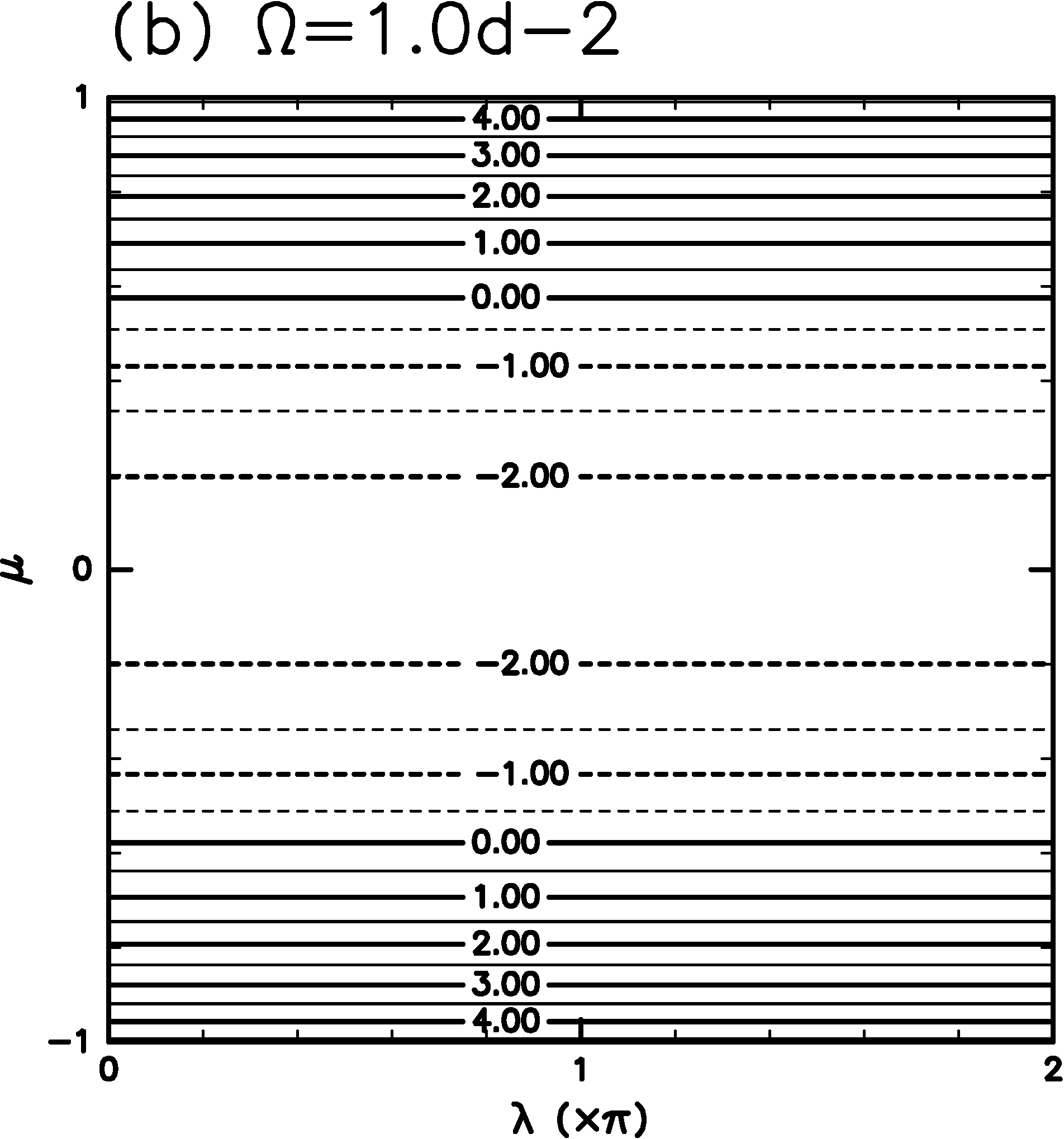}  
            \centering
        \end{minipage}
        \\
        \begin{minipage}[t]{0.40\hsize}
            \includegraphics[scale=0.3]{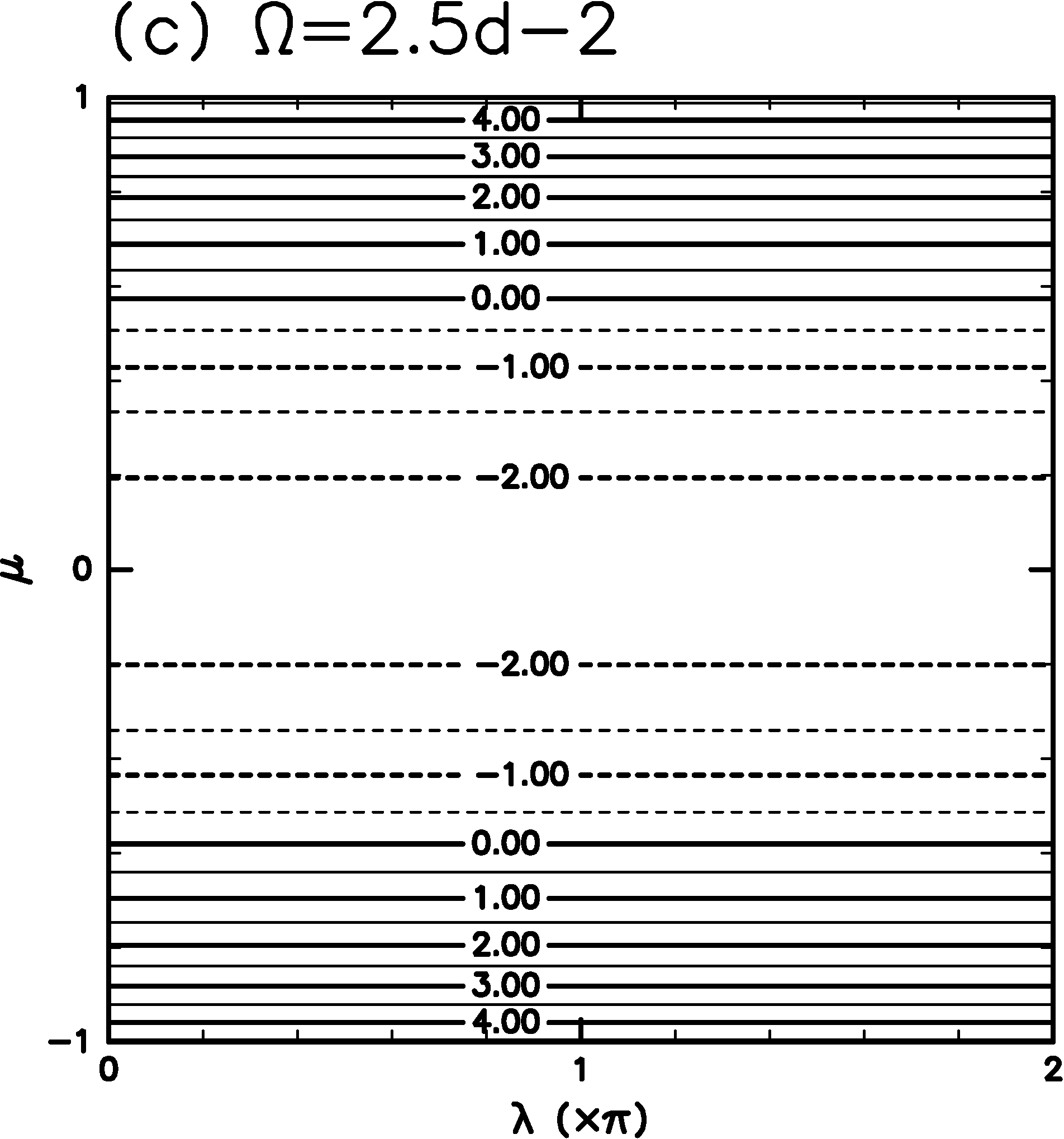}       
            \centering
        \end{minipage}
        &
        \begin{minipage}[t]{0.40\hsize}
            \includegraphics[scale=0.3]{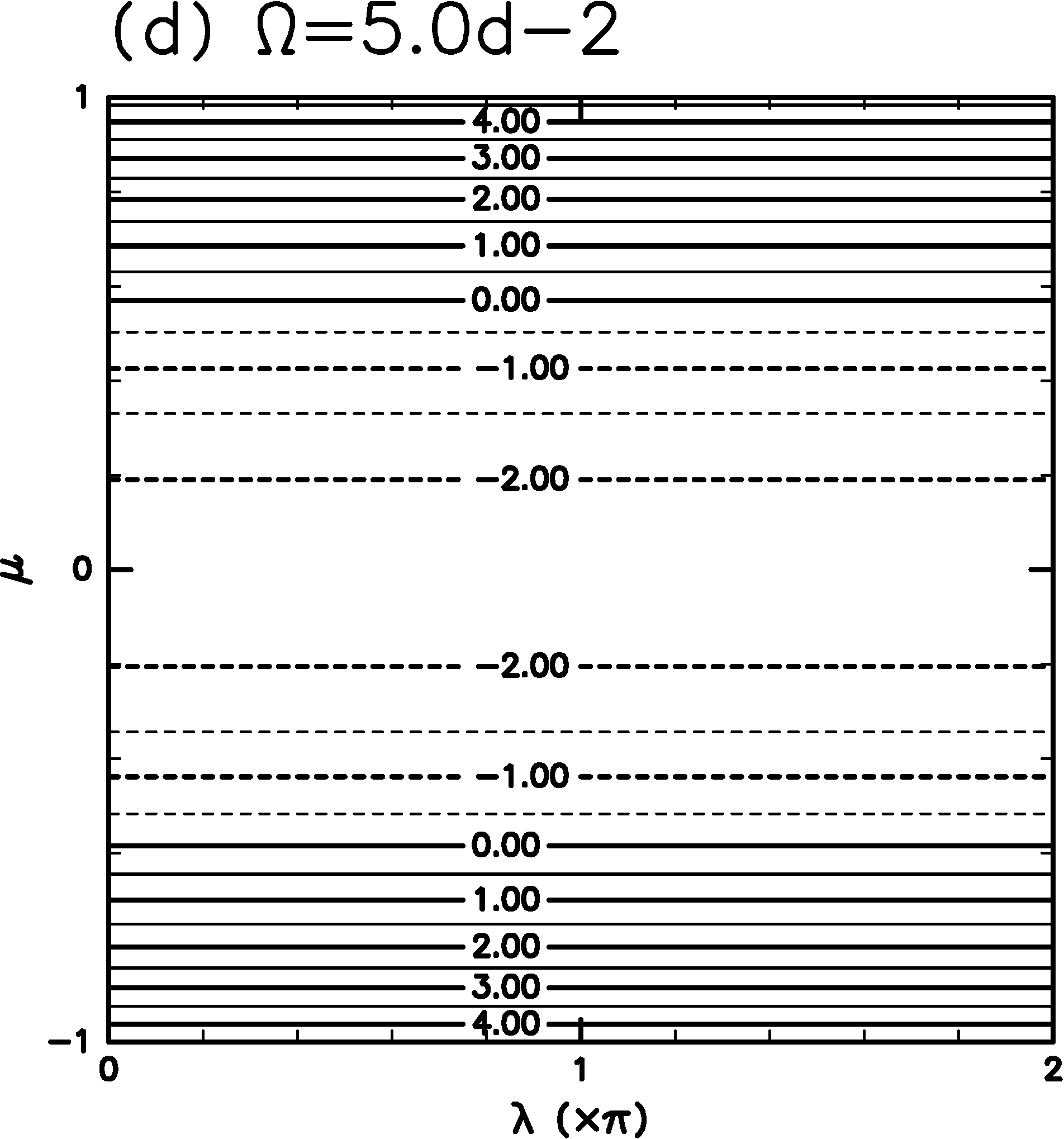}
            \centering
        \end{minipage}
        \\
        \begin{minipage}[t]{0.40\hsize}
            \includegraphics[scale=0.3]{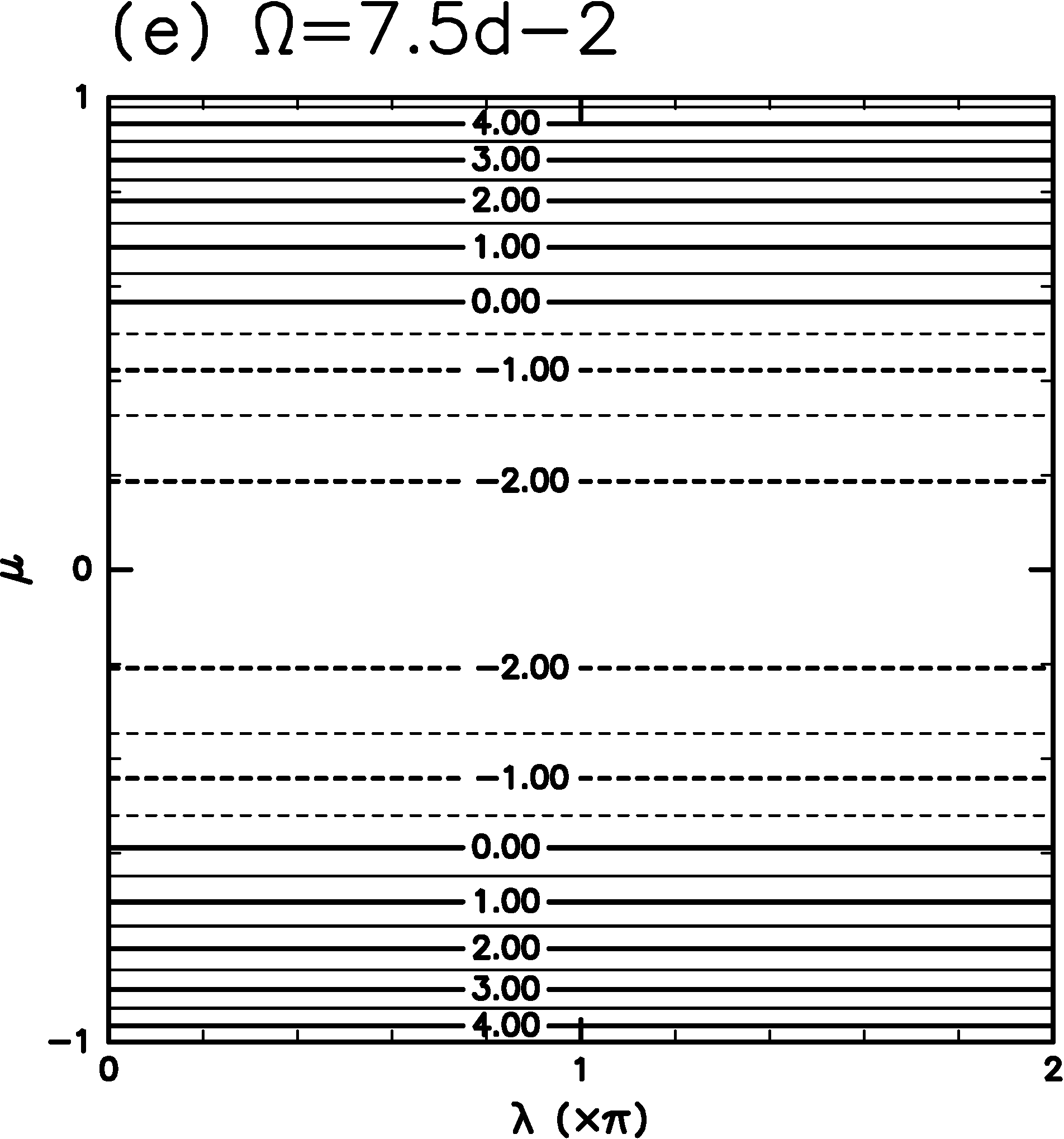}
            \centering
        \end{minipage}

    \end{tabular}
    \caption{The same as figure \ref{equilibria_a-d} but for zonal statistical equilibria only for \(\Omega=0, 1.0\times 10^{-2}, 2.5\times 10^{-2}, 5.0\times 10^{-2}\), and \(7.5\times 10^{-2}\).}
    \label{zonal_equilibria}
\end{figure}

\section{Discussion and Conclusions}\label{section_discussion}
In the present study, we computed the statistical equilibria defined by the original MRS theory (\citealp{miller1990statistical}, \citealp{robert1991statistical}) for the initial vorticity field of \cite{dritschel2015late}, as well as that with a planetary vorticity term, using the calculation method of \cite{ryono2022new}. To compute statistical equilibria for an initial vorticity field given in continuous form, we must approximate the initial vorticity field by a finite set of vorticity patches. This was not a problem when computing the statistical equilibrium in \cite{ryono2022new}, even though their initial vorticity field was continuous. This is because their initial vorticity field was zonally symmetric and the approximation was straightforward. In the present study, we proposed a new method of approximation which can be applied to non-zonal initial vorticity fields. In the new method, for an integer \(K\), which is the number of vorticity patches, the areas \(S_1, \cdots, S_K\) and the vorticity values \(Q_1, \cdots, Q_K\) are determined so that the point \(Z_{\rm ini}\) corresponding to the grid values of the initial vorticity \(q_{ij}\) belongs to the interior of the feasible region. In the computation of \cite{ryono2022new}, a geometric technique with a complex procedure was needed to generate initial search points in the interior of the feasible region. However, by using the new method introduced in this study, we can use the point \(Z_{\rm ini}\) immediately as the initial search point, which greatly simplifies the procedure of computing statistical equilibria. 

For the case of the initial vorticity field of \cite{dritschel2015late} (\(\Omega=0\)), we obtained a quadrupole state with two positive and two negative vortices as the statistical equilibrium. The quadrupole structure was in agreement with the resulting vorticity field of time integration, but the vorticity value in the core of each vortex was much smaller than that of the time integration result. There was also a large difference in the spatial distribution of the vorticity. Each vortex in the statistical equilibrium state had a broad vorticity distribution, while concentrated vortices appeared as a result of the time integration. As already pointed out by \cite{dritschel2015late} and \cite{modin2020casimir}, the resulting vorticity field of time integration was neither a stationary state nor a solid-rotating state, unlike the prediction of the MRS theory on the sphere.
\cite{qi2014hyperviscosity} pointed out the possibility that the conservation of the Casimir invariants of higher order could explain the reason for the appearance of the concentrated coherent vortices in the time evolution. They deduced this conjecture by comparing the MRS--2 theory, in which the Casimir invariants are considered up to the second order, and the perturbative MRS--4 theory, in which the Casimir invariants are considered up to the fourth order (but partially). {That is, in the MRS--2 theory, the functional relationship between \(\overline{q}\) and \(\overline{\psi}+\Omega_1 \mu\) becomes linear, whereas in the perturbative MRS--4 theory a cubic functional relation, which can be treated as a polynomial approximation of sinh-type relation, arises. The coefficient of the cubic term is determined by the Lagrange multiplier with respect to the Casimir invariant of fourth-order. Therefore, to obtain a statistical equilibrium with sinh-type \(\overline{q}\)--\(\overline{\psi}\) relation, which represents pairs of concentrated vortices, the Casimir invariants of at least fourth order should be taken into account.} However, the present study showed that concentrated vortices, as they appear in the time evolution, cannot be obtained as a statistical equilibrium, even if the conservations of all Casimir invariants are considered. The difficulty of predicting the emergence of concentrated coherent vortices from the random initial vorticity field by the MRS theory, even if we take the conservations of the higher order Casimir invariants into account, is also evident by looking at the \(\overline{q}\)--\(\overline{\psi}\) relation (figure \ref{qpsirel_a-d}). The obtained \(\overline{q}\)--\(\overline{\psi}\) relation was almost linear and we could not see a bending of the \(\overline{q}\)--\(\overline{\psi}\) curve which was predicted by the perturbative MRS--4 theory. This means that the conservation of fourth-order or higher-order Casimir invariants has virtually no effect on the determination of statistical equilibrium solutions for random initial vorticity fields such as those considered here. {Therefore, the final states of \cite{dritschel2015late} and \cite{modin2020casimir} are not even perturbed states of some statistical equilbria, even if they show quasi-periodic states and two-branched sinh-type relations between \(q\) and \(\psi\).} Note, however, that the almost linear \(\overline{q}\)--\(\overline{\psi}\) relations obtained in the present study, do not necessarily lead to the conclusion that higher-order Casimir invariants are generally unimportant for the computation of statistical equilibria. For example, \cite{ryono2022new} obtained statistical equilibria with nonlinear \(\overline{q}\)--\(\overline{\psi}\) relations, implying the importance of higher-order Casimir invariants. Furthermore, without considering higher-order Casimir invariants, the statistical equilibrium solution for the \(\Omega=0\) case in the present study (figure \ref{equilibria_a-d}(a)), a quadrupole state with two positive vortices and relatively weaker negative vortices, cannot be uniquely determined. Here, ``uniquely'' means excluding the indeterminacy due to the \(O(3)\) symmetry. In simplified theories such as MRS--2 and perturbative MRS--4, which only consider low-order Casimir invariants, there is no necessity for such asymmetric equilibrium solutions to be chosen. {We again note that the obtained statistical equilibria are not proved to be a global maximizer of the mixing entropy, however, the obtained ones are sufficient evidence that the vorticity fields can be much more mixed than the results of time integrations, as we discuss in the following.}

For the random initial vorticity field considered in the present study, there was a significant difference between the prediction by the MRS theory and the resulting vorticity field in the time evolution, in contrast to the case of \cite{ryono2022new}, but we believe that some insight into pattern formation from two-dimensional turbulence can be gained by considering the causes of this difference. The mathematical principle of \cite{robert1991statistical} states that if the time evolution of two-dimensional turbulence is ergodic, the vorticity field will evolve to the neighborhood of the statistical equilibrium. Although the initial vorticity field of \cite{dritschel2015late} differed from the type of checkerboard pattern introduced by \cite{segre1998late} and was generic in the sense that it was randomly generated \citep{modin2020casimir}, the time evolution from the initial vorticity field did not fully approach statistical equilibrium. This means the time evolution could not reach ergodicity. We believe that the structure of the coherent vortices appearing in the time evolution prevented the vorticity field from reaching a fully mixed state. \cite{dritschel2015late} pointed out that a stepped vorticity distribution (staircase structure) can be observed in the cross-section through the center of the coherent vortices, where the sharp vorticity gradient and the nearly constant vorticity are aligned alternately in the direction of the radius. The almost constant regions of vorticity are produced by mixing of the vorticity field, but further mixing is prevented by the sharp vorticity gradient, which acts as a ``mixing barrier''. We therefore consider that the concentrated coherent vortices appearing in the time evolution survive without relaxation to the statistical equilibrium by the effect of the mixing barrier. In such situations, the point vortex approximation \citep{dritschel2015late} becomes effective because the few coherent vortices that survive after a long time evolution behave like point vortices, and whether the motion of these point vortices is chaotic or integrable determines the subsequent mixing of the vorticity field \citep{modin2020casimir}. 
The fact that the point vortex approximation works well in the time evolution results may be related to the difference that for non-zero but small values of \(\Omega\), two positive vortices were fixed at both poles in the quadrupole state of the statistical equilibrium solution, but not in the time evolution results. This is because when \(\Omega\) is small, if the point vortex approximation holds, the motion of the four point vortices is close to an integrable system and there is no necessity for the vortices to have fixed positions, as they will each continue to move along their quasi-periodic trajectories without triggering chaotic mixing.  

The calculation of statistical equilibrium for initial vorticity fields with non-zero planetary vorticity in this study also revealed {the change of} the type of the statistical equilibrium from a quadrupole state to a fully zonal state via a wavenumber one state as the value of \(\Omega\) increases. There have been several previous studies on {such changes} of statistical equilibria (\citealp{sommeria1991final}, \citealp{chavanis1996classification}, \citealp{venaille2011solvable}), but for cases with two or three initial vorticity patches, or cases based on simplified MRS theories. However, to the best of the authors' knowledge, the present study is the first to investigate the {type transition} of statistical equilibria based on the original MRS theory for a generic initial vorticity field. 

To trace the branch of zonal equilibria, we have {performed the search restricted to zonal equilibria} (figure \ref{zonal_equilibria}), and it turned out that for each value of \(\Omega\) there is only a small difference between the values of \(S_{\rm mix}\) for the zonal equilibrium and the non-zonal equilibrium (table \ref{tab:entropy_values}). 
For the case of \(\Omega=0\), the non-zonal equilibrium has the value of \(S_{\rm mix}=2.949253274176\) and the zonal equilibrium has the value of \(S_{\rm mix}=2.949248609918\). If we compute the value \(S_{\rm mix}\) for the grid value \((q_{ij})\) of the initial vorticity field \(q_{\rm ini}\), we get \(S_{\rm mix}=0.504779627054\), which is non-zero because, as described in section \ref{section_method}, the initial vorticity field is approximated by a smaller number of patch levels than the number of grids. Although the two equilibria are very different in appearance (figure \ref{equilibria_a-d}(a) and \ref{zonal_equilibria}(a)), the difference between the values of \(S_{\rm mix}\) for them is very small compared to the difference between the values of \(S_{\rm mix}\) for the initial vorticity field and each of the two equilibria. {However, the small difference in mixing entropy can still lead to a large difference in the corresponding microscopic states. In the MRS theory, the large deviation principle stands (\citealp{robert1991maximum}, \citealp{michel1994large}) with an entropy functional \(\mathcal{K}=S_{\rm mix}+{\rm (const.)}\). To set up the discussion with the large deviation theory, we construct microscopic vorticity field \(q\) by partitioning the domain into \(n\) subdomains of equal area and putting a constant vorticity patch independently on each subdomain, with the probability distribution of the initial vorticity field (which is, in our setting, \(\sum S_{k}\delta_{Q_k}/(4\pi)\)). We then consider the corresponding macroscopic state given by the Young measure naturally obtained from \(q\). Under the above preparation, we can relate the difference in the mixing entropy to the difference in the number of microscopic states. That is, for any \(\alpha \in (0,1)\), any neighborhood \(W\) of a macroscopic state with \(S_{\rm mix}=S_0\) and any closed subset \(C\) in which the entropy fulfills \(S_{\rm mix}\leq S_{0}-\Delta S\), the probability of getting a macroscopic state in \(W\) is much larger than that of getting one in a neighborhood (which depends on \(\alpha\)) of \(C\), and the ratio is larger than \( \exp (n\alpha\Delta S) \) for sufficiently large \(n\). For details, see  \ref{sec_appendix_ldp}. Consider the non-zonal equilibrium and the zonal one for \(\Omega=0.0\) case, the difference in entropy value between two state{s} is about \(\Delta S = 5\times 10^{-6}\). {For example, if} we set \(n=2048\times 1024\approx 2\times 10^6\), which equals to the number of numerical grids we used for the time integration, and \(\alpha=0.9\), then \(\exp (n\alpha\Delta S) \approx \exp (9) \approx 8\times 10^3\). {Note that, the value of \(\alpha\in (0,1)\) in Proposition \ref{prop_comparestates} can be taken arbitrarily close to one, but the neighborhood \(V\) which is guaranteed to exist by the proposition depends on \(\alpha\). Of course, in our time integration, the high wavenumber components of the vorticity field are affected by the artificial viscosity. However, if we perform a time integration with many more grids and with much smaller viscosity or if we use the Casimir-preserving scheme of \cite{modin2020casimir} with higher resolution, we will be able to observe a turbulence of smaller scale. Thus, no matter how large we choose the value of \(n\), \(n\) can never be excessive to express a fully evolved turbulence as a microscopic field by the procedure of \ref{sec_app_ldp_micro}.} Therefore, the probability of getting a macroscopic state in a fixed neighborhood of a macroscopic state with \(S_{\rm mix}=S_0\) is overwhelmingly larger than the probability of getting a state such that \(S_{\rm mix} \leq S_0-\Delta S\).} 

{The above discussion is about comparing the mixing entropy values of two points, a statistical equilibrium and one of the other macroscopic states. One additional remark on the values of the mixing entropy is that, in the process of computing the optimization to find the non-zonal equilibrium before convergence, the appearance of the macroscopic vorticity field changed as the iterations progressed, but the increase in the value of \(S_{\rm mix}\) per iteration was very small (not shown). This implies that there are many macroscopic states which have as high entropy value as the entropy maximizing state has. Nevertheless, the priority of the statistical equilibrium is again robust, since the concentration theorem guarantees the following: for any neighborhood of the statistical equilibria, the probability of getting a macroscopic state outside the neighborhood is less than \(\exp(-\alpha' n)\), where \(n\) is the number of domain partition, which is sufficiently large, and \(\alpha'>0\) is a constant depending on the choice of the neighborhood.} The discussion{s in the previous and present paragraphs} assume the ergodicity of the system. Therefore, if the ergodicity holds, the statistical equilibrium is overwhelmingly more likely to emerge than the other states. Conversely, we can conclude that the time evolutions in the present manuscript are not ergodic.

As we have seen above, the present study has explicitly calculated the statistical equilibria for generic initial vorticity fields on a sphere, and shown how the statistical equilibria and the time integration results differ.   
The discrepancy found in the present study may again cast doubt on the predictive ability of the MRS theory, as did \cite{dritschel2015late} and \cite{modin2020casimir}. However, we believe that the results obtained in the present study provide a new perspective on the pattern formation in two-dimensional turbulence, emphasizing the formation of a mixing barrier that prevents the relaxation of the vorticity fields to the entropy maximizing state. It should also be noted that the MRS theory is not completely ineffective, recalling that the initial vorticity field treated by \cite{ryono2022new}, which was bartropically unstable, evolved to the final state in good agreement with the statistical equilibrium. In this case, the time evolution of the vorticity field may have been close to ergodicity due to the strong global mixing caused by the barotropic instability. It will be our future task to investigate in detail the differences between initial vorticity fields for which the predictions of the MRS theory are in good agreement with the time integration results, and those for which they are not. Such a study will provide new insights into the phenomena of mixing and pattern formation of vorticity fields in two-dimensional turbulence.

\appendix

\section{Proof of the inequality \eqref{casimir_error}}\label{sec_appendix_casimir_error}
For grid values \((q_{ij})\) of the initial vorticity field \(q_{\rm ini}\) and a Lipschitz function \(f:[q_{\rm min}, q_{\rm max}] \longrightarrow \mathbf{R}\) (where \(q_{\rm min}={\rm min}\, q_{ij}, q_{\rm max} = {\rm max}\, q_{ij}\)), the Casimir invariant for \(f\), calculated using the grid values \((q_{ij})\) is written as:
\begin{align*}
    C_{f,\rm grid} = \frac{1}{2}\sum_{i,j} w_j f(q_{ij}).
\end{align*}
Note that, this value already has an error of the Casimir invariant \(C_f\) from the true value, which is defined by the integral, because of the discretization. However, we regard \(C_{f, \rm grid}\) as the best available estimate of \(C_f\) value that we can know. On the other hand, the value of Casimir invariant for \(f\) calculated by using the determined set of vorticity patches is written as:
\begin{align*}
    C_{f,\rm patch} = \frac{1}{4\pi} \sum_{k=1}^K S_k f(Q_k) = \frac{1}{2}\sum_{i,j,k} w_j f(Q_k) \tilde{r}_{ijk}.
\end{align*}

For each \((i,j)\), let \(k_{ij}\) be the unique integer \(k\) which fulfills \(Q_k < q_{ij} \leq Q_{k+1}\). By using the definition of \(\tilde{r}_{ijk}\), we obtain:
\begin{align*}
    C_{f,\rm grid} &= \frac{1}{2}\sum_{i,j}\sum_{k=1}^K w_j f(q_{ij}) \tilde{r}_{ijk}\\
    &= \frac{1}{2}\sum_{i,j} \left\{(K-2)\varepsilon f(q_{ij})+w_j f(q_{ij}) (\tilde{r}_{ij,k_{ij}}+\tilde{r}_{ij,k_{ij}+1})\right\}, \\
    C_{f,\rm patch} &= \frac{1}{2} \sum_{i,j} \left[\sum_{k\neq k_{ij},k_{ij}+1} \varepsilon f(Q_k) + w_j \left\{f(Q_{k_{ij}}) \tilde{r}_{ij,k_{ij}}+f(Q_{k_{ij}+1})\tilde{r}_{ij,k_{ij}+1}\right\}\right].
\end{align*}
Therefore, we get
\begin{align*}
    |C_{f,\rm grid}-C_{f,\rm patch}|&\leq  \frac{1}{2} \sum_{i,j} \sum_{k\neq k_{ij},k_{ij}+1} \varepsilon |f(q_{i,j})-f(Q_{k_{ij}})| \\
    &\qquad +\frac{1}{2}\sum_{i,j}\Big\{ w_j |f(q_{ij})-f(Q_{k_{ij}})| \tilde{r}_{ij,k_{ij}}\\
    &\qquad\qquad+ w_j |f(q_{ij})-f(Q_{k_{ij}+1})| \tilde{r}_{ij,k_{ij}+1}\Big\}\\
    &\leq  \frac{1}{2} \sum_{i,j} \sum_{k\neq k_{ij},k_{ij}+1} \varepsilon\cdot 2M \\
    &\qquad + \frac{1}{2}\sum_{i,j} \Big\{w_j c\,|q_{ij}-Q_{k_{ij}}| \tilde{r}_{ij,k_{ij}} \\
    &\qquad\qquad+ w_j c\,|q_{ij} - Q_{k_{ij}+1} | \tilde{r}_{ij,k_{ij}+1}\Big\}\\
    &\leq  IJ(K-2)M\varepsilon  + \frac{1}{2}c(\Delta Q +\alpha) \sum_{i,j} w_j (\tilde{r}_{ij,k_{ij}}+  \tilde{r}_{ij,k_{ij}+1} )\\
    &\leq  IJ(K-2)M\varepsilon + \frac{1}{2}c (\Delta Q +\alpha) \sum_{i,j} w_j \\
    &= IJ(K-2) M\varepsilon + c(\Delta Q+\alpha).
\end{align*}
Note that, we used the following inequality:
\begin{align*}
|f(Q_k)-f(q_{ij})|\leq c\, |Q_k - q_{ij}| \leq c(\Delta Q+\alpha)\qquad (k=k_{ij}, k_{ij}+1).
\end{align*}
The inequality \eqref{casimir_error} has been proved.  

\section{ Determination of \(\Omega_1\) value for statistical equilibria with zonal symmetry}\label{sec_appendix_omega1}

{From a general discussion on the MRS theory, there is a parameter \(\Omega_1\) and a functional relation between the macroscopic vorticity \(\overline{q}\) in equilibrium and the corresponding stream function \(\overline{\psi}\), that is, \(\overline{q} = f(\overline{\psi }+ \Omega_1 \mu)\). We here describe a method for determining \(\Omega_1\) from a finite number of grid data of \(\overline{q}\) and \(\overline{\psi}\) is given, especially when \(\overline{q}\) is zonally symmetric.}

{Suppose \((\overline{q}_j, \overline{\psi}_j)=(\overline{q}(0, \mu_j),\overline{\psi}(0,\mu_j))_{j=1,\cdots,J}\), grid values of \(\overline{q}\) and \(\overline{\psi}\) are given. We assume \(\overline{q}_j \neq \overline{q}_k\) if \(k\neq j\). First, recall that \(f\) is a smooth function and a monotone function, and therefore \(f\) has a single-valued inverse function \(g\). However, if we take a wrong value for \(\Omega_1\), the plot of \(\overline{q}\) versus \(\overline{\psi}+ \Omega_1 \mu\) looks like a graph of a multi-valued function. To exclude such cases, we consider a function of \(\Omega_1\) defined as follows:}

\begin{align*}
    F(\Omega_1) = \sum_{k=1}^{J-1} \frac{[\overline{\psi}_{j(k+1)}+\Omega_1 \mu_{j(k+1)}-(\overline{\psi}_{j(k)}+\Omega_1 \mu_{j(k)})]^2}{\overline{q}_{j(k+1)}-\overline{q}_{j(k)}},
\end{align*}

\noindent
{where the indices \(j(k)\) are defined so that \(q_{j(1)}<q_{j(2)}<\cdots< q_{j(J)}\). If the plot is single-valued, the value of \(F(\Omega_1)\) is an approximation of the value of}

\begin{align*}
    \int_{\overline{q}_{\rm min}}^{\overline{q}_{\rm max}} \{g'(q)\}^2 dq.
\end{align*}

\noindent
{Otherwise, the value of \(F(\Omega_1)\) becomes very large because of the branched plot of \(\overline{\psi}+\Omega_1 \mu\) versus \(\overline{q}\). Therefore, we can determine the value of \(\Omega_1\) by minimizing \(F(\Omega_1)\).} {Note that, if there exist two point \((\lambda_{\rm a}, \mu_{\rm a})\) and \((\lambda_{\rm b}, \mu_{\rm b})\) such that \(\mu_{\rm a}\neq \mu_{\rm b}\) and \(\overline{q}(\lambda_{\rm a},\mu_{\rm a}) = \overline{q}(\lambda_{\rm b}, \mu_{\rm b})\), then \(\Omega_1\) such that \(\overline{q}\) becomes a monotone function of \(\overline{\psi}+\Omega_1\mu\) is unique. Indeed, if such a functional relationship exists, \(\overline{\psi}+\Omega_1 \mu\) is determined by the value of \(\overline{q}\). Therefore, by using \(g(\overline{q}(\lambda_{\rm a}, \mu_{\rm a}))= \overline{\psi}_{\rm a} + \Omega_1\mu_{\rm a} = \overline{\psi}_{\rm b} + \Omega_1\mu_{\rm b} = g(\overline{q}(\lambda_{\rm b}, \mu_{\rm b}))\), we obtain \((\mu_{\rm a}-\mu_{\rm b})\Omega_1=-(\overline{\psi}_{\rm a}-\overline{\psi}_{\rm b})\).}

\section{On the large deviation principle in the MRS theory}\label{sec_appendix_ldp}

In the following, we consider an initial vorticity field \(q_0 \in L^\infty (S)\) and let \(M = ||q_0||_{\infty} = \sup_{x\in S} |q_0(x)|\).

Let \(q \in L^\infty (S)\), then it defines a linear functional \(\pi_q\) on \(C(S)\) by
\begin{align*}
    \langle \pi_q, f \rangle = \frac{1}{4\pi}\int_S f(q(x)) dx.
\end{align*}
Note that, \(\pi_q\) is non-negative (if \(f\geq 0\) then \(\langle \pi_q, f\rangle \geq 0\)) and continuous on \(C(S)\), and therefore defines a probability measure \(\mu_q\) such that
\begin{align*}
    \langle \pi_q, f\rangle = \int_\mathbf{R} f(y) d\mu_q (y),
\end{align*}
by virtue of Riesz's representation theorem. We denote the probability measure as \(\pi_q\) rather than \(\mu_q\) by abuse of notation.

\subsection{Macroscopic fields}
In the context of the large deviation theory, the macroscopic vorticity field is a Young measure on \(S\times [-M, M]\). Here, a Young measure \(\nu\) is a measure on \(S\times [-M,M]\) such that, for each \(x\in S\) a probability measure \(\nu_x\) on \([-M,M]\) corresponds, and for a function \(f\in C(S\times [-M,M])\) 
\begin{align*}
    \int_{S\times[-M,M]} f(x,y) d\nu(x,y) = \int_S \langle \nu_x, f(x,\cdot) \rangle dx = \int_S \int_{-M}^{M} f(x,y) d\nu_x (y) dx
\end{align*}
stands. Note that in the following the letter \(\nu\) is used for Young measures, rather than for viscosity coefficients as in Section \ref{subsection_timeevo}. {We also note that we use the term ``macroscopic vorticity field" for Young measures here, rather than for the scalar field \(\overline{q}\) as in the main part of the present manuscript, to emphasize the theoretical correspondence with the microscopic vorticity field, which will be defined in the next subsection.} In the following, the set of Young measures \(Y\) is endowed with the narrow topology (which is defined by weak convergence about continuous bounded test functions on \(S\times[-M, M]\)). Note that we are considering a compact flow domain \(S\), and therefore the weak star topology (defined by weak convergence about compactly supported continuous test functions) and the narrow topology coincide. Furthermore, since \(C(S\times[-M,M])\) is a separable Banach space, \(Y\) is metrizable. By using Banach-Alaoglu's theorem, \(Y\) is also compact.

Macroscopic vorticity field defined above is not given as a function, but as a field of probability measures. Indeed, it gives a mathematical expression of coarse graining process of a scalar field on \(S\). For a given \(q\in L^\infty (S), ||q||_\infty \leq M\), we associate a Young measure \(\delta_q \in Y\) defined by \((\delta_q)_x = \delta_{q(x)}\) with it. Here the right hand side means the Dirac mass at \(q(x)\).

In the following, we assume that the initial field \(q_0\) consists of a finite number of constant vorticity patches, as we considered so far, then \(\pi_{q_0}\) can be expressed as:
\begin{align*}
    \pi_{q_0} = \frac{1}{4\pi}\sum_{k=1}^K S_k \delta_{Q_k}.
\end{align*}
If a macroscopic field \(\nu\) is absolute continuous with respect to \(dx\otimes\pi_{q_0}\), then for a.e. \(x\in S\), \(\nu_x\) is absolute continuous with respect to {\(\pi_{q_0}\)}. The Radon-Nikodym derivative \(d\nu_x/d\pi_{q_0}\) is defined at a.e. \(x\in S\) and it can be expressed as:
\begin{align}
    \frac{d\nu_x}{d\pi_{q_0}} (y) = \left\{
    \begin{array}{cc}
     \frac{4\pi r_k(x)}{S_k}    &  ({\rm for}\, y=Q_k)\\
     0    &  ({\rm else})
    \end{array}
    \right. \label{dnu_dpi}
\end{align}
where \(r_k \in L^1 (S)\). By definition, \(r_k \geq 0, \sum r_k = 1\) a.e. and \(\int_S r_k(x) dx = S_k\). Note that they correspond to the macroscopic states that we defined in Section \ref{section_theory}.  

\subsection{Microscopic fields}\label{sec_app_ldp_micro}
We define the microscopic vorticity field as a random variable that takes its value in \(L^\infty (S)\), defined by the following procedure.

First, we partition the sphere into \(n\) subdomains with equal areas \(D_1, \cdots, D_n\), so that the maximal diameter of the subdomains tends to zero when \(n\to \infty\). Let \(y_1,\cdots, y_n\in [-M, M]\) be \(n\) random variables with independent identical distributions \(\pi_{q_0}\). For \(y_1,\cdots, y_n\), define a simple function \(\mathcal{Q}(\cdot\,; y_1,\cdots,y_n) : S\longrightarrow [-M,M]\) by \(\mathcal{Q}(x; y_1,\cdots, y_n) = y_k\) for \(x\in D_k\). We call each \(\mathcal{Q}(\cdot\,; y_1,\cdots,y_n)\) a microscopic vorticity field in the sense of large deviation theory. 

As explained in the above subsection, we can correspond a macroscopic vorticity field \(\delta_{\mathcal{Q}(\cdot\,; y_1,\cdots,y_n)}\) to a microscopic vorticity field \(\mathcal{Q}(\cdot\,; y_1,\cdots,y_n)\). We denote the probability of  \(\mathcal{Q}(\cdot\,; y_1,\cdots,y_n)\) being in a Borel subset \(A\subset Y\) of macroscopic fields by \({\rm Prob}(\mathcal{Q}\in A)\).

\subsection{Large deviation principle}
\cite{michel1994large} proved that the large deviation principle stands in the MRS theory.
\begin{thm}
For any Borel subset \(A\subset Y\), the following inequality holds:
\begin{align*}
\sup_{\nu \in A^\circ} \mathcal{K}(\nu) &\leq \liminf_{n\to \infty} \frac{1}{n}\log {\rm Prob}(\mathcal{Q}\in A)\\ &
\leq \limsup_{n\to \infty} \frac{1}{n}\log {\rm Prob}(\mathcal{Q}\in A) \leq \sup_{\nu\in \overline{A}} \mathcal{K}(\nu),
\end{align*}
where \(\mathcal{K}\) is the entropy functional defined as
\begin{align*}
    \mathcal{K}(\nu) = -\frac{1}{4\pi} \int_S \sum_{k=1}^K r_k(x)\log r_k(x) dx + \frac{1}{4\pi} \sum_{k=1}^K S_k \log \frac{S_k}{4\pi}
\end{align*}
if \(\nu\in Y\) is absolute continuous with respect to \(dx\otimes \pi_{q_0}\), and \(\mathcal{K}(\nu) = -\infty\) otherwise. Here, we denote the interior of \(A\) by \(A^\circ\) and the closure of \(A\) by \(\overline{A}\) .
\end{thm}
Note that, in fact \(\mathcal{K}\) is upper-semicontinuous, concave function on \(Y\), and \(\mathcal{K}\leq 0\). When \(\nu\) is absolute continuous with respect to \(dx\otimes \pi_{q_0}\), the entropy functional is the mixing entropy \(S_{\rm mix}\) plus an additional constant. 
  
\subsection{Concentration of microscopic fields}
The concentration theorem, which states that a great majority of microscopic states concentrate near the set of the maximizers of the mixing entropy, follows from the above theorem \citep{michel1994large, robert1991statistical}. We describe it here in a more elementary form than its original form. In the derivation, we obtain Proposition \ref{prop_comparestates}, which relates the difference in the entropy between two macroscopic states to the ratio of the numbers of the corresponding microscopic states to them. It is useful for quantitative discussion on computed statistical equilibria.

The conserved quantities of the Euler equation can be considered as functions on \(Y\): The energy \(E(\nu)\) is defined as:
\begin{align*}
    E(\nu) = -\frac{1}{4\pi}\int_S \overline{q}(x) (-\triangle)^{-1} \overline{q}(x) dx, 
\end{align*}
where 
\begin{align}
    \overline{q}(x) = \int_{-M}^M y d\nu_x (y) - \frac{1}{4\pi} \int_S \int_{-M}^M y d\nu_x (y) dx\label{macroscopic_minus_mean}
\end{align}
is the macroscopic vorticity minus its mean value over \(S\). The angular momentum is \(M_i (\nu) = (4\pi)^{-1}\int_S L_i(x) \overline{q}(x) dx\, (i=1,2,3)\), where \(L_1(x)=\mu, L_2(x)=\sqrt{1-\mu^2}\cos\lambda\), and \(L_3(x)= \sqrt{1-\mu^2}\sin\lambda\). The energy and the three components of angular momentum are all continuous on \(Y\). To show the continuity of the energy, the spherical harmonics expansion is useful. We define
\begin{align*}
    A_{m,n}(\nu)=\frac{1}{4\pi}\int_S \overline{q}(x) Y_{m,n}(x)^* dx\qquad (m=-n, -n+1, \cdots, n;\,n=0,1,\cdots)
\end{align*}
for \(\nu\in Y\), where \(\overline{q}\in L^2(S)\) represents a function defined by \eqref{macroscopic_minus_mean}.
We consider \(\nu_0\in Y\) and the corresponding \(\overline{q_0}\), and let \(\varepsilon\) be an arbitrary positive number. Since \(-M\leq \overline{q_0}(x) = \int_{-M}^M y (\nu_0)_x(dy)\leq M\), we have \(\sum_{m,n}|A_{m,n}(\nu_0)|^2= \|\overline{q_0}\|^2_{L^2}\leq M^2\), and we can take a large integer \(N>0\) such that \(M^2/(N(N+1))<\varepsilon\) and 
\begin{align*}
    \sum_{n>N}\sum_{m=-n}^n |A_{m,n}(\nu_0)|^2 < \varepsilon.
\end{align*}
Let \(U_{N,\varepsilon}(\nu_0)\) be an open neighborhood of \(\nu_0\in Y\) in the weak star topology:
\begin{align*}
    U_{N,\varepsilon}(\nu_0) = \left\{\nu\in Y \mid |A_{m,n}(\nu)-A_{m,n}(\nu_0)|<\frac{\varepsilon}{(N+1)^2},\,(|m|\leq n, n=1,\cdots,N) \right\}.
\end{align*}
Let \(\nu \in U_{N,\varepsilon}(\nu_0)\). By using \(|A_{m,n}(\nu)|\leq M\),
\begin{align*}
    \sum_{n=1}^N \sum_{m=-n}^n \left|\frac{|A_{m,n}(\nu)|^2-|A_{m,n}(\nu_0)|^2}{n(n+1)}\right| < M\varepsilon^2.
\end{align*}
By using \(\sum_{n>N}|A_{m,n}(\nu)|^2\leq \|q\|^2_{L^2}\leq M^2\),
\begin{align*}
    \sum_{n>N}\sum_{m=-n}^n \frac{|A_{m,n}(\nu)|^2}{n(n+1)} < \frac{M^2}{N(N+1)} < \varepsilon.
\end{align*}
Therefore,
\begin{align*}
    |E(\nu)-E(\nu_0)|=\left|\frac{1}{2}\sum_{n=1}^\infty \sum_{m=-n}^n \frac{|A_{m,n}(\nu)|^2}{n(n+1)}-\frac{1}{2}\sum_{n=1}^\infty \sum_{m=-n}^n \frac{|A_{m,n}(\nu_0)|^2}{n(n+1)}\right|\leq \frac{1}{2}M\varepsilon^2+\varepsilon.
\end{align*}
This completes the proof.

The conservation of Casimir invariants corresponds to the equation of Young measure: \(\frac{1}{4\pi}\int_S \nu_x dx = \pi_{q_0}\). We define a closed subset of \(Y\) by
\begin{align*}
    A &= \Bigg\{ \nu\in Y \bigg| |E(\nu)-E_0|\leq\varepsilon_E, \\
    &\qquad \qquad |M_i(\nu)-M_i^{\rm ini}|\leq\varepsilon_{M_i}\,(i=1,2,3), d\left(\frac{1}{4\pi}\int_S \nu_x dx, \pi_{q_0}\right)\leq\varepsilon_C \Bigg\},
\end{align*}
where \(\varepsilon_E, \varepsilon_{M_i}\,(i=1,2,3)\), and \(\varepsilon_C\) are small positive constants, and \(d(\cdot,\cdot)\) in the right-hand side is a distance of the space of probability measures on \([-M,M]\) with weak star topology (or, equivalently narrow topology). {The following proposition allows us to estimate how relatively rare a macroscopic state is compared to the other state by using the difference of the entropy value \(\Delta S\).}
\begin{prop}\label{prop_comparestates}
    Let \(\alpha \in (0,1)\) {and \(\Delta S > 0\)}. Let \(\nu_{0}\in Y\) be an element of 
    \begin{align*}
        A_0=\left\{\nu\in Y \mid E(\nu)=E_0, M_i(\nu)=M_i^{\rm ini}\,(i=1,2,3), \frac{1}{4\pi}\int_S \nu_x dx = \pi_{q_0}\right\},
    \end{align*}
    and assume that \(\mathcal{K}_0 = \mathcal{K}(\nu_0) > -\infty\){,} and let \(C\) be a closed subset included in
    \begin{align*}
        B = \{\nu\in Y\mid \mathcal{K}(\nu) \leq \mathcal{K}_0-\Delta S \}.
    \end{align*}
    Then, there exists an open neighborhood \(V\) of \(C\), such that for any open neighborhood \(W\) of \(\nu_{0}\), the following inequality holds for sufficiently large \(n\):
    \begin{align*}
        \frac{{\rm Prob}(\mathcal{Q}\in A\cap W)}{{\rm Prob}(\mathcal{Q}\in A\cap V)}\geq e^{n\alpha\Delta S}.
    \end{align*}
\end{prop}
\begin{proof}
    From the large deviation principle, one obtains 
    \begin{align*}
        {\rm Prob}(\mathcal{Q}\in A\cap W) \geq \exp\left(n\mathcal{K}_0\right)
    \end{align*}
    for sufficiently large \(n\). 
    \begin{align*}
        B' = \{\nu\in Y \mid \mathcal{K}(\nu) \geq \mathcal{K}_0-\alpha\Delta S\}
    \end{align*}
    is also closed due to the upper-semicontinuity of \(\mathcal{K}\). 
    Note that \(Y\) is a metrizable space, and therefore it is a normal topological space, i.e., any two closed sets with empty intersection can be separated by their open neighborhoods. Since now we have \(C\cap B'=\emptyset\) (compare the value of \(\mathcal{K}\)), there are open subsets \(V\) and \(V'\) such that \(C\subset V\), \(B'\subset V'\), and \(V\cap V'=\emptyset\).
    Since we have \(\overline{A\cap V} \subset \overline{V} \subset {(V')^c}\), 
    \begin{align*}
        \sup_{\nu\in\overline{A\cap V}} \mathcal{K}(\nu) \leq \sup_{\nu\in (V')^c} \mathcal{K}(\nu) \leq \mathcal{K}_0-\alpha\Delta S.
    \end{align*}
    Therefore, by the large deviation principle,
    \begin{align*}
        {\rm Prob}(\mathcal{Q}\in A\cap V) \leq \exp\left(n (\mathcal{K}_0-\alpha\Delta S)\right)
    \end{align*}
    for sufficiently large \(n\), The claim readily follows. 
\end{proof}
Note that, we can obtain the concentration theorem (in the case of our setting) from the above proposition. 
\begin{thm}\label{theorem_concentration}
    Let \(U\) be any open neighborhood of the set of all maximizers of \(S_{\rm mix}\) under the constraints of the energy, the angular momentum, and the Casimir invariants. Then, for sufficiently small \(\varepsilon_E, \varepsilon_{M_i}\,(i=1,2,3)\), and \(\varepsilon_C\), there exists a positive number \(\alpha'>0\) such that
    \begin{align*}
        \frac{{\rm Prob}(\mathcal{Q}\in A\setminus U)}{{\rm Prob}(\mathcal{Q}\in A^\circ)}\leq e^{-\alpha'n}
    \end{align*}
    holds for sufficiently large \(n\).
\end{thm}
\begin{proof}
     
     Let \(k = \sup_{\nu \in A_0\cap U^c} \mathcal{K}(\nu)\) and \(\nu_{\rm max}\) be one of the maximizers of \(\mathcal{K}\) on \(A_0\). If we take sufficiently small \(\varepsilon_E, \varepsilon_{M_i}\,(i=1,2,3)\), and \(\varepsilon_C\), then the inequality
    \begin{align*}
        k\leq \sup_{\nu\in A\cap U^c} \mathcal{K}(\nu) < \mathcal{K}(\nu_{\rm max})
    \end{align*}
    holds (use the upper semicontinuity of \(\mathcal{K}\), and compactness of \(Y\)). Applying the above proposition to \(\Delta S = \mathcal{K}(\nu_{\rm max})-\sup_{\nu\in A\cap U^c} \mathcal{K}(\nu), C=A\setminus U\), \(\nu_0=\nu_{\rm max}\), and \(W = A^\circ\), we obtain the inequality
    \begin{align*}
        \frac{{\rm Prob}(\mathcal{Q}\in A\cap V)}{{\rm Prob}(\mathcal{Q}\in A^\circ)} \leq e^{-n\alpha \Delta S}
    \end{align*}
    for a{n} open neighborhood \(V\) of \(A\setminus U\). Since \(A\setminus U \subset A\cap V\), we obtain the result (with \(\alpha'= \alpha \Delta S\)).
\end{proof}

\section*{Acknowledgment}
We are grateful to {two} anonymous reviewer{s} for helpful comments. This work was supported by JST SPRING, Grant Number JPMJSP2110, and JSPS KAKENHI, Grant Number {JP}20K04061 {and JP24KJ1340}. The GFD-DENNOU Library (\url{http://www.gfd-dennou.org/arch/dcl/}) was used to draw the figures.

\bibliographystyle{jphysicsB}
\bibliography{ref}

\end{document}